\documentclass{article}

\usepackage[main, final]{neurips_2026}

\usepackage[utf8]{inputenc} 
\usepackage[T1]{fontenc}    
\usepackage{hyperref}       
\usepackage{url}            
\usepackage{booktabs}       
\usepackage{amsfonts}       
\usepackage{nicefrac}       
\usepackage{microtype}      
\usepackage[dvipsnames]{xcolor}   
\usepackage{amsthm}
\usepackage{amsmath}
\usepackage{amssymb}
\usepackage{dsfont}
\usepackage{caption}
\usepackage{subcaption}
\usepackage{multirow}
\usepackage{float}
\usepackage[capitalise]{cleveref}
\usepackage{tikz}
\usetikzlibrary{arrows.meta, calc, shapes.misc, shapes.geometric}
\hypersetup{colorlinks,citecolor=blue,urlcolor=blue,linkcolor=blue}
\setcitestyle{round}

\newcommand{\apptocentry}[2]{%
  \par\noindent
  \makebox[2em][l]{\textbf{\ref{#1}.}}
  \textbf{\hyperref[#1]{#2}}
  \nobreak\leaders\hbox to .7em{\hss.\hss}\hfill\nobreak      
  \mbox{Page~\pageref{#1}}
  \par%
}

\newcommand{\rb}{\mathrm{rb}}
\newcommand{\db}{\mathrm{db}}

\newcommand{\biased}{\hat{\theta}^\mathrm{b}}
\newcommand{\rebiased}{\hat{\theta}^\rb}
\newcommand{\debiased}{\hat{\theta}^\mathrm{db}}
\newcommand{\unbiased}{\hat{\theta}^{\mathrm{ub}}}

\newcommand{\ML}{\text{ML}}
\newcommand{\PPI}{\text{PPI}}
\newcommand{\hthetaML}{\hat\theta_i^{\ML}}
\newcommand{\hthetaPPI}{\hat\theta_i^{\PPI}}
\newcommand{\PT}{\text{PT}}
\newcommand{\hthetaPT}{\hat\theta_i^{\PT}}

\theoremstyle{definition}
\newtheorem{proposition}{Proposition}
\newtheorem{assumption}[proposition]{Assumption}

\newtheorem{lemm}[proposition]{Lemma}
\newtheorem{theorem}[proposition]{Theorem}

\usepackage{nikos_tex}

\usepackage{natbib}

\title{Empirical Bayes Rebiasing}

\author{%
  Wanyi Ling\textsuperscript{1}\qquad
  Sida Li\textsuperscript{2}\qquad
  Junming Guan\textsuperscript{1}\qquad
  Nikolaos Ignatiadis\textsuperscript{1,2} \\[0.4em]
  \textsuperscript{1}Department of Statistics, University of Chicago \\
  \textsuperscript{2}Data Science Institute, University of Chicago \\[0.4em]
  \small\texttt{\{wanyiling, listar2000, junmingguan, ignat\}@uchicago.edu}
}

\begin{document}

\maketitle

\begin{abstract}
We study methods for simultaneous analysis of many noisy and biased estimates, each paired with an even noisier estimate of its own bias. The analyst's goal is to construct short calibrated intervals for each parameter. The standard debiasing approach, which subtracts the bias estimate from each biased estimate, inflates variance and yields long intervals. In this paper, we propose an empirical Bayes rebiasing strategy that starts from the fully debiased estimates and learns from data how much bias to reintroduce by estimating the unknown bias distribution. We provide convergence rates for the coverage of our intervals when the bias distribution is estimated using nonparametric maximum likelihood. Furthermore, we demonstrate substantial precision gains in prediction-powered inference, including pairwise LLM win-rate evaluations, as well as for inference of direct genetic effects in family-based GWAS.
\end{abstract}

\section{Introduction}
\label{sec:introduction}
One of the simplest, yet most widely used statistical constructions is the Wald interval. Given a parameter $\theta$ and an estimator \smash{$\hat{\theta}$}, we report the interval \smash{$\mathcal{I}=\hat{\theta} \pm 1.96\SEInline{\hat{\theta}}$}, where \smash{$\hVarInline{\cdot}$} denotes the estimated variance. This interval satisfies $\PP{\theta \in \mathcal{I}} \approx 95\%$ under three conditions:
\begin{equation}
\text{(i) }\;\, \hat{\theta} \approx \mathrm{N}\big(\EEInline[\theta]{\hat{\theta}},\, \VarInline[\theta]{\hat{\theta}} \big),\quad \text{(ii) }\;\, \frac{\hVarInline{\hat{\theta}}}{\VarInline[\theta]{\hat{\theta}}} \approx 1,\quad \text{(iii) }\; \,\frac{|\BiasInline[\theta]{\hat{\theta}}|^2}{\VarInline[\theta]{\hat{\theta}}} \approx 0,
\label{eq:wald_assumptions}
\end{equation}
where \smash{$\BiasInline[\theta]{\hat{\theta}} = \EEInline[\theta]{\hat{\theta}} - \theta$}. In this paper, we are concerned with (iii); namely the requirement that \smash{$\hat{\theta}$} be (nearly) unbiased for $\theta$, and we propose methods that relax this condition. The approximate normality in (i) can be justified by the central limit theorem, and the variance in (ii) can be estimated accurately, e.g., via the bootstrap.

Our starting point is that the analyst may prefer a potentially \underline{b}iased estimator \smash{$\biased$} of $\theta$ (say, an ML-based estimator computed on a large unlabeled corpus) because it has substantially lower variance than unbiased alternatives, and because the analyst has reason to believe the bias to be small (Fig.~\ref{fig:intro-progression}a). In terms of overall mean squared error, \smash{$\biased$} may incur a favorable bias-variance tradeoff. At the same time, the analyst may hesitate to report Wald intervals centered at \smash{$\biased$}, since the coverage guarantee breaks down when (iii) is violated. A standard remedy, when available, is to debias: with access to a noisy unbiased estimator \smash{$\hat{b}$} of the bias \smash{$b = \BiasInline[\theta]{\biased}$}, one forms the \underline{d}e\underline{b}iased estimator \smash{$\debiased = \biased - \hat{b}$}, which is unbiased for $\theta$ (Fig.~\ref{fig:intro-progression}b). The debiased estimator is protected against arbitrarily large bias in \smash{$\biased$}: it is unbiased no matter how badly biased \smash{$\biased$} happens to be. Such robustness is well-motivated, since the naïve alternative of trusting \smash{$\biased$} fails to deliver coverage unless the bias is negligible, and without prior knowledge about the bias, it is prudent to be robust against non-negligible bias. The cost is that the variance of \smash{$\debiased$} is generally larger than that of \smash{$\biased$} (by \smash{$\VarInline[b]{\hat{b}}$} when \smash{$\biased$} and \smash{$\hat{b}$} are uncorrelated), undoing the efficiency gains that motivated \smash{$\biased$} in the first place.

\begin{figure}
    \centering
    \includegraphics[width=\linewidth]{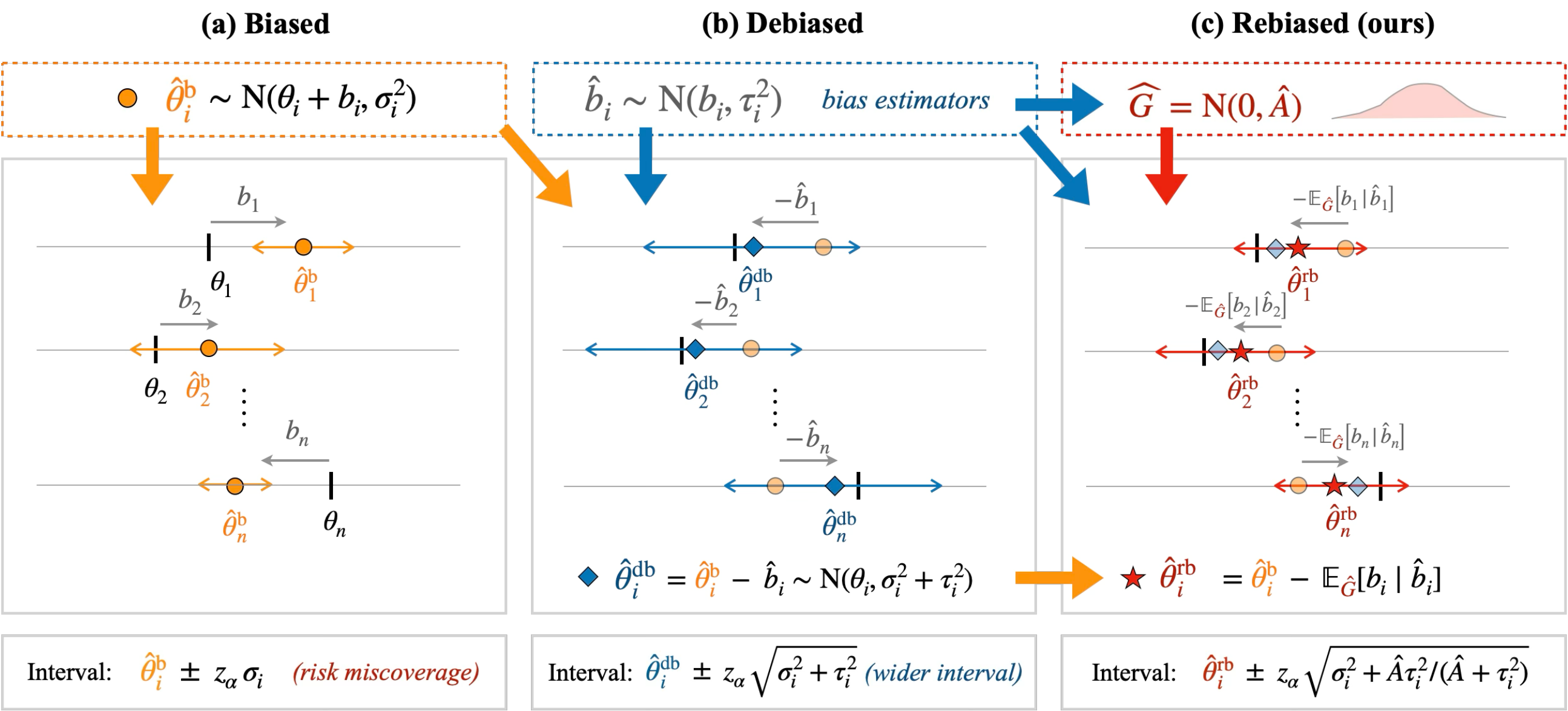}
    \vspace{-1em}
    \caption{\textbf{From biased to debiased to rebiased.}
    Vertical black ticks mark the true $\theta_i$.
    \textbf{(a)} Across $n$ parallel tasks, each biased estimator
    \smash{$\biased_i$} is offset from its target $\theta_i$ by a draw
    $b_i\sim G$; the small variance $\sigma_i^2$ makes \smash{$\biased_i$}\vadjust{\kern1.3pt}
    precise, but the resulting interval may undercover.
    \textbf{(b)} Subtracting an unbiased but noisy estimator \smash{$\hat b_i$} centers
    the estimator at $\theta_i$ but each individual
    \smash{$\debiased_i$} is itself noisy, with a wider confidence interval.
    \textbf{(c)} Empirical Bayes rebiasing estimates the task-level bias distribution \smash{$\widehat G$} from \smash{$\hat b_1,\ldots,\hat b_n$} (shown here as $\widehat G=\mathrm{N}(0,\widehat A)$ for illustration) and uses its posterior mean to partially undo the bias correction and yield a substantially shorter calibrated interval.}
    \label{fig:intro-progression}
    \vspace{-1em}
\end{figure}

In a range of settings, the analyst faces not a single inference task but many. We consider parameters $\theta_1, \ldots, \theta_n$, where $i$ indexes the task; for instance, the $\theta_i$ may be win-rates for $n$ LLM pairs, or effect sizes for $n$ genetic variants. Each task comes with its own pair \smash{$(\biased_i, \hat{b}_i)$}. Across these tasks, the bias estimators \smash{$\hat{b}_1, \ldots, \hat{b}_n$} carry rich indirect information about the typical magnitude of the bias. The central premise of this paper is that bias should not always be treated as either zero, as in the naïve approach of reporting \smash{$\biased_i$} directly, or as arbitrary, as in full debiasing. Instead, we propose to model the distribution of the bias, and we develop new calibrated inference methods that use empirical Bayes (EB) ideas~\citep{robbins1956empirical, efron2010largescale,
stephens2017false} to learn how much bias correction to apply. Concretely, we place a distribution on task-level biases, \smash{$b_i \sim G$} for an unknown $G$, estimate $G$ as \smash{$\widehat{G}$} via empirical Bayes, and use \smash{$\widehat{G}$} to partially undo the full bias correction (i.e., to ``rebias'') when the evidence supports doing so. We refer to this approach as empirical Bayes rebiasing (Fig.~\ref{fig:intro-progression}c).

To preview the methodology,\vadjust{\kern1.3pt} suppose for now that \smash{$G = \mathrm{N}(0, A)$}, where $A \geq 0$ is the variance of the bias, and write \smash{$\sigma_i^2 = \VarInline[\theta_i]{\biased_i}$}\vadjust{\kern1.3pt} and \smash{$\tau_i^2 = \VarInline[b_i]{\hat{b}_i}$}. We model \smash{$\hat{b}_i \mid b_i \sim \mathrm{N}(b_i, \tau_i^2)$} with $\tau_i^2$ known, in the spirit of (i) and (ii) in \eqref{eq:wald_assumptions}. Given access to the unbiased estimators of bias \smash{$\hat{b}_1, \ldots, \hat{b}_n$}, we learn $A$ via marginal maximum likelihood: if biases are small across tasks, then \smash{$\widehat{A} \approx 0$},\vadjust{\kern1.3pt} while if they are large and heterogeneous, then \smash{$\widehat{A} \gg 0$}. We then combine the prior information \smash{$b_i \sim \mathrm{N}(0, \widehat{A})$}\vadjust{\kern1.3pt} with the unbiased estimator \smash{$\hat{b}_i$} via Bayes' rule, yielding the posterior \smash{$b_i \mid \hat{b}_i \sim \mathrm{N}(\widehat{A}/\{\widehat{A}+\tau_i^2\}\hat{b}_i,\, \widehat{A}\tau_i^2/\{\widehat{A}+\tau_i^2\})$},\vadjust{\kern1.3pt} and subtract the posterior mean of the bias rather than \smash{$\hat{b}_i$} itself. The resulting \underline{r}e\underline{b}iased estimator takes three equivalent forms,
\begin{equation}
\rebiased_i \;=\; \biased_i - \frac{\widehat{A}}{\widehat{A} + \tau_i^2}\hat{b}_i \;=\; \debiased_i + \frac{\tau_i^2}{\widehat{A} + \tau_i^2}\hat{b}_i \;=\; \frac{\widehat{A}}{\widehat{A} + \tau_i^2}\debiased_i + \frac{\tau_i^2}{\widehat{A} + \tau_i^2}\biased_i,
\label{eq:rebiased_point_estimator}
\end{equation}
and satisfies \smash{$\rebiased_i - \theta_i \sim \mathrm{N}(0,\, \sigma_i^2 + \widehat{A}\tau_i^2/\{\widehat{A}+\tau_i^2\})$} (with $b_i$ integrated out), yielding the rebiased interval $\mathcal{I}_{i}^{\rb} = \rebiased_i \pm z_{\alpha}(\sigma_i^2 + \widehat{A}\tau_i^2/\{\widehat{A}+\tau_i^2\})^{1/2}$, where $z_{\alpha}$ is the $(1-\alpha/2)$-standard normal quantile. The estimator interpolates between the two endpoints the analyst was previously stuck choosing between: \smash{$\widehat{A} = 0$} recovers \smash{$\biased_i$}, \smash{$\widehat{A} = \infty$} recovers \smash{$\debiased_i$}, and intermediate \smash{$\widehat{A}$} rebiases, i.e., partially undoes the bias correction. As Fig.~\ref{fig:intro-progression} illustrates, $\mathcal{I}_i^{\rb}$ is shorter than the Wald interval of the debiased estimator, with the gap governed by \smash{$\widehat{A}$}. The interval \smash{$\mathcal{I}_i^{\rb}$} synthesizes three sources of information: the biased estimator \smash{$\biased_i$}, the bias correction \smash{$\hat{b}_i$}, and the estimated task-level distribution of the bias \smash{$\widehat{G}$}.

\section{Background and related work}
\label{sec:background}
\looseness=-1
A common modeling assumption is that biases across tasks are exchangeable draws from a distribution, often  normal: $b_i \sim \mathrm{N}(0, A)$ with $A \geq 0$ (or uncentered, $b_i \sim \mathrm{N}(\mu, A)$). This normal component can capture, e.g., biases from unobserved confounding in observational studies, and $A$ can be estimated from negative control studies. For non-control studies, one then reports the inflated interval \smash{$\biased_i \pm z_{\alpha}(\sigma_i^2 + \widehat{A})^{1/2}$}, where \smash{$\sigma_i^2 = \VarInline[\theta_i]{\biased_i}$}. This construction~\citep{schuemie2014interpreting, schuemie2016robust} is routinely used before reporting OHDSI (Observational Health Data Sciences and Informatics) results~\citep{hripcsak2015observational} and has recently been advocated in economics~\citep{bernard2024howmuch}, and for inference with proxy outcomes~\citep{wilkins-reeves2026estimate}. As~\citet{efron2022discussion} notes, treating bias as a further random component is akin to the physicist's ``propagation of uncertainty''; 
we observe that it is also closely related to empirical null
modeling~\citep{efron2004largescale}. We likewise model biases as drawn from a distribution $G$, but synthesize this distributional information with direct per-task estimators \smash{$\hat{b}_i$} to obtain shorter intervals, and allow for nonparametric $G$. The normality assumption on biases also appears in related but distinct settings: \citet{wu2026illusion} use it for meta-analysis combining several biased observational estimators with one unbiased experimental estimator, and \citet{zhao2020statistical} use it to model bias from invalid genetic instruments (systematic pleiotropy) in Mendelian randomization.

Several authors have proposed taking task-wise convex combinations of
unbiased and biased estimators using empirical Bayes
ideas~\citep{green1991jamesstein, green2005improved, chen2015data,
ignatiadis2019covariatepowered, rosenman2023combining,
li2025predictionpowered}, which is precisely the form of the rebiased
estimator in~\eqref{eq:rebiased_point_estimator} (last equality).
Building on the seminal work of~\citet{james1961estimation}, these papers seek estimators with smaller frequentist mean squared error, but focus on point estimation rather than interval coverage. Such methods admit an empirical Bayes interpretation under a normal prior for $b_i$ and a flat
improper prior for $\theta_i$~\citep{green1991jamesstein}. A fully
hierarchical approach with proper priors on both $b_i$ and $\theta_i$ is
taken in~\citet{gelman2021slamming, rosenman2023empirical}. We instead
adopt a partially Bayes analysis~\citep{brown1965secondarily, cox1975note} that places a distribution only on $b_i$, while
treating $\theta_i$ as fixed; in particular, we do not assume
exchangeability of the $\theta_i$ across tasks. Our framing lets us
state coverage guarantees in the empirical-Bayes-coverage
tradition~\citep{morris1983parametric}, with formal
results in~\S\ref{sec:theory}.

A different line of work assumes a deterministic upper bound on the bias, $|b_i| \leq \Delta_i$, and combines~\smash{$\biased_i$}
and~\smash{$\hat{b}_i$} so as to (nearly) minimize the worst-case risk
over $\{|b_i| \leq \Delta_i\}$~\citep{donoho1994statistical,
armstrong2018optimal, lin2026introducing}. Adaptation
results~\citep{cai2004adaptation, armstrong2021sensitivity} show that
$\Delta_i$ cannot be learned adaptively in a way that yields shorter
intervals. These impossibility results concern a single task in
isolation, whereas our setting offers a different opportunity: that of
having many related tasks and the empirical Bayes
paradigm~\citep{robbins1956empirical, efron2010largescale,
stephens2017false}. By replacing the deterministic constraint
$|b_i| \leq \Delta_i$ with a distributional assumption $b_i \sim G$, we
can learn $G$ from the data across tasks.

\section{Statistical setting and proposed methods}
\label{sec:setting}
Our statistical model for the $i$-th task is as follows,
\begin{align}
\label{eq:modeling}
    \begin{pmatrix}
    \biased_i \\
    \hat b_i
    \end{pmatrix}  \,\bigg |\, \begin{pmatrix}
    \theta_i  \\
    b_i
    \end{pmatrix} \,\sim\, 
    \mathrm{N}\left\{
    \begin{pmatrix}
    \theta_i + b_i\\
    b_i
    \end{pmatrix},\,
    \begin{pmatrix}
    \sigma_i^2
    &
    \rho_i \sigma_i \tau_i
    \\
    \rho_i \sigma_i \tau_i
    &
    \tau_i^2
    \end{pmatrix}
    \right\},\qquad b_i \sim G,\qquad \theta_i \,\text{ is fixed}.
\end{align}
To focus on the issue of the bias, we treat normality as holding exactly and assume that $\rho_i \in (-1,1)$ and $\sigma_i^2, \tau_i^2>0$ are known (these pertain to (i) and (ii) in~\eqref{eq:wald_assumptions}). 
Meanwhile, $\theta_i$, $b_i$, and $G$ are unknown. We posit that all biases $b_i$ are drawn from the same distribution $G$, that is, we assume bias is exchangeable across tasks and $G$ does not depend on $\theta_i, \sigma_i, \rho_i, \tau_i$. We treat the $\theta_i$ as fixed (as in a frequentist analysis). An analysis in which we treat the primary parameters $(\theta_i)$ as fixed, while the nuisance parameters $(b_i)$ as drawn from a distribution is called a partially Bayes analysis~\citep{cox1975note}, which dates back to a proposal by John Tukey~\citep{brown1965secondarily}; also see~\citet{ignatiadis2025empirical}.

Our procedure below also admits a (fully) Bayes interpretation by placing a flat prior on $\theta_i$, with $\theta_i$ independent of $b_i\sim G$.\footnote{
One can approximately interpret the flat prior as follows:
start with a fixed smooth proper prior on $\theta_i$, with $\theta_i$ independent
of $b_i$. As the sample size within each task grows, rescale the
bias distribution $G$ so that biases remain of the same order as
the standard errors. The prior on $\theta_i$ then becomes locally
flat, while $G$ remains informative
\citep[Section~5E]{hjort1993estimation}.
} We prefer the partially Bayes interpretation because it makes explicit that we assume shared structure in the biases, without imposing such structure on the $\theta_i$.

\paragraph{Oracle partially Bayes rebiasing.} We start by explaining our approach to inference when the bias distribution $G$ in~\eqref{eq:modeling} is known. (Later, we will explain how to estimate $G$.) The goal is to combine the distribution of the biases ($G$) with  \smash{$\hat b_i$} via (a partial version of) Bayes' rule. 
We state our constructions in a general way that accommodates \emph{any} fixed distribution $G$. When $G=\mathrm{N}(\mu,A)$, these expressions simplify, as we saw for $\mu=0$ in \S\ref{sec:introduction}; see Appendix~\ref{app:normal} for general $\mu$. To streamline exposition, below we assume $\rho_i = 0$; the general case (relevant for our applications) is treated in
Appendix~\ref{app:generalization}.

Although point estimation is not our focus, for intuition, we first state a generalization of the estimator~\smash{$\rebiased_i$} in~\eqref{eq:rebiased_point_estimator}. Rather than subtracting
\smash{$\hat b_i$} from \smash{$\biased_i$}, we subtract the posterior mean of the bias,\footnote{
A similar estimator for general $G$ appears in~\citet[Proposition 1]{kwon2024empirical} based on a flat prior for $\theta_i$.
}
\begin{equation}
\rebiased_i=\biased_i-\mathbb E_G[b_i\mid \hat b_i]=\hat\theta_i^{\mathrm{db}} +\left\{\hat b_i-\mathbb E_G[b_i\mid \hat b_i]\right\}.
\label{eq:postmean_general}
\end{equation}
We also interpret this procedure via the right-hand side of~\eqref{eq:postmean_general}
as rebiasing the debiased estimator \smash{$\debiased_i$}.
In the posterior for $b_i$, we condition only on \smash{$\hat b_i$}. Explicitly, for any measurable set $B \subset \RR$, we let
$$
\PPInline[G]{b_i \in B \cond \hat{b}_i} = \frac{ \int_{B} \varphi(\hat{b}_i - b;\, \tau_i^2) G(\dd b)}{\int \varphi(\hat{b}_i - b;\, \tau_i^2) G(\dd b)},\;\;\; \EEInline[G]{b_i \cond \hat{b}_i} = \frac{ \int b\,\varphi(\hat{b}_i - b;\, \tau_i^2) G(\dd b)}{\int \varphi(\hat{b}_i - b;\, \tau_i^2) G(\dd b)},
$$
where $\varphi(\cdot; \tau_i^2)$ is the density of the $\mathrm{N}(0, \tau_i^2)$ distribution. For our purposes, the  relevant object is the conditional density of \smash{$\debiased_i - \theta_i$} conditional on \smash{$\hat{b}_i$} (that is, integrating over the posterior \smash{$b_i \mid \hat{b}_i)$},
\begin{equation}
\label{eq:db_cond_density}
    f_{G,i}(t \mid \hat{b}_i) \;=\; \frac{\int \varphi(t - b + \hat{b}_i;\, \sigma_i^2)\, \varphi(\hat{b}_i - b;\, \tau_i^2)\, G(\dd b)}{\int \varphi(\hat{b}_i - b;\, \tau_i^2)\, G(\dd b)}.
\end{equation}
To derive the above conditional density, observe that
\smash{$\debiased_i - \theta_i= \{\biased_i - (\theta_i+b_i)\} + \{b_i - \hat{b}_i\}$} and so \smash{$\debiased_i - \theta_i \mid b_i, \hat{b}_i \sim \mathrm{N}(b_i - \hat{b}_i, \sigma_i^2)$}; then integrate against the posterior of $b_i$ given \smash{$\hat{b}_i$}. 

Next, let us write \smash{$F_{G,i}(\cdot \mid \hat{b}_i)$} for the distribution function with density \smash{$f_{G,i}(\cdot \mid \hat{b}_i)$} and \smash{$q_{G,i,\alpha}(\hat{b}_i)$} for its $\alpha$-quantile. The oracle rebiased equal-tailed $(1-\alpha)$-interval is then defined as
\begin{equation}
\label{eq:oracle_ci}
    \mathcal{I}_{G,i}^{\rb}(1-\alpha)=\left[\hat\theta_i^{\mathrm{db}}-q_{G,i,1-\alpha/2}(\hat b_i),\,\hat\theta_i^{\mathrm{db}}-q_{G,i,\alpha/2}(\hat b_i) \right].
\end{equation}
Beyond intervals, we  also consider the corresponding testing
problem. For testing $H_{0i}: \theta_i=\theta_{i0}$ for pre-specified $\theta_{i0}$, we define the oracle rebiased p-value as,
\begin{equation}
\label{eq:def_pvalue}
P_{G,i}^{\rb}=P_{G,\theta_{i0}}^{(i)}(\hat\theta_i^{\mathrm{db}},\hat b_i),\,\text{ with }\,  P_{G,\theta_{i0}}^{(i)}(z,l)
    =
    2\min\!\left\{
    F_{G,i}(z-\theta_{i0}\mid l),
    1-F_{G,i}(z-\theta_{i0}\mid l)
    \right\}.
\end{equation}
Our proposed oracle rebiased interval in~\eqref{eq:oracle_ci} can be derived by test inversion of the oracle p-values, 
\begin{equation}
\label{eq:CI_inversion}
    \mathcal{I}^{\rb}_{G,i}(1-\alpha)=\{\theta_{i0}\in\mathbb{R}\,:\,\alpha/2 \leq F_{G,i}(\hat{\theta}_i^{\db}-\theta_{i0}\mid\hat{b}_i)\leq1-\alpha/2\}.
\end{equation}

\paragraph{Empirical Bayes rebiasing.}  We next turn to the estimation of $G$ in~\eqref{eq:modeling}. As mentioned in \S\ref{sec:background}, a common modeling choice for the bias distribution is that $G=\mathrm{N}(\mu, A)$. We can estimate the unknown $\mu$ and $A$ using maximum marginal likelihood; see Appendix~\ref{app:computation}. We emphasize, that unlike normality statements for \smash{$\biased_i$} and \smash{$\hat{b}_i$} in~\eqref{eq:modeling}, normality of $b_i$ does not follow from the central limit theorem. For this reason, when normal $G$ is imposed, it is prudent to check the fit of the model. In what follows we pursue an alternative avenue: estimating $G$ in a fully nonparametric way, without imposing a parametric form.

We propose to estimate $G$ based on \smash{$(\hat{b}_1,...,\hat{b}_n)$} by maximizing the marginal log-likelihood,
\begin{equation}
\label{eq:npmle_G}
    \widehat G \in \argmax_{G^{\prime}} \cb{\frac{1}{n}\sum_{i=1}^n \log\p{ \int \varphi(\hat{b}_i - b;\, \tau_i^2)\, G'\!(\dd b)}},
\end{equation}
over \emph{all} possible distributions $G'$. This is the nonparametric maximum likelihood estimator (NPMLE) of \cite{robbins1950generalization, keifer_npmle}. The NPMLE is tuning-free and has strong regret guarantees for empirical Bayes problems~\citep{jiang2009general}. Although the feasible set is infinite-dimensional, the optimizer is a discrete distribution with at most $n$ atoms, and typically only $O(\log n)$, supported in \smash{$[\min_i\{\hat b_i\},\, \max_i \{\hat b_i\}]$} \citep{lindsay1995mixture, polyanskiy2020selfregularizing}. The problem can be discretized without sacrificing statistical guarantees \citep{dicker2016highdimensional, soloff2024multivariate} and recast as a conic program \citep{koenker2014convex, koenker2017rebayes}, which we solve with MOSEK \citep{aps2024mosek}. Further details appear in Appendix~\ref{app:computation}.

Given our estimate of \smash{$\widehat G$} (which could be the parametric normal fit or the NPMLE), we then compute the rebiased intervals using the plug-in principle, that is, pretending that \smash{$\widehat{G}$} is the true distribution $G$.
The rebiased $(1-\alpha)$-interval is constructed as
\begin{equation}
\label{eq:CI_npmle}
\mathcal{I}^{\rb}_{\widehat G,i}(1-\alpha)
    =
    \left[
    \hat\theta_i^{\mathrm{db}}-q_{\widehat G,i,1-\alpha/2}(\hat{b}_i),
    \,
    \hat\theta_i^{\mathrm{db}}- q_{\widehat G,i,\alpha/2}(\hat{b}_i)
    \right].
\end{equation}
We can also estimate the oracle rebiased p-value \smash{$P^{\rb}_{G,i}$} by the empirical Bayes rebiased p-value \smash{$\widehat{P}^{\rb}_{\widehat{G},i}$}\vadjust{\kern1.3pt}, obtained by replacing $G$ with \smash{$\widehat G$} in \eqref{eq:def_pvalue}.
All the computations can be carried out readily; e.g.,\linebreak for the NPMLE we use the fact that it yields a discrete \smash{$\widehat{G}$} and so integrals reduce to sums.

\section{Applications of empirical Bayes rebiasing}
\label{sec:applications}

Our framework can be applied in two ways, depending on the auxiliary estimator paired with~\smash{$\biased_i$}. In the first, the researcher has the biased estimator~\smash{$\biased_i$} and a direct bias-correction estimator~\smash{$\hat{b}_i$}, as in the setup so far. In the second, the researcher has access to~\smash{$\biased_i$}\vadjust{\kern1.2pt} together with an \underline{u}n\underline{b}iased estimator~\smash{$\unbiased_i$}, in which case the framework applies with~\smash{$\hat{b}_i := \biased_i - \unbiased_i$} (so that~\smash{$\debiased_i = \unbiased_i$}). Although formally equivalent, the two cases arise differently in practice. Here, we consider two prototypical applications: prediction-powered inference (PPI) and family-based genome-wide association studies (GWAS). The framework is broadly applicable beyond these examples; for instance, it covers combining observational and experimental causal estimates~\citep{rosenman2023combining, rosenman2023empirical}, combining exposed-only with difference estimates in sham-controlled experiments~\citep{gelman2021slamming}, and combining estimates from a narrow parametric model with those from a wider model that nests it~\citep{hjort1993estimation}.

\paragraph[Prediction-powered inference (PPI)]{Prediction-powered inference (PPI).} PPI~\citep{angelopoulos2023predictionpowered, angelopoulos2024ppi} is a framework for constructing more accurate estimates by incorporating predictions from a black-box ML model $h$. For task $i$, we observe a small labeled sample \smash{$\{(X_{ij}, Y_{ij})\}_{j=1}^{m_i}$} with $Y_{ij} \in \RR$ and a larger unlabeled sample \smash{$\{\tilde X_{ij}\}_{j=1}^{M_i}$} (with $m_i \ll M_i$ and \smash{$\tilde X_{ij} \overset{d}{=} X_{ij}$}); the goal is to estimate $\theta_i = \EE{Y_{ij}}$.\footnote{Our framework extends to general estimands beyond the mean via the predict-then-debias strategy described in~\citet{kluger2025predictionpowered}; see also~\citet{miao2024taskagnostic}. Rebiasing also applies to batch active statistical inference~\citep{zrnic2024active}, an extension of PPI in which labels are sampled with covariate-dependent probabilities (Appendix~\ref{subsec:active_app_details}).
} The classical mean $\bar Y_i = m_i^{-1}\sum_j Y_{ij}$ is unbiased but has variance $\propto 1/m_i$, while the ML-only estimate \smash{$\hthetaML = M_i^{-1}\sum_j h(\tilde X_{ij})$}\vadjust{\kern1.3pt} has small variance ($\propto 1/M_i$) but bias $b_i = \EE[\theta_i]{h(X_{ij})} - \theta_i$ that depends on the quality of $h$. The PPI estimator bridges the two by debiasing \smash{$\hthetaML$} with the labeled sample:
\begin{equation}
\hthetaPPI \,=\, \hthetaML - \hat b_i, \quad \hat b_i \,=\, m_i^{-1}\textstyle\sum_{j=1}^{m_i}\{h(X_{ij}) - Y_{ij}\}.
\label{eq:ppi}
\end{equation}
Setting \smash{$\biased_i = \hthetaML$} as the biased estimator, the PPI estimator is exactly the debiased estimator $\debiased_i$ in our framework. Vanilla PPI is asymptotically dominated by the PPI++/PT (power tuned) estimator of~\citet{angelopoulos2024ppi} that takes the form 
$$
\hthetaPT = (1-\lambda_i^*)\bar Y_i + \lambda_i^*\hthetaPPI \,\;\text{ for a choice of } \lambda_i^*.
$$
To rebias PT, we treat \smash{$\hthetaPT$} as the unbiased estimator \smash{$\unbiased_i$} and \smash{$\hthetaML$} as the biased estimator \smash{$\biased_i$}. Our experiments use PT throughout. We provide more details in Appendix~\ref{subsec:ppi_app_details}, including expressions for $\rho_i, \sigma_i, \tau_i$ in~\eqref{eq:modeling} and the plug-in estimates we use in practice.

PPI/PT benefit from the ML predictor without requiring it to be accurate, but they consider only unbiased estimators of $\theta_i$. Our rebiasing goes further by learning the bias distribution across tasks from the observed bias estimates. When the data reveal that biases are small, i.e., $h$ is approximately calibrated (which is a milder requirement than accurately predicting $Y_{ij}$), the rebiased estimator moves back toward \smash{$\hthetaML$}, yielding shorter intervals; when biases are large, it stays close to PPI/PT.
The prediction-powered adaptive shrinkage estimator of~\citet{li2025predictionpowered} also shares information across tasks, but aims to reduce the mean squared error of PT rather than provide calibrated inference.

\paragraph{Inference of direct genetic effects in family-based GWAS.}
A standard genome-wide association study (GWAS) regresses a trait on each SNP (single nucleotide polymorphism) across the genome. The regression coefficient, denoted $\smash{\biased_i}$ in our framework, and often termed the population-effect estimate in the literature, is biased as it captures not only the direct genetic effect $\theta_i$, but also a bias component $b_i$, which may reflect contributions from population stratification, indirect genetic effects, and assortative mating \citep{Kong2018-ld}. The variance of $\smash{\biased_i}$ typically scales as $1/m_i$, where $m_i$ is the regression sample size for the $i$-th SNP. When parental genotypes are observed or unbiasedly imputed~\citep{Young2022-tz, Guan2025-jq}, family-based GWAS includes them as covariates in the SNP-level regression; the resulting offspring-genotype coefficient \smash{$\unbiased_i$} is unbiased for $\theta_i$, but has the larger variance \smash{$\propto 1/\{m_i(1-r_i^2)\}$}, where $r_i$ is the offspring-parental genotype correlation at SNP $i$. The pair \smash{$(\biased_i, \unbiased_i)$} fits our framework with the bias estimator \smash{$\hat b_i = \biased_i - \unbiased_i$}, which coincides with the parental-genotype regression coefficient.
Summary statistics provided for family-based GWAS allow us to also recover $\rho_i, \sigma_i^2, \tau_i^2$ in~\eqref{eq:modeling} (see Appendix~\ref{subsec:gwas_app_details} for details) and so we can apply our empirical Bayes rebiasing strategy that 
learns the bias distribution across SNPs.

\section{Theoretical results}
\label{sec:theory}

For our first results, we work under the partially Bayes model~\eqref{eq:modeling}, with $\theta_i$ a fixed parameter, \smash{$b_i \sim G$}, and \smash{$(\hat{\theta}_i^{\db}, \hat{b}_i)$} jointly normal given $(\theta_i, b_i)$.  Later in this section, we also treat $\theta_i$ and $b_i$ as both random (Theorem~\ref{thm:rate_coverage_joint}) and as both fixed (Theorem~\ref{thm:rate_average_coverage}).
We write $\mathbb P_{\boldsymbol\theta,G}$ and $\mathbb E_{\boldsymbol\theta,G}$ for probabilities and expectations taken jointly over the $n$ tasks under the partially Bayes model with $\boldtheta=(\theta_1,\ldots,\theta_n)$ fixed. Our goal is to show coverage guarantees in this setting, in the spirit of the empirical Bayes coverage tradition~\citep{morris1983parametric}, except that the classical version also averages over the parameters of interest $\theta_i$.
Throughout, \smash{$\widehat{G}$} denotes the NPMLE in~\eqref{eq:npmle_G}; 
we focus on the nonparametric case, which both reveals the statistical 
structure of the problem and the cost of fully nonparametric modeling.

Let \smash{$\tilde{\sigma}^2_i=\VarInline[\theta_i]{\hat{\theta}_i^{\db}}=\sigma_i^2+\tau_i^2-2\rho_i\sigma_i\tau_i$}, \smash{$\gamma_i = \CorrInline[\theta_i,b_i]{\hat{\theta}_i^{\mathrm{db}},\hat b_i}= (\rho_i\sigma_i-\tau_i)/\tilde{\sigma}_i$},\vadjust{\kern1pt} and 
$\check\gamma=\max_i|\gamma_i|$. For $\Gamma>0$ we write $\mathcal{G}_{\Gamma} = \{G'\,:\,\EE[G']{\exp\{\lambda (b - \EE[G']{b})\}} \le \exp(\Gamma\lambda^2/2) \text{ for all }\lambda\in\RR\}$ for the class of $\Gamma$-sub-Gaussian distributions. Our main assumption in this section is as follows.
\begin{assumption}
\label{assm:subgaussian}
    We assume that the true bias distribution satisfies $G \in \mathcal{G}_{\Gamma}$ for some $\Gamma>0$ and that the triples \smash{$(b_i, \biased_i, \hat{b}_i)$} for $i=1,\ldots,n$ are jointly independent. Finally, we assume that for all $i$, $ 0<\underline{\gamma}\leq|\gamma_i| \leq\bar{\gamma}<1$ and $0<\underline{\sigma}^2\leq\tilde{\sigma}_i^2,\tau_i^2\leq \bar{\sigma}^2 <\infty$ for some constants $\underline{\gamma},\bar{\gamma},\underline{\sigma}, \bar{\sigma}>0$.
\end{assumption}
The next result establishes the rate at which we estimate the oracle rebiasing p-values in~\eqref{eq:def_pvalue}.
\begin{theorem}
\label{thm:covergence_pvalue}
Under Assumption~\ref{assm:subgaussian}, there exists 
$C=C(\Gamma,\underline{\gamma},\bar{\gamma},\underline{\sigma},\bar{\sigma})$ 
such that, 
$$
    \frac{1}{n}\sum_{i=1}^n\EE[\boldtheta, G]{\left|P_{G,\theta_{i0}}^{(i)}(\hat{\theta}^{\mathrm{db}}_i,\hat{b}_i)-P_{\widehat{G},\theta_{i0}}^{(i)}(\hat{\theta}^{\mathrm{db}}_i,\hat{b}_i)\right|} \leq C\frac{(\log{n})^{3/2}}{n^{(1-\check{\gamma}^2)/2}}\quad\text{for all } n\in\mathbb{N}_{\geq2}.
$$
\end{theorem}
For moderate $\check{\gamma}$ and $n$, this rate suggests little cost to 
pursuing the fully nonparametric strategy; as \smash{$\check{\gamma}\nearrow 1$}, 
however, the rate slows, and stronger (e.g., parametric) modeling assumptions on 
$G$ may be worthwhile. The value of \smash{$\check{\gamma}$} is known to the analyst;
in our three applications (\S\ref{sec:numerical}), $\check{\gamma}\in\{0.752, 0.891, 0.991\}$.
The dependence on $\check{\gamma}$ is nearly sharp in a minimax sense.
\begin{proposition}
Suppose that $\tilde{\sigma}_i^2, \tau_i^2, \gamma_i$ do not depend on $i$ and drop the subscript $i$. Let $\check{\gamma}=|\gamma| \in (0,1)$. Fix $t \in \RR$.
Then, for any $\beta > (1-\check{\gamma}^2)/2$ there exists $\Gamma>0$ and $c>0$ such that
$$
\inf_{\widehat{\psi}(t)} \sup_{G \in \mathcal{G}_\Gamma} \EE[\boldtheta, G]{\abs{\widehat{\psi}(t) - F_{G,i}(t \mid \hat{b}_i) }} \geq c \frac{1}{n^{\beta}}, \mbox{~where $\widehat{\psi}(t)$ is any measurable function of \smash{$\hat{b}_1,\ldots,\hat{b}_n$}}.
$$

\label{prop:minimax}
\end{proposition}

Building upon Theorem \ref{thm:covergence_pvalue}, we have the following convergence rate for the coverage of \smash{$\mathcal{I}^{\rb}_{\widehat{G},i}(1-\alpha)$}.
\begin{theorem}
\label{thm:rate_coverage}
Under Assumption~\ref{assm:subgaussian}, there exists 
$C'=C'(\Gamma,\underline{\gamma},\bar{\gamma},\underline{\sigma},\bar{\sigma})$ 
such that,
    $$
    \sup_{\alpha \in (0,1)}\frac{1}{n}\sum_{i=1}^n\left|\PP[\boldtheta,G]{\theta_i\in\mathcal{I}^{\rb}_{\widehat{G},i}(1-\alpha)}-(1-\alpha)\right|\leq C^{\prime}\frac{(\log{n})^{3/2}}{n^{(1-\check{\gamma}^2)/2}}\quad\text{for all } n\in\mathbb{N}_{\geq2}.
    $$
\end{theorem}

A classical empirical Bayes analysis would instead estimate the distribution of $\theta_i$ using the NPMLE. The NPMLE's discreteness can lead to poor coverage in finite samples~\citep{jiang2019comment, koenker2020empirical}. The partially Bayes formulation sidesteps this issue, since the NPMLE is applied to $b_i$ and integrated out when forming intervals on $\theta_i$, so its discreteness does not propagate to the resulting inference.

\paragraph{Hierarchical sampling model.}
A hierarchical Bayes model would typically posit
\smash{$(\theta_i,b_i)\sim H$} for a (proper) joint distribution $H$. Under the additional assumption of independence between $\theta_i$ and $b_i$, we could apply Theorem~\ref{thm:rate_coverage} conditional on the $\theta_i$. 
Our next theorem avoids this independence assumption, which may be implausible in some applications~\citep{schuemie2018empirical}, by positing $(\theta_i,b_i)\sim H_i$, where each $H_i$
has the same bias marginal, $b_i \sim G$, but can vary otherwise (that is, the marginal distribution of $\theta_i$ and the dependence of $(\theta_i, b_i)$ can vary with $i$).\footnote{
This gives another interpretation of our method: we posit a
common distribution $G$ for the biases while leaving the
distributions of the $\theta_i$ unrestricted across tasks.} We also assume that the tasks are
independent and that, conditional on $(\theta_i,b_i)$, the
observations \smash{$(\biased_i,\hat b_i)$} follow the normal distribution
in~\eqref{eq:modeling}. The coverage statement is averaged over both $\theta_i$ and $b_i$.

\begin{theorem}
\label{thm:rate_coverage_joint}
Under the hierarchical sampling model above and Assumption~\ref{assm:subgaussian}, the coverage bound in Theorem~\ref{thm:rate_coverage} holds with $\mathbb P_{\boldsymbol{\theta},G}$ replaced by $\mathbb P_{\boldsymbol H}$, the probability law under $\boldsymbol H=(H_1,\ldots,H_n)$.
\end{theorem}

\paragraph{Compound decisions.}
We also obtain a purely frequentist guarantee with both the 
$\boldtheta=(\theta_1,\ldots,\theta_n)$ and $\boldb = (b_1,\ldots,b_n)$ fixed. 
The guarantee, which builds on compound decision theory~\citep{robbins1951asymptotically, zhang2003compound}, concerns average coverage across the $\theta_i$~\citep{wahba1983bayesian, armstrong2022robust}. 
We write $\mathbb P_{\boldtheta,\boldb}$ to denote the probability with $\boldtheta$ and $\boldb$ fixed and the pairs \smash{$(\biased_i,\hat b_i)$} jointly independent across tasks with the normal distribution in~\eqref{eq:modeling}.
\begin{theorem}
\label{thm:rate_average_coverage}
Suppose that $\tilde{\sigma}_i^2, \tau_i^2, \gamma_i$ do not depend on $i$ (and drop the subscript). Let $\check{\gamma}=|\gamma| \in (0,1)$, and $0<\tilde{\sigma}^2,\tau^2<\infty$. Suppose that $|b_i| \leq M$ for all $i$.
Then there exists 
$C'=C'(M, \gamma, \tilde{\sigma}^2, \tau^2)$ 
with
$$
\sup_{\alpha\in(0,1)}\left|\frac{1}{n}\sum_{i=1}^n\PP[\boldtheta,\boldb]{\theta_i\in\mathcal{I}^{\rb}_{\widehat{G},i}(1-\alpha)}-(1-\alpha) \right|\;\leq \;  C^{\prime}\frac{(\log{n})^{3/4}}{n^{(1-\check{\gamma}^2)/4}}\quad\text{for all } n\in\mathbb{N}_{\geq2}.
$$
\end{theorem}

\paragraph{Further robustness to assumptions.}
The preceding results assume exact normality in~\eqref{eq:modeling}, known second moments, and independence across tasks. In practice, normality is an approximation motivated by the CLT, we replace $\sigma_i,\tau_i,\rho_i$ by plug-in estimates, and there may be cross-task dependence. Appendix~\ref{app:robustness} extends the coverage guarantees of Theorems~\ref{thm:rate_coverage},~\ref{thm:rate_coverage_joint}, and~\ref{thm:rate_average_coverage} to approximate normality, estimated second moments, and weak dependence across tasks. These extensions establish consistency without quantitative rates,
and build on results of~\citet{chen2026normal}. Our robustness results still assume that the distribution of $b_i$ does not depend on $\theta_i, \sigma_i, \tau_i, \rho_i$ (for Theorem~\ref{thm:rate_coverage}) or on $\sigma_i, \tau_i, \rho_i$ (for Theorem~\ref{thm:rate_coverage_joint}), or that all second moments are identical across tasks (for Theorem~\ref{thm:rate_average_coverage}). Even with these robustness results, it is prudent to consider the suitability of the assumptions in practice. 
For each of our applications in \S\ref{sec:numerical}, we provide diagnostics and a discussion of these assumptions in Appendix~\ref{app:robustness_diagnostic}.

\section[Numerical results]{Numerical results\footnote{We provide code to reproduce all numerical results: \url{https://github.com/wyling01/eb-rebiasing}. }}
\label{sec:numerical}

\subsection{Experiments for prediction-powered inference}
\label{subsec:numerical_ppi}

\paragraph{Intervals.} Recall the PPI setup of \S\ref{sec:applications}. We compare our rebiased intervals against five baselines: the \emph{Classical} interval ($\bar Y_i \pm z_{\alpha}\VarInline{\bar{Y}_i}^{1/2}$; not using ML information); the \emph{Pred Mean} interval (\smash{$\hthetaML \pm z_{\alpha}\VarInline{\hthetaML}^{1/2}$}; without correcting for bias); the power-tuned PPI \emph{PT} interval of~\citet{angelopoulos2024ppi} centered around the unbiased \smash{$\hthetaPT$}; the empirical Bayes confidence intervals of~\citet{armstrong2022robust}, applied to the classical and PT point estimators to obtain the \emph{EBCI Classical} and \emph{EBCI PT} intervals, respectively. Our rebiased intervals (\emph{RB}) are built on top of PT by fitting the bias distribution $G$ on $b_i$, either as normal $\mathrm{N}(\mu, A)$ (suffix \emph{Normal}) or nonparametrically (suffix \emph{NPMLE}).

\begin{figure}
\centering
\includegraphics[width=\linewidth]{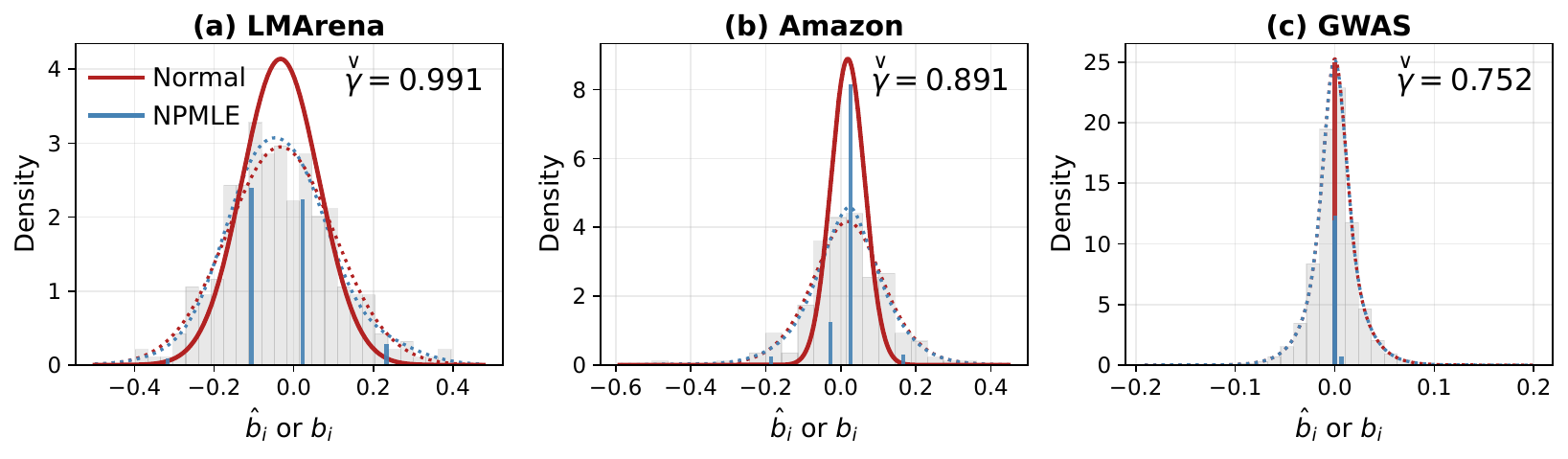}
\vspace{-2em}
\caption{\textbf{Estimated bias distributions in the three applications.} Histogram of the estimator of bias \smash{$\hat{b}_i$} with the distribution estimated for the bias $b_i$ (solid line) and the average marginal density of \smash{$\hat{b}_i$} implied by the estimated bias distribution (dotted line) overlaid.  Two approaches to estimating $G$ are shown: \emph{Normal} and the \emph{NPMLE}. The fitted estimated bias distributions are more concentrated than the empirical distribution of \smash{$\hat{b}_i$}, as the latter is further dispersed due to noise in the bias estimators. 
The implied average marginal densities align well with the observed histograms.}
\label{fig:bias_fit}
\vspace{-1em}
\end{figure}

\paragraph{Metrics.} We report the empirical \emph{miscoverage rate} at nominal level $1-\alpha$, the average \emph{width}, and the average \emph{width-ratio} relative to the \emph{Classical} interval (smaller is better). For PPI applications involving real-world datasets, the true $\theta_i$ values are not directly observed. Instead, we follow common
practice in the PPI literature~\citep{angelopoulos2023predictionpowered}: we start with a fully labeled dataset, and treat the empirical mean over it as pseudo-ground truth for $\theta_i$. (We acknowledge that this pseudo-ground truth is itself estimated and therefore introduces additional uncertainty.)
Then, we randomly partition the data into $20/80$ labeled/unlabeled subsets by (artificially)  masking labels in the ``unlabeled'' subset, imitating a situation with scarcer access to labels. We report all metrics averaged over $K = 200$ replications of this random partitioning,  with Monte Carlo standard errors reported in Appendix~\ref{app:experiment}.

\paragraph{LM Arena.}

We first consider the problem of estimating $n = 298$ average win rates between pairs of LLMs on the LM Arena platform~\citep{chiang2024chatbot}. The $i$-th task corresponds to a comparison between two LLMs (LLM A versus LLM B), where the covariates $X_{ij}$ consist of the two models’ responses to a prompt, and $Y_{ij} \in \{0,1\}$ denotes the human preference outcome (with $Y_{ij}=1$ if LLM A is preferred, and $0$ otherwise). The primary parameter $\theta_i$ therefore represents the probability that humans prefer LLM A over LLM B across prompts. Recently, applying PPI to LLM evaluation and ranking has attracted growing interest; see, for example,~\citet{chatzi2024predictionpowered}. We generate predictions from the raw scores of the \texttt{Skywork-reward-v2} reward model~\citep{liu2025skywork} followed by a Bradley-Terry transformation (see Appendix~\ref{app:experiment-lmarena} for details).
Fig.~\ref{fig:bias_fit}a shows the fitted normal and NPMLE distributions for $b_i$ used to construct the rebiased intervals, together with their implied marginal densities on $\smash{\hat{b}_i}$.

The predictor is biased enough that naïve intervals built from its predictions alone fail to cover the truth $65\%$ of the time at $\alpha=0.10$. \emph{RB-Normal} can still effectively leverage this predictor: it has a much smaller miscoverage rate of  $11\%$ while reducing average interval width by approximately $25\%$ relative to \emph{Classical}, $23\%$ relative to \emph{PT}, and $14\%$ relative to \emph{EBCI PT}. \emph{RB-NPMLE} further reduces width by approximately $29\%$, $27\%$, and $19\%$, respectively, but has a miscoverage rate of $15\%$; see Fig.~\ref{fig:ppi-over-alphas}a, which also shows results at other values of $\alpha$. The more severe undercoverage of \emph{RB-NPMLE}\footnote{The miscoverage of \emph{RB-NPMLE} may also reflect departures from our assumptions. As discussed in Appendix~\ref{app:robustness_diagnostic}, the small labeled samples of binary outcomes per task may lead to a poor task-level normal approximation and imprecise plug-in second-moment estimates. Fig.~\ref{fig:diagnostic_nuisance}a also suggests that the distribution of $b_i$ may vary with the task-specific second moments.} is in line with Proposition~\ref{prop:minimax}: the lower bound there implies a slow estimation rate as \smash{$\check{\gamma} \nearrow 1$}, and here \smash{$\check{\gamma} = 0.991$}. The better performance of \emph{RB-Normal} in the same regime is consistent with the remark following Theorem~\ref{thm:covergence_pvalue}: when $\check{\gamma}$ is close to one, committing to a well-motivated parametric family for $G$ is a defensible alternative to fully nonparametric estimation.

\paragraph{Amazon Fine Food.}
Following \citet{li2025predictionpowered}, we consider parallel PPI tasks using the \textit{Amazon Fine Food Reviews} dataset~\citep{snap2014amazon, mcauley2013amateurs}, where the goal is to recover the average customer rating for each of $n = 200$ food products. For product $i$, the parameter $\theta_i = \EE{Y_{ij}}$ is the population mean rating, with $Y_{ij} \in \{1,\ldots,5\}$ the star score assigned by reviewer $j$ and $X_{ij}$ the concatenation of that review's title and body text. 
The predictor $h$ is fine-tuned BERT~\citep{devlin2019bert} (with fine-tuning on a disjoint pool of reviews). 
This application is meant to mimic situations in which annotation of text by trained human coders can be supplemented by ML predictions~\citep{fan2024narratives,mozer2025more}; here, the star rating plays the role of the human-coded label.
We leave additional details of our dataset and predictor to~\cref{app:experiment-amazon}.

On the Amazon benchmark, the predictor is fairly informative; this is evident in the concentrated estimated bias distributions (Fig.~\ref{fig:bias_fit}b). Intervals built from the prediction mean alone are about half the width of the \emph{Classical} interval, but they exhibit undercoverage ($11.3\%$ miscoverage at $\alpha=0.1$). Even so, our rebiasing approach further improves coverage and delivers shorter intervals than competing methods (see Fig.~\ref{fig:ppi-over-alphas}b). At $\alpha=0.1$, \emph{RB-NPMLE} attains a miscoverage rate of $3.7\%$ with width $0.286$, roughly $48\%$ shorter than \emph{Classical}, $17\%$ shorter than \emph{PT}, and $15\%$ shorter than \emph{EBCI PT}. 

\begin{figure}[t]
    \centering
    \includegraphics[width=\linewidth]{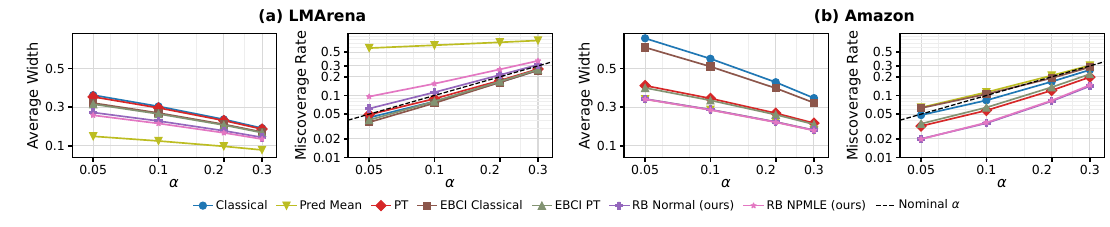}
    \vspace{-1.5em}
    \caption{\textbf{Interval width and miscoverage rate for LM Arena and Amazon Fine Food.} Average width (left) and miscoverage rate (right) for two PPI applications. \textbf{(a)} Pairwise LLM win-rate tasks ($n=298$) in LM Arena dataset. 
    Both rebiasing methods shorten the intervals, but \emph{RB-NPMLE} has higher miscoverage. \textbf{(b)} Product rating problems ($n=200$) in the Amazon dataset. Both \textit{RB-Normal} and \textit{RB-NPMLE} produce the shortest intervals while attaining miscoverage rates below the nominal level.
    Standard errors are reported in Tables~\ref{tbl:lmarena_ci}, \ref{tbl:amazon_ci}.}
    \label{fig:ppi-over-alphas}
    \vspace{-1.3em}
\end{figure}

\begin{figure}
    \centering
\includegraphics[width=\linewidth]{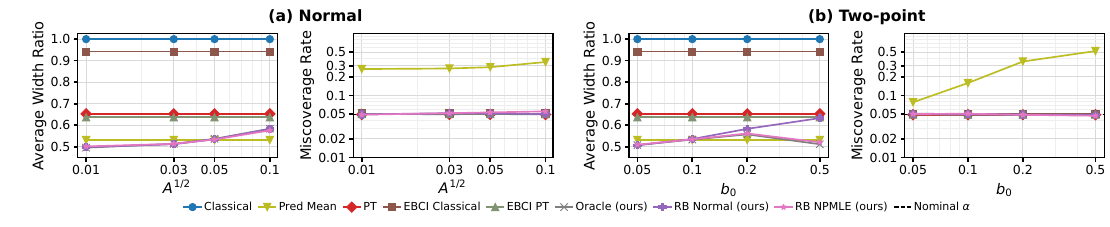}
    \vspace{-1.5em}
    \caption{\textbf{Interval width ratio and miscoverage rate in the synthetic Amazon simulations.} Average width ratio (left) and miscoverage rate (right) for $n=200$ simulated tasks in the synthetic Amazon dataset under two choices of bias distribution $G$: \textbf{(a)} normal and \textbf{(b)} two-point. Across both settings, our rebiased methods yield substantially shorter intervals with near-nominal miscoverage. \textit{RB-Normal} and \textit{RB-NPMLE} are nearly indistinguishable under the well-specified normal bias distribution, while \textit{RB-NPMLE} is more robust under the two-point distribution, where the \textit{RB-Normal} is misspecified and therefore loses efficiency. Standard errors are reported in Tables~\ref{tbl:synthetic_amazon} and~\ref{tbl:synthetic_amazon_twopoint}.}
\label{fig:synthetic_normal}
\end{figure}

\paragraph{Synthetic data using covariance structure in Amazon data.}
To further evaluate our rebiased intervals in a setting with ground truth, we conduct simulation studies that generate synthetic estimators using the estimated task-level second moments from the Amazon data, while independently generating the underlying $\theta_i$ and $b_i$.
We generate biases $b_i$ from two types of bias distribution $G$: normal distribution $\mathrm{N}(-0.1,A)$ and two-point distribution $(\delta_0+\delta_{b_0})/2$, where $\delta_u$ denotes a Dirac point mass at $u$. See details for the choices of $A$ and $b_0$, and the construction of the synthetic dataset in Appendix~\ref{app:experiment-synth}. 
As a benchmark, we also include an \emph{oracle} variant that plugs in
the true data-generating distribution $G$ rather than the estimate
\smash{$\widehat{G}$} when forming rebiased intervals.

Results are displayed in Fig.~\ref{fig:synthetic_normal}. \emph{RB-Normal} and \emph{RB-NPMLE} perform similarly, as well as the oracle interval, when the bias distribution is normal. They maintain close-to-nominal coverage while achieving substantially shorter average widths than the fully debiased \emph{PT} interval. Under this correctly specified setting, the normal model is already sufficient, so the extra flexibility of NPMLE yields no efficiency gain, but no efficiency loss either.
Under the two-point distribution, the effect of model misspecification becomes apparent. Although the \emph{RB-Normal} intervals still maintain approximately nominal coverage, their average widths increase as $b_0$ grows and eventually become close to \emph{PT}. This reflects the efficiency loss from approximating a discrete, nonnormal bias distribution by a normal bias distribution. The \emph{RB-NPMLE} interval is more adaptive in this setting: when $b_0$ is small, the bias distribution is nearly degenerate, and both rebiased methods perform well; as the two-point structure becomes more pronounced, \emph{RB-NPMLE} better accommodates this structure, producing intervals closer in width to the oracle while maintaining coverage.

We additionally vary the dispersion of the primary parameter 
by generating $\theta_i\sim\mathrm{N}(4,\sigma_0^2)$ with different values of $\sigma_0$, while fixing the bias distribution, with results shown in Fig.~\ref{fig:synthetic_theta}. When primary parameters are highly concentrated, EBCI exploits the strong across-task similarity to produce shorter intervals than our rebiasing procedures. As $\theta_i$ dispersion increases, however, this advantage quickly diminishes, with EBCI intervals approaching their non-Bayes counterparts, while the widths and coverage of our procedures remain essentially unchanged. This highlights the complementary regimes of the two approaches: EBCI borrows information across $\theta_i$, whereas our partially Bayes procedure borrows only through the biases and is therefore insensitive to the primary parameter heterogeneity.

\subsection{Family-based GWAS}

\paragraph{Height GWAS.}
We illustrate rebiased p-values for SNP discovery in family-based GWAS. We analyze human height direct-effect estimates \smash{$\unbiased_i$} and population-effect estimates \smash{$\biased_i$} for $n = 572{,}911$ SNPs from~\citet{Guan2025-jq}, estimating the bias distribution $G$ via NPMLE and via the normal parametric model. From each estimated bias distribution, we compute empirical Bayes rebiasing p-values and apply Benjamini-Hochberg (BH) procedure~\citep{BH} at a nominal false discovery rate (FDR) of $0.05$; baselines are BH on the p-values of (i) the unbiased direct-effect \smash{$\unbiased_i$}, and (ii) the biased population-effect \smash{$\biased_i$}. We also compare against a random set of $5{,}000$ SNPs as a negative control. Since we lack individual-level data, we cannot reuse the PPI evaluation strategy. Instead we compare discoveries against a separate GWAS with a similar family design~\citep{Howe2022-fm} as preliminary supporting evidence.

\begin{table}[t]
  \centering
  \caption{\textbf{Summary of discoveries from the family-based GWAS summary statistics.}
We report the numbers of SNPs discovered from BH at FDR $0.05$ applied to our
\textit{RB-NPMLE} and \textit{RB-Normal} p-values, the (unbiased) direct-effect
p-values, and the (biased) population-effect p-values, together with their
overlaps (defined in Appendix~\ref{app:experiment-gwas}) with \citet{Howe2022-fm} signals. We also include overlaps for a random
set of 5,000 SNPs, averaged over 50 repetitions and reported as mean $\pm$
SD.
  }
  \label{tab:discovery-overlap}
  \begin{tabular}{lrrrrr}
  \toprule
   & \multicolumn{1}{c}{\textbf{RB-NPMLE}}
   & \multicolumn{1}{c}{\textbf{RB-Normal}}
   & \multicolumn{1}{c}{\textbf{Direct-effect}}
   & \multicolumn{1}{c}{\textbf{Population-effect}}
   & \multicolumn{1}{c}{\raisebox{-0.5\height}{\textbf{Random}}} \\
   & \multicolumn{1}{c}{\textbf{(ours)}}
   & \multicolumn{1}{c}{\textbf{(ours)}}
   & \multicolumn{1}{c}{\textbf{(unbiased)}}
   & \multicolumn{1}{c}{\textbf{(biased)}}
   & \\
  \midrule
  \textbf{\# SNPs} & $1{,}547$ & $1{,}700$ & $272$ & $1{,}594$ & $5{,}000$ \\
  \textbf{\# overlaps} & $731$ & $743$ & $234$ & $658$ & $10.84{\scriptstyle\,\pm\,2.37}$ \\ 
  \bottomrule
\end{tabular}
\end{table}

A priori, one might trust the population-effect estimator when confounding is believed to be small and otherwise fall back to the much noisier direct-effect estimator. Rebiasing automates this choice: both estimated bias distributions concentrate near zero (Fig.~\ref{fig:bias_fit}c), consistent with existing biological evidence specific to height (Appendix~\ref{app:experiment-gwas}), and the rebiased estimators are similar (but not identical) to the population-effect estimator.  Table~\ref{tab:discovery-overlap} shows that rebiasing recovers more~\citet{Howe2022-fm} signals than the baselines, while the random control shows essentially no overlap. Notably, RB-NPMLE makes slightly fewer total discoveries than the population-effect baseline, yet has more overlaps (see Appendix~\ref{app:experiment-gwas} for further discussion). 

\paragraph{Synthetic GWAS study.}
In Appendix~\ref{app:experiment-ld-simulation}, we conduct a simulation study
to check type-I error control of our proposed methods in a setting where we know the ground truth. In our simulation design, we mimic realistic genomic dependence structure and consider a setting with zero direct genetic effects and weak parental effects. We find that RB-NPMLE p-values used alongside BH are calibrated, while RB-Normal has slightly inflated type-I error.

\section{Conclusion}
\label{sec:conclusion}
For point estimation, forfeiting unbiasedness in exchange for variance savings is an uncontroversial tradeoff. For interval inference, it is common for the analyst to trust a biased estimator outright, hoping the bias is negligible, or to fully debias it and absorb the associated variance cost. This paper proposes empirical Bayes rebiasing as a way to navigate the analogous tradeoff for calibrated inference, by modeling task-level bias as drawn from a distribution learned across many related tasks rather than treated as either zero or arbitrary. 
Our framework applies wherever a biased estimator is paired with a bias correction (equivalently, an unbiased estimator) across many parallel tasks; PPI and family-based GWAS as in this paper, or combining observational with experimental causal estimates, and so forth. Calibrated inference in empirical Bayes problems is much less developed than the corresponding point estimation theory despite its practical importance, see e.g.,~\citet{armstrong2022robust} for recent progress; the empirical partially Bayes formulation pursued here seems a promising avenue in this regard.

\begin{ack}
Part of the computing for this project was conducted on UChicago's Data Science Institute cluster. N.I.\ gratefully acknowledges support from NSF (DMS 2443410). The authors declare no conflicts of interest.
\end{ack}

\bibliographystyle{abbrvnat}
\bibliography{impartially_bayes}


\appendix
\newpage
\renewcommand{\thesection}{\Alph{section}}

{\centering
  {\Large\bfseries$\clubsuit$\enspace Appendix: Table of Contents}\par
}
\vspace{1em}
\apptocentry{app:detail_method}{Details on methods (Section~\ref{sec:setting})}
\apptocentry{app:detail_applications}{Details on applications (Section~\ref{sec:applications})}
\apptocentry{app:robustness}{Robustness results (Section~\ref{sec:theory})}
\apptocentry{app:proof}{Proof of theoretical results (Section~\ref{sec:theory})}
\apptocentry{app:experiment}{Details on numerical studies (Section~\ref{sec:numerical})}
\apptocentry{app:robustness_diagnostic}{Diagnostics for the theoretical assumptions}

\setcounter{equation}{0}
\setcounter{figure}{0}
\setcounter{table}{0}

\renewcommand{\theequation}{S\arabic{equation}}
\renewcommand{\thefigure}{S\arabic{figure}}
\renewcommand{\thetable}{S\arabic{table}}

\section{Details on methods (Section~\ref{sec:setting})}
\label{app:detail_method}
\subsection{Normal bias distribution when \texorpdfstring{$\rho=0$}{rho}}
\label{app:normal}
Consider the parametric case where we are willing to assume that the bias distribution is normal,
$$
b_i\sim \mathrm{N}(\mu,A).
$$
Combining with model \eqref{eq:modeling}, Bayes' rule gives
$$
 b_i\mid \hat b_i\sim \mathrm{N}\left(\mu+\frac{A}{A+\tau_i^2}(\hat b_i-\mu),\frac{A\tau_i^2}{A+\tau_i^2}\right),
$$
and the oracle rebiased estimator takes the form
$$
\rebiased
    =\biased_i-\left\{\mu+\frac{A}{A+\tau_i^2}(\hat b_i-\mu)\right\}
    = \hat\theta_i^{\mathrm{db}}+\frac{\tau_i^2}{A+\tau_i^2}(\hat b_i-\mu).
$$
This expression shows explicitly how the normal model interpolates between the
biased and fully debiased estimators. When $A=0$, the correction is shrunk completely toward $\mu$;
in particular, if $\mu=0$, \smash{$\rebiased=\hat\theta_i^{\mathrm b}$}, while for $A=\infty$, \smash{$\rebiased=\hat\theta_i^{\mathrm{db}}$}. 

Moreover, the oracle rebiased $(1-\alpha)$-interval
reduces to the symmetric interval
\begin{equation}
\label{eq:CI_normal}
     \mathcal{I}^{\rb}_{G,i}(1-\alpha)
    =
    \hat\theta_i^{\mathrm{rb}}
    \pm
    z_{\alpha}
    \left(
        \sigma_i^2+\frac{A\tau_i^2}{A+\tau_i^2}
    \right)^{1/2}.
\end{equation}
The unknown parameters $(\mu, A)$ from $G$ can be estimated via maximum marginal likelihood (see \ref{app:computation}). With \smash{$\widehat{G}= \mathrm{N}(\widehat{\mu},\widehat{A})$} in hand, we can then plug in \eqref{eq:CI_normal} and obtain the rebiased $(1-\alpha)$-intervals $\mathcal{I}^{\rb}_{\widehat{G},i}(1-\alpha)$.

Similarly, we have the following form of the oracle rebiased p-value $P_i^{\rb}=P_{G,\theta_{i0}}^{(i)}(\hat\theta_i^{\db},\hat{b}_i)$, with
\begin{equation}
\label{eq:normal_p_value}
P_{G,\theta_{i0}}^{(i)}(z,l)=
2\Phi\left(-
\left|z+\frac{\tau_i^2}{A+\tau_i^2}(l-\mu)-\theta_{i0}\right|
\bigg/
\left(\sigma_i^2+\frac{A\tau_i^2}{A+\tau_i^2}\right)^{1/2}
\right),
\end{equation}
where $\Phi$ denotes the cumulative distribution function of a standard normal random variable. Plugging the estimated $\widehat{G}$, we can obtain the empirical Bayes rebiased p-value $\widehat{P}_i^{\rb}=P_{\widehat{G},\theta_{i0}}^{(i)}(\hat\theta_i^{\db},\hat{b}_i)$ by replacing $(\mu,A)$ with \smash{$(\wh \mu,\wh A)$} in \eqref{eq:normal_p_value}.
\subsection{Rebiasing for general \texorpdfstring{$\rho$}{rho}}
\label{app:generalization}
Consider a more generic setting of this problem. We now allow nonzero conditional correlation $\rho_i$ between $(\biased_i,\hat{b}_i)$ given $b_i$ in \eqref{eq:modeling}. We can consider the following conditional distribution of $\hat{\theta}_i^{\db}$ given $\hat{b}_i$ that integrates out $b_i$ over its posterior $\Pi_i(\cdot \mid \hat{b}_i)$ under $G$,
\begin{equation}
\label{eq:db_cond_dist}
    \hat{\theta}_i^{\mathrm{db}}\mid \theta_i,\hat b_i
    \sim
    \int
    \mathrm{N}\left(
    \theta_i+c_i(\hat b_i-b),\;
    \sigma_i^2(1-\rho_i^2)
    \right)
    \,\Pi_i(\dd b\mid\hat{b}_i),
\end{equation}
with $c_i=\rho_i\sigma_i/\tau_i-1$. Writing $m_i\equiv m_i(\hat b_i)=\mathbb E_G[b_i\mid \hat b_i]$, the conditional mean of
$\hat\theta_i^{\mathrm{db}}$ is $\mathbb E[\hat\theta_i^{\mathrm{db}}\mid \theta_i,\hat b_i]
=\theta_i+c_i(\hat b_i-m_i)$, and we construct
$$
\hat\theta_i^{\rb} = \hat\theta_i^{\mathrm{db}}
- c_i(\hat b_i-m_i).
$$
Let $q_{G,i,\alpha}(\hat{b}_i)$ denote the
$\alpha$-quantile of the conditional distribution of
$\hat\theta_i^{\db}-\theta_i$ given $\hat b_i$, and we could report the $(1-\alpha)$-level oracle rebiased interval $ \mathcal{I}^{\rb}_{G,i}(1-\alpha)=[\hat\theta_i^{\db}-q_{G,i,1-\alpha/2},\,\hat\theta_i^{\db}-q_{G,i,\alpha/2}]$.

In particular, if we assume $G=\mathrm{N}(\mu,A)$, then \eqref{eq:db_cond_dist} can be written as
$$
\hat{\theta}_i^{\mathrm{db}}\mid \theta_i,\hat b_i
\sim\mathrm N\left(\theta_i+\frac{\rho_i\sigma_i\tau_i-\tau_i^2}{\tau_i^2+A}(\hat b_i-\mu),
\,\sigma_i^2+\tau_i^2-2\rho_i\sigma_i\tau_i-\frac{(\rho_i\sigma_i\tau_i-\tau_i^2)^2}{\tau_i^2+A}
\right).
$$
After estimating the parameters $(\widehat{\mu},\widehat{A})$ via maximum marginal likelihood, we can plug back and form the following rebiased $(1-\alpha)$-intervals for each $i$:
$$
\hat{\theta}_i^{\rb} \pm z_{\alpha }\left[\sigma_i^2+\tau_i^2-2 \rho_i \sigma_i \tau_i-\frac{\left(\rho_i \sigma_i \tau_i-\tau_i^2\right)^2}{\tau_i^2+\widehat{A}}\right]^{1 / 2},\quad \hat\theta_i^{\rb}
:=
\hat\theta_i^{\mathrm{db}}
-
\frac{
\rho_i\sigma_i\tau_i-\tau_i^2
}{
\tau_i^2+\widehat A
}
(\hat b_i-\widehat\mu).
$$
\subsection{Estimation for \texorpdfstring{$G$}{G}}
\label{app:computation}
\paragraph{Parametric Normal.} Assume $G=\mathrm{N}(\mu,A)$, and we use the maximum marginal likelihood to estimate $(\mu,A)$. In particular, combining with the model \eqref{eq:modeling}, the marginal distribution of \smash{$\hat b_i$} is
$$
\hat b_i \sim \mathrm{N}(\mu,A+\tau_i^2).
$$
Therefore, we define the maximum marginal-likelihood estimator by
$$
(\widehat\mu,\widehat A)
=
\arg\max_{\mu\in\mathbb{R},\,A\geq0} \sum_{i=1}^n
\log \varphi\!\left(\hat b_i;\,\mu,\,A+\tau_i^2\right).
$$
For implementation, we parameterize $A=\exp(\eta)$ and jointly optimize
over $(\mu,\eta)\in\mathbb{R}^2$ using BFGS. The numerical optimization
therefore searches over the interior $A>0$. This is equivalent to the constrained optimization whenever the maximizer satisfies $A>0$, and approximates the boundary $A=0$ as $\eta\to-\infty$.

\paragraph{NPMLE.} As observed before, the optimization program in \eqref{eq:npmle_G} is convex. We use the discretization technique proposed in \cite{koenker2014convex} to solve this problem. In particular, we choose $B=50$ grid points (For GWAS we choose $B=300$) that are equally spaced between the smallest and largest value of \smash{$\{\hat{b}_1,...,\hat{b}_n\}$}, and we optimize \eqref{eq:npmle_G} over all possible distributions supported on this finite grid, which is a conic programming problem and we solve it with the interior point convex programming solver MOSEK \citep{aps2024mosek}. 

\section{Details on applications (Section~\ref{sec:applications})}
\label{app:detail_applications}
\subsection{Prediction powered inference}
\label{subsec:ppi_app_details}

Here we expand on our description of the PPI/PT estimators in \S\ref{sec:applications}, using the per-task quantities and the vanilla PPI estimator already introduced in~\eqref{eq:ppi}. Throughout this subsection, we assume that the labeled samples \smash{$\{(X_{ij}, Y_{ij})\}_{j=1}^{m_i}$} and the unlabeled samples \smash{$\{\tilde X_{ij}\}_{j=1}^{M_i}$} are jointly independent both across the labeled/unlabeled split and across tasks $i = 1, \ldots, n$. The ML predictor $h$ is considered fixed and independent of all of them. To compress notation, write the second moments for the samples in the $i$-th task as
\begin{equation}
\label{eq:ppi_moments}
v_i^2 \,:=\, \VarInline{h(X_{ij})},
\qquad
w_i^2 \,:=\, \VarInline{Y_{ij}},
\qquad
c_i \,:=\, \Cov{h(X_{ij}),\, Y_{ij}},
\end{equation}
and we denote $\mu_i := \EEInline{h(X_{ij})}$ so the true bias is $b_i = \mu_i - \theta_i$.

\paragraph{Power-tuning derivation.} \citet{angelopoulos2024ppi} introduce a one-parameter family of estimators with a tuning parameter $\lambda_i \in \mathbb{R}$,
\begin{equation}
\label{eq:pt_family}
\hat\theta_{i,\lambda_i} \,=\, \bar Y_i \,+\, \lambda_i\big(\tilde Z_i^h - \bar Z_i^h\big),
\qquad \lambda_i \in \RR,
\end{equation}
which interpolates between the classical estimator $\bar{Y}_i$ ($\lambda_i = 0$) and $\hthetaPPI$ ($\lambda_i = 1$). For every choice of $\lambda_i$, $\hat\theta_{i,\lambda_i}$ is unbiased since $\EEInline{\tilde Z_i^h} = \EEInline{\bar Z_i^h} = \mu_i$. A direct computation gives
\begin{equation}
\label{eq:pt_variance}
\VarInline{\hat\theta_{i,\lambda_i}}
\,=\, \frac{w_i^2}{m_i}
\,+\, \lambda_i^2\, v_i^2\!\left(\frac{1}{M_i} + \frac{1}{m_i}\right)
\,-\, \frac{2\lambda_i\, c_i}{m_i}.
\end{equation}
Minimizing this variance with respect to $\lambda_i$ yields the optimal power tuning parameter
\begin{equation}
\label{eq:pt_lambda_star}
\lambda_i^* \,=\, \frac{M_i}{m_i + M_i}\,\frac{c_i}{v_i^2},
\end{equation}
and the corresponding power-tuned (PT) estimator is defined as $\hthetaPT := \hat\theta_{i,\lambda_i^*}$.

\paragraph{PT in the rebiasing framework.} The PT estimator is unbiased for $\theta_i$, but its variance \eqref{eq:pt_variance} at $\lambda_i^*$ is generally larger than $v_i^2/M_i$, the variance of $\hthetaML$. Following \S\ref{sec:applications}, we therefore identify
\begin{equation}
\label{eq:ppi_b_bhat_pt}
\biased_i \,:=\, \hthetaML \,=\, \tilde Z_i^h,
\qquad
\hat b_i \,:=\, \biased_i - \hthetaPT \,=\, (1-\lambda_i^*)\,\tilde Z_i^h \,+\, \lambda_i^*\,\bar Z_i^h \,-\, \bar Y_i.
\end{equation}
Then $\EEInline{\hat b_i} = (1-\lambda_i^*)\mu_i + \lambda_i^*\mu_i - \theta_i = b_i$, so $\hat b_i$ is unbiased for the bias of $\biased_i$. The pair $(\biased_i, \hat b_i)$ thus fits the model in~\eqref{eq:modeling} with $\debiased_i = \biased_i - \hat b_i = \hthetaPT$, and the rebiased estimator and intervals from \S\ref{sec:setting} apply directly. Setting $\lambda_i^* = 1$ recovers the standard vanilla-PPI debiasing ($\hat b_i = \bar Z_i^h - \bar Y_i$); the rebiasing perspective therefore subsumes both PPI and PT under a single framework.

\paragraph{Expressions for $\sigma_i^2, \tau_i^2, \rho_i$.} A short computation gives
\begin{align}
\label{eq:ppi_sigma}
\sigma_i^2 \,&=\, \VarInline{\biased_i} \,=\, \frac{v_i^2}{M_i}, \\
\label{eq:ppi_tau}
\tau_i^2 \,&=\, \VarInline{\hat b_i}
\,=\, (1-\lambda_i^*)^2\,\frac{v_i^2}{M_i}
\,+\, \frac{(\lambda_i^*)^2\, v_i^2 \,-\, 2\lambda_i^*\, c_i \,+\, w_i^2}{m_i}, \\
\label{eq:ppi_rho}
\rho_i\sigma_i\tau_i \,&=\, \Cov{\biased_i,\, \hat b_i}
\,=\, (1-\lambda_i^*)\,\frac{v_i^2}{M_i},
\qquad
\rho_i \,=\, \frac{(1-\lambda_i^*)\, v_i / \sqrt{M_i}}{\tau_i}.
\end{align}
When $\lambda_i^* = 0$, $\hthetaPT$ collapses to the classical mean $\bar Y_i$. When $\lambda_i^* = 1$, $\hthetaPT$ coincides with vanilla PPI and $\hat b_i = \bar Z_i^h - \bar Y_i$; the unlabeled-only $\biased_i$ and the labeled-only $\hat b_i$ are then independent, giving $\rho_i = 0$ and matching the simplified development in \S\ref{sec:setting}.

\paragraph{Plug-in estimation.} The expressions above involve the unknown second moments $(v_i^2, w_i^2, c_i)$. For the tuning coefficient in
\eqref{eq:pt_lambda_star}, we replace $c_i$ by the sample covariance between $Y_{ij}$ and $h(X_{ij})$ in the labeled sample, and $v_i^2$ by the sample variance of the predictions pooled across the labeled and unlabeled samples. We clip the resulting plug-in estimate to $[0,1]$ to obtain the feasible coefficient $\hat{\lambda}_i^*$. For estimating the covariance parameters in \eqref{eq:modeling}, we estimate the terms $v_i^2/M_i$ and $v_i^2/m_i$ using the unlabeled- and labeled-sample prediction variances, respectively, while $w_i^2$ and $c_i$ are estimated from the labeled sample. Substituting these estimates into~\eqref{eq:ppi_sigma}--\eqref{eq:ppi_rho} yields $(\hat\sigma_i^2, \hat\tau_i^2, \hat\rho_i)$. Following standard practice in PPI~\citep{angelopoulos2024ppi}, we treat these plug-ins as known when forming Wald and rebiased intervals; see Appendix~\ref{app:robustness_diagnostic} for a discussion of whether this plug-in approximation is reasonable in each application.

\subsection{Active statistical inference}
\label{subsec:active_app_details}

A closely related setting to PPI is active statistical inference, which aims to improve statistical efficiency under a limited labeling budget by selectively acquiring informative labels. \citet{zrnic2024active} consider a setting where one starts with unlabeled covariates $\{X_{ij}\}_{j=1}^{K_i}$, then based on $X_{ij}$, the analyst observes $Y_{ij}$ with probability $\pi(X_{ij})$. Effectively, one stochastically prioritizes measuring labels that will be more informative. Let $\xi_{ij} \in \{0,\,1\}$ be the indicator of whether $Y_{ij}$ is observed. Given a further predictive model $h(\cdot)$ as in PPI, their batch active mean estimator takes the following form:
$$
\widehat{\theta}_i^{\mathrm{AS}} = \underbrace{\frac{1}{K_i}\sum_{j=1}^{K_i} h(X_{ij})}_{\biased_i} \,+\, \underbrace{\frac{1}{K_i}\sum_{j=1}^{K_i} (Y_{ij} - h(X_{ij})) \frac{\xi_{ij}}{\pi(X_{ij})}}_{-\hat{b}_i}.
$$
The first term is a prediction-based, potentially biased estimator, while the second estimates its prediction bias correction. Thus, the batch active mean estimator can naturally be incorporated into our empirical Bayes rebiasing framework to borrow information across related tasks.

Recent work has further explored robust and balanced sampling strategies to improve the statistical efficiency of active inference~\citep{li2026robust,chen2026balanced}. More generally, combining active inference with empirical Bayes rebiasing raises additional challenges. In particular, with a limited labeling budget across many related tasks, one needs to jointly determine which observations to label within each task and how to allocate the budget across tasks. Developing such an adaptive labeling and budget-allocation strategy together with our empirical Bayes framework represents an interesting direction for future work.

\subsection{Family-based GWAS}
\label{subsec:gwas_app_details}

The summary statistics provide the regression coefficient estimates $(\hat{\theta}^{\mathrm{ub}}_i,\hat{b}_i)$, with variances $\tilde{\sigma}_i^2$ and $\tau_i^2$ and correlation $\gamma_i$ obtained from the estimated covariance matrix of the coefficient estimates. The representation above is equivalent to the pair \smash{$(\biased_i, \unbiased_i)$}: indeed, the biased estimator can be written as $\smash{\biased_i=\unbiased_i+\hat{b}_i}$, with variance $\sigma^2_i=\smash{\tilde{\sigma}_i^2+\tau_i^2+2\gamma_i\tilde{\sigma}_i\tau_i}$, and its
correlation with \smash{$\hat{b}_i$} is
\smash{$\rho_i=(\tau_i+\gamma_i\tilde{\sigma}_i)/
(\tilde{\sigma}_i^2+\tau_i^2+2\gamma_i\tilde{\sigma}_i\tau_i)^{1/2}$}.

\section{Robustness results (Section~\ref{sec:theory})}\label{app:robustness}

We extend the coverage guarantees of Theorems~\ref{thm:rate_coverage}, \ref{thm:rate_coverage_joint}, and \ref{thm:rate_average_coverage} to allow for approximate normality, estimated second moments, and weak dependence across tasks. Following the robustness framework of \citet{chen2026normal}, throughout this section we impose the bounded-support condition \(b_i\in[-M,M]\); in the exchangeable setting this is encoded by compactly supported bias distributions, \(G\in\mathcal P([-M,M])\) (for the robustness extensions of Theorems~\ref{thm:rate_coverage} and \ref{thm:rate_coverage_joint}).

Throughout this section, we work in a triangular-array asymptotic framework, allowing the task-specific parameters, sampling laws, and common bias distribution $G$ (in the partially Bayes and hierarchical settings) to depend on $n$, while maintaining $G\in\mathcal P([-M,M])$ for a fixed $M<\infty$. 
More precisely, for each $n$, we have $n$ tasks with primary parameters $\theta_{1,n}, \ldots, \theta_{n,n}$, biased estimators \smash{$\biased_{1,n},\ldots,\biased_{n,n}$}, biases $b_{1,n}, \ldots, b_{n,n}$  (that are drawn from a distribution \smash{$G^{(n)}$} in the partially Bayes and hierarchical settings), bias estimators \smash{$\hat{b}_{1,n},\ldots, \hat{b}_{n,n}$}, and true and estimated second moments. In words, for each $n$, we imagine that we receive $n$ new tasks.  
This allows us to study sequences of models in which departures from normality diminish on average across tasks. We refer to \citet{chen2026normal} for detailed discussion and for examples where these asymptotics are applicable with a particular emphasis on asymptotic regimes that both make the normality assumption more plausible and still benefit from empirical Bayes shrinkage.
We suppress the dependence on $n$ in the notation and e.g., write $\theta_1,\ldots,\theta_n$ instead of $\theta_{1,n}, \ldots, \theta_{n,n}$.

To formulate the robustness conditions uniformly across three settings, we use $\mathbb P_{\star}$, $\mathbb E_{\star}$, and $\mathcal L_{\star}$ to denote probability, expectation, and distribution under the sampling law specified in each theorem below. For each task $i$, let
$$
\Sigma_i:=\begin{pmatrix}
\tilde\sigma_i^2 & \gamma_i\tilde\sigma_i\tau_i\\
\gamma_i\tilde\sigma_i\tau_i & \tau_i^2
\end{pmatrix}
$$
denote the reference sampling covariance matrix for $(\hat{\theta}_i^{\db}-\theta_i,\hat{b}_i-b_i)$, and let $\widehat\Sigma_i$ be its estimator. We take $\widehat G$ to be a global maximizer of the normal working marginal likelihood over $\mathcal P([-M,M])$, using the estimated second moments $\widehat\Sigma_{i,22}$. 

To quantify departures from the normal working model underlying Theorems~\ref{thm:rate_coverage} to \ref{thm:rate_average_coverage}, we compare the actual bivariate experiment of $(\hat{\theta}_i^{\db}-\theta_i,\hat{b}_i)$ with its normal counterpart, $\mathrm{N}_2((0,b_i)^\top,\Sigma_i)$, through the averaged Wasserstein-2 discrepancy
$$
\mathcal W_n^2:=\frac{1}{n}\sum_{i=1}^n\EE[\star]{W_2^2\left(\mathcal L_{\star}\left(\left.
\begin{pmatrix}
\hat\theta_i^{\db}-\theta_i\\
\hat b_i
\end{pmatrix}\right|b_i\right),\mathrm N_2\left(
\begin{pmatrix}0\\b_i\end{pmatrix},\Sigma_i
\right)\right)}.
$$
where $W_2$ denotes the Wasserstein-$2$ distance on $\mathbb{R}^2$.

\begin{assumption}[Approximate normality]
\label{assm:robust_normality}
We assume that $\mathcal W_n\to0$.
\end{assumption}

\begin{assumption}[Estimated second moments]
\label{assm:robust_second_moment}
We assume that
$$
\frac1n\sum_{i=1}^n
\|\widehat\Sigma_i-\Sigma_i\|
\xrightarrow{\mathbb P_{\star}}0.
$$
We further assume that there exist constants $0<\underline\lambda<\overline\lambda<\infty$ such that $\underline\lambda I_2 \preceq\Sigma_i,\widehat{\Sigma}_i\preceq\overline\lambda I_2$ for all $1\leq i \leq n$.
\end{assumption}
The bounds on $\widehat\Sigma_i$ are technical regularity conditions and can, in principle, be enforced by eigenvalue clipping. Our numerical implementation uses the unmodified plug-in estimates. 

Finally, we replace cross-task independence with the following weak empirical averaging condition.
\begin{assumption}[Approximate independence]
\label{assm:robust_independence}
For every fixed triangular sequence of functions $f_i:\mathbb R\to\mathbb R$, $1\le i\le n$, such that $|f_i'(0)|\le1$ and $\max_i\|f_i''\|_\infty\le1$, we assume that
$$
\frac{1}{n}\sum_{i=1}^n
\left(f_i(\hat b_i)-\mathbb{E}_\star[f_i(\hat b_i)]\right)
\xrightarrow{\mathbb P_{\star}}0.
$$
\end{assumption}

We first consider the partially Bayes setting of Theorem~\ref{thm:rate_coverage}, where $\theta_i$ is fixed and $b_i\simiid G$.

\begin{theorem}[Robustness extension for Theorem~\ref{thm:rate_coverage}]\label{thm:robust_w_exchangeability}
Suppose that Assumptions~\ref{assm:robust_normality}, \ref{assm:robust_second_moment}, and~\ref{assm:robust_independence} hold with $\mathbb P_{\star}=\mathbb P_{\boldsymbol\theta,G}$, $\mathbb E_{\star}=\mathbb E_{\boldsymbol\theta,G}$, and $\mathcal L_{\star}=\mathcal L_{\boldsymbol\theta,G}$. Then for every fixed $\alpha\in(0,1)$,
$$
\lim_{n\to\infty}\frac{1}{n}\sum_{i=1}^n
\left|\mathbb P_{\boldsymbol{\theta},G}\left[\theta_i\in\mathcal I_{\widehat G,i}^{\mathrm{rb}}(1-\alpha)\right]-(1-\alpha)\right|=0.
$$
\end{theorem}
We next consider the hierarchical setting of Theorem~\ref{thm:rate_coverage_joint}, where  $(\theta_i,b_i)\sim H_i$, with a common marginal bias distribution $b_i\sim G$. In the definition of $\mathcal W_n$, the conditional law marginalizes over $\theta_i$ given $b_i$, while the outer expectation is taken over $b_i\sim G$, thereby allowing dependence between the primary parameter and the latent bias.
\begin{theorem}[Robustness extension for Theorem~\ref{thm:rate_coverage_joint}]\label{thm:robust_joint}
Suppose that Assumptions~\ref{assm:robust_normality}, \ref{assm:robust_second_moment}, and~\ref{assm:robust_independence} hold with $\mathbb P_{\star}=\mathbb P_{\boldsymbol H}$, $\mathbb E_{\star}=\mathbb E_{\boldsymbol H}$, and $\mathcal L_{\star}=\mathcal L_{\boldsymbol H}$. Then for every fixed $\alpha\in(0,1)$,
$$
\lim_{n\to\infty}\frac{1}{n}\sum_{i=1}^n
\left|\mathbb P_{\boldsymbol{H}}\left[\theta_i\in\mathcal I_{\widehat G,i}^{\mathrm{rb}}(1-\alpha)\right]-(1-\alpha)\right|=0.
$$
\end{theorem}

Finally, we consider the compound decision setting of Theorem~\ref{thm:rate_average_coverage}, where $(\theta_i,b_i)$ are fixed, $b_i\in[-M,M]$, and $\Sigma_i=\Sigma$ for all $i$. Since $b_i$ is fixed, the conditional law in $\mathcal{W}_n$ reduces to the unconditional law, and the outer expectation can be omitted.
\begin{theorem}[Robustness extension for Theorem~\ref{thm:rate_average_coverage}]
\label{thm:robust_compound}
Suppose that Assumptions~\ref{assm:robust_normality}, \ref{assm:robust_second_moment}, and~\ref{assm:robust_independence} hold with $\mathbb P_{\star}=\mathbb P_{\boldsymbol\theta,\boldsymbol b}$, $\mathbb E_{\star}=\mathbb E_{\boldsymbol\theta,\boldsymbol b}$, and $\mathcal L_{\star}=\mathcal L_{\boldsymbol\theta,\boldsymbol b}$. Then, for every fixed $\alpha\in(0,1)$,
$$
\lim_{n\to\infty}\frac1n\sum_{i=1}^n\mathbb P_{\boldsymbol{\theta},\boldsymbol{b}}\!\left[
\theta_i\in\mathcal I^{\rm rb}_{\widehat G,i}(1-\alpha)
\right]= 1-\alpha.
$$
\end{theorem}

\section{Proof of theoretical results (Section~\ref{sec:theory})}\label{app:proof}

For notation simplicity, we denote 
$$
\Sigma_{i} = (\Sigma_{i,jq}) := \begin{pmatrix}
\tilde{\sigma}_i^2&\gamma_i\tilde{\sigma}_i\tau_i
\\
\gamma_i\tilde{\sigma}_i\tau_i&\tau_i^2
\end{pmatrix}, \quad \Omega_{i} =(\Omega_{i,jq}):= \Sigma_i^{-1}.
$$
Start by defining the class of marginal densities of $\hat{b}_i$ across $i$,
\begin{equation}
\label{def:hetero_class_mixture_density}
    \mathcal{F}_n :=\left\{\left(f_{G^{\prime}}^{(1)},...,f_{G^{\prime}}^{(n)}\right): G^{\prime}\in\mathcal{G}\right\},
\end{equation}
where $\mathcal{G}$ is the set of all possible distribution $G^{\prime}$, and
$$
f_{G^{\prime}}^{(i)}(l) \equiv f_{G^{\prime}}(l;\Sigma_{i,22}) := \int \varphi(l- b;\Sigma_{i,22})~\dd G^{\prime}(b).
$$
Define the supremum norm in bounded intervals 
$$
\lVert\boldsymbol{h}\rVert_{\infty,U} =\max_{1\leq i\leq n}\lVert h_i\rVert_{\infty,U} = \max_{1\leq i\leq n}\sup_{|x|\leq U}|h_i(x)|,
$$
where $\boldsymbol{h}=(h_1,...,h_n)$, and also define the average Hellinger distance
$$
\bar{d}(G,\widehat{G}) = \left(\frac{1}{n}\sum_{i=1}^n\mathfrak{H}^2\left(f_{G}^{(i)},f_{\widehat{G}}^{(i)}\right)\right)^{\frac{1}{2}}, \quad \text{where }\mathfrak{H}^2(g,h) = \frac{1}{2}\int_{}\left(\sqrt{g(t)}-\sqrt{h(t)}\right)^2 \dd t.
$$
Throughout the proofs in Appendix~\ref{app:preliminary_lemma}--\ref{app:proof_compound}, we may assume without loss of generality that $\mathbb E_G [b]=0$. To see this, let $m=\mathbb E_G [b]$ and define
$$
b_i'=b_i-m,\qquad
\hat b_i'=\hat b_i-m,\qquad
(\hat\theta_i^{\mathrm b})'=\hat\theta_i^{\mathrm b}-m.
$$
Let $G'$ denote the distribution of $b_i'$. Then $G'$ remains in $\mathcal G_\Gamma$ and satisfies $\mathbb E_{G'}[b]=0$. The debiased estimator $\hat\theta_i^{\db}$ and the covariance parameters are unchanged under this translation. Moreover, for any mixing distribution $H$, letting $H'$ denote its translation by $-m$, we have
$$
f_{H'}^{(i)}(l-m)=f_H^{(i)}(l),
\qquad
F_{H',i}(z\mid l-m)=F_{H,i}(z\mid l).
$$
Consequently, the marginal likelihood and its maximizers transform accordingly, while the conditional CDFs, resulting p-values, and relevant Hellinger distances are invariant under the corresponding translations. Since any selected NPMLE for the original data corresponds to a valid maximizer for the centered data, it suffices to establish the following bounds under $\mathbb E_G[b]=0$.
\subsection{Preliminary lemmata and propositions}
\label{app:preliminary_lemma}
We state the following results from \cite{jiang2020general}.
\begin{lemm}[Lemma 4 of \cite{jiang2020general}]
\label{lemm:hetero_cardinality_cover}
    For $0<\varepsilon<1/\sqrt{2\pi}$, $U>0$,
    $$
    \log{(\mathcal{N}(\varepsilon,\mathcal{F}_n,\lVert\cdot\rVert_{\infty,U}))} \lesssim_{\underline{\sigma},\bar{\sigma}} \left(\log{\left(\frac{1}{\varepsilon}\right)}\right)^2\max\left\{\frac{U}{\sqrt{|\log \varepsilon|}},1\right\}.
    $$
    Here, $\mathcal{N}(\varepsilon,\mathcal F,d)$ denotes the $\varepsilon$-covering number of $\mathcal F$ with respect to distance metric $d$.
\end{lemm}
\begin{lemm}[Theorem 4 of \cite{jiang2020general}]
\label{lemm:hetero_hellienger}
    Under Assumption~\ref{assm:subgaussian}, fix $c_0\geq 2$, then there exists a constant $C_0$ depending only on $c_0,\Gamma,\underline{\sigma},\bar{\sigma}$ such that 
    $$
    \PP[G]{\bar{d}(G,\widehat{G})\geq C_0\frac{\log{n}}{\sqrt{n}}}\leq \exp{\left(-c_0\log{n}\right)} \quad\text{for all } n \in \mathbb{N}_{\geq2}.
    $$
\end{lemm}
\begin{proof}[Proof of Lemma~\ref{lemm:hetero_hellienger}]
We verify the rate obtained from Theorem 4 of \cite{jiang2020general}. In the notation of that theorem, the Hellinger rate is of the form
\begin{equation}
\label{eq:hellinger_rate}
    \varepsilon_n(n,G,p)
=
\max\left\{
\sqrt{2\log n},
\left[
n^{1/p}\sqrt{\log n}\cdot \mu_p(G)
\right]^{p/(2p+2)}
\right\}
\sqrt{\frac{\log n}{n}},
\end{equation}
where $G$ is the true bias distribution, and $\mu_p(G):=\left(\int |b|^p\,\dd G(b)\right)^{1/p}$. Since \(G\in\mathcal G_\Gamma\), the standard sub-Gaussian moment bound gives $\mu_p(G)\le c\sqrt{\Gamma p}$ for all $p\geq1$, where $c>0$ is a universal constant. 

For $n \geq 3$, we set $p=\log n\geq1$ in \eqref{eq:hellinger_rate}. For $A_n := \big[n^{1/p}\sqrt{\log n}\cdot \mu_p(G)\big]^{p/(2p+2)}$, we have
$$
n^{1/p}\sqrt{\log n}\cdot\mu_p(G)\leq e\cdot\sqrt{\log n}\cdot c\sqrt{\Gamma\log n}=ec\sqrt{\Gamma}\log n,\qquad \frac{p}{2p+2} = \frac{\log n}{2\log n + 2} \le \frac{1}{2}.
$$
Then
\[
A_n\leq
\max\left\{1,\sqrt{ec\sqrt{\Gamma}\log n}\right\}
\leq C_1(\Gamma)\sqrt{\log n},
\]
where one may take $C_1(\Gamma)=\max\left\{1,\sqrt{ec}\,\Gamma^{1/4}\right\}.$
Plugging back to \eqref{eq:hellinger_rate},
$$
\varepsilon_n(n,G,\log n)
\leq
C_2(\Gamma)\frac{\log n}{\sqrt n},\qquad C_2(\Gamma)
=
\max\left\{
\sqrt{2},C_1(\Gamma)
\right\},
\qquad n\geq3
$$
Since $\overline d$ is uniformly bounded, the remaining case $n=2$
can be covered by enlarging the constant.

Therefore, with Theorem 4 of \cite{jiang2020general} and a fixed $c_0\geq2$, there exists $C_0$ depending only on $c_0,\Gamma,\underline{\sigma},\bar{\sigma}$ such that
$$
 \PP[G]{\bar{d}(G,\widehat{G})\geq C_0\frac{\log{n}}{\sqrt{n}}}\leq \exp{\left(-c_0\log{n}\right)} \quad\text{for all } n \in \mathbb{N}_{\geq2}.
$$
\end{proof}
In the following proposition, we consider to express $F_{G,i}(z\mid l)$ in terms of $f_G^{(i)}$.
\begin{proposition}
\label{prop:f_formula}
It holds that
$$
    F_{G,i}(z\mid l)
    =
    \frac{C(\Sigma_i)}{f_G^{(i)}(l)}
    \int_{-\infty}^{z}
    \exp\left(-\frac{u^2}{2\Sigma_{i,11}}\right)
    f_G\left(
    \frac{\Omega_{i,12}}{\Omega_{i,22}}u+l;
    \Omega_{i,22}^{-1}
    \right)\,\dd u,
$$
where $C(\Sigma_i)=(2\pi|\Sigma_i|\Omega_{i,22})^{-1/2}.$
\end{proposition}
\begin{proof}
We start by rewriting the distribution of $(\hat\theta_i^{\mathrm{db}}-\theta_i,\hat b_i)=(z,l)$ given $b_i$,
$$
\begin{aligned}
p_i(z,l\mid b_i) &= (2\pi)^{-1}\lvert\Sigma_i\rvert^{-1/2}\exp{\left(-\frac{1}{2}\left[\Omega_{i,11}z^2+2\Omega_{i,12}z(l-b_i)+\Omega_{i,22}(l-b_i)^2\right]\right)} \\
&= C^{\prime}\exp{\left(-\frac{1}{2}\left[\Omega_{i,22}b_i^2-2(\Omega_{i,12}z+\Omega_{i,22}l)b_i+(\Omega_{i,11}z^2+2\Omega_{i,12}zl+\Omega_{i,22}l^2)\right]\right)} \\
&=  C^{\prime}\exp{\left(-\frac{\Omega_{i,22}}{2}\left(b_i - \frac{\Omega_{i,12}z+\Omega_{i,22}l}{\Omega_{i,22}}\right)^2\right)}\exp{\left(-\frac{1}{2}\left(\Omega_{i,11}-\frac{\Omega_{i,12}^2}{\Omega_{i,22}}\right)z^2\right)} \\
&= C(\Sigma_i)\varphi\left(\frac{\Omega_{i,12}z+\Omega_{i,22}l}{\Omega_{i,22}}-b_i;\Omega_{i,22}^{-1}\right)\exp{\left(-\frac{1}{2}\left(\Omega_{i,11}-\frac{\Omega_{i,12}^2}{\Omega_{i,22}}\right)z^2\right)},
\end{aligned}
$$
with $C(\Sigma_i):=(2\pi|\Sigma_i|\Omega_{i,22})^{-1/2}$. Denote $p_G(u,l)=\int p_i(u,l\mid b_i)\,\dd G(b_i)$ as the joint marginal density of $(\hat\theta_i^{\mathrm{db}}-\theta_i,\hat b_i)$ under $G$, then we can express $F_{G,i}(z\mid l)$ as
$$
\begin{aligned}
F_{G,i}(z\mid l)
&=
\PP[G]{\hat\theta_i^{\mathrm{db}}-\theta_i\le z\mid \hat b_i=l}  \\
&=
\int_{-\infty}^{z}
\frac{p_G(u,l)}{f_G^{(i)}(l)}\,\dd u   \\
&=
\frac{1}{f_G^{(i)}(l)}
\int_{-\infty}^{z}
\int
p_i(u,l\mid b_i)\,\dd G(b_i)\,\dd u \\
&=
\frac{C(\Sigma_i)}{f_G^{(i)}(l)}
\int_{-\infty}^{z}
\exp\left(
-\frac{u^2}{2\Sigma_{i,11}}
\right)
\int
\varphi\left(
\frac{\Omega_{i,12}}{\Omega_{i,22}}u+l
- b_i;\Omega_{i,22}^{-1}\right)
\dd G(b_i)\,\dd u  \\
&=
\frac{C(\Sigma_i)}{f_G^{(i)}(l)}
\int_{-\infty}^{z}
\exp\left(
-\frac{u^2}{2\Sigma_{i,11}}
\right)
f_G\left(
\frac{\Omega_{i,12}}{\Omega_{i,22}}u+l;
\Omega_{i,22}^{-1}
\right)\,\dd u.
\end{aligned}
$$
\end{proof}

\subsection{Proof of Theorem~\ref{thm:covergence_pvalue}}
\label{app:proof_theorem2}
Since the map $x\mapsto 2\min\{x,1-x\}$ is $2$-Lipschitz on $[0,1]$, we have that,
$$
\begin{aligned}
&\left|
P_{G,\theta_{i0}}^{(i)}(\hat\theta_i^{\mathrm{db}},\hat b_i)
-
P_{\widehat G,\theta_{i0}}^{(i)}(\hat\theta_i^{\mathrm{db}},\hat b_i)
\right|  \\
&\qquad =
\left|
2\min\left\{
F_{G,i}(\hat{\theta}_i^{\db}-\theta_{i0}\mid \hat{b}_i),1-F_{G,i}(\hat{\theta}_i^{\db}-\theta_{i0}\mid \hat{b}_i)
\right\}\right.\\
&\qquad\qquad\left.-
2\min\left\{
F_{\widehat G,i}(\hat{\theta}_i^{\db}-\theta_{i0}\mid  \hat{b}_i),1-F_{\widehat G,i}(\hat{\theta}_i^{\db}-\theta_{i0}\mid  \hat{b}_i)
\right\}
\right| \\
&\qquad \le
2\left|
F_{G,i}(\hat{\theta}_i^{\db}-\theta_{i0}\mid  \hat{b}_i)-F_{\widehat G,i}(\hat{\theta}_i^{\db}-\theta_{i0}\mid  \hat{b}_i)
\right|  \\
&\qquad \le
2\sup_{z\in\mathbb R}
\left|
F_{G,i}(z\mid \hat{b}_i)-F_{\widehat G,i}(z\mid  \hat{b}_i)
\right|.
\end{aligned}
$$
Therefore, it suffices to prove
\[
    \frac1n\sum_{i=1}^n
    \EE[G]{
    \sup_{z\in\mathbb R}
    \left|
    F_{G,i}(z\mid  \hat{b}_i)-F_{\widehat G,i}(z\mid  \hat{b}_i)
    \right|
    }
    \lesssim
    \frac{(\log n)^{3/2}}{n^{(1-\check\gamma^2)/2}}.
\]
Let $\mathcal{A}$ be the event with the constant $C_0$ in Lemma~\ref{lemm:hetero_hellienger}:
$$
\mathcal A:=\left\{\bar d(G,\widehat G) < C_0\log n/\sqrt n\right\}.
$$
and for any distribution $\widetilde G$  supported on $\mathbb{R}$, define
$$
\begin{aligned}
    N_{\widetilde G}^{(i)}(z,l)&:=F_{\widetilde G,i}(z\mid l)\cdot f_{\widetilde G}^{(i)}(l), \\
    N_i(z,\widetilde G)&:= N_{\widetilde G}^{(i)}(z,\hat{b}_i),\text{ and }\\
    D_i(\widetilde G)&:=f_{\widetilde G}^{(i)}(\hat{b}_i).
\end{aligned}
$$
By these definitions, it holds that $ F_{\widetilde G,i}(z\mid \hat{b}_i)=N_i(z,\widetilde G)/D_i(\widetilde G)$. Let $ \widehat G_*=(\widehat G+G)/2$, then

$$
\begin{aligned}
&\left|
F_{G,i}(z\mid \hat{b}_i)-F_{\widehat G,i}(z\mid \hat{b}_i)
\right|  \\
&\qquad =\left|\frac{N_i(z,G)}{D_i(G)}
-\frac{N_i(z,\widehat G)}{D_i(\widehat G)}\right|  \\
&\qquad =\left|
\frac{N_i(z,G)}{D_i(G)} -\frac{N_i(z,G)}{D_i(\widehat G_*)}+\frac{N_i(z,G)}{D_i(\widehat G_*)}-\frac{N_i(z,\widehat G)}{D_i(\widehat G_*)}+\frac{N_i(z,\widehat G)}{D_i(\widehat G_*)}-\frac{N_i(z,\widehat G)}{D_i(\widehat G)}\right| \\
&\qquad \le
\frac{N_i(z,G)}{D_i(G)}\frac{|D_i(G)-D_i(\widehat G_*)|}{D_i(\widehat G_*)}+\frac{|N_i(z,G)-N_i(z,\widehat G)|}{D_i(\widehat G_*)}  + \frac{N_i(z,\widehat G)}{D_i(\widehat G)}\frac{|D_i(\widehat G_*)-D_i(\widehat G)|}{D_i(\widehat G_*)} \\
&\qquad \le \frac{|N_i(z,G)-N_i(z,\widehat G)|}{D_i(\widehat G_*)}
+ \frac{|D_i(G)-D_i(\widehat G)|}{D_i(\widehat G_*)}.
\end{aligned}
$$
In the last step, we used two facts: first, it holds that $N_i(z,\widetilde{G})/D_i(\widetilde{G})\in[0,1]$ for all $\widetilde{G}$ since they are conditional CDFs, and second, the map $\widetilde{G} \mapsto D_i(\widetilde{G})$ is linear, which implies that
$$
D_i(G)-D_i(\widehat{G}_*) = (D_i(G)-D_i(\widehat{G}))/2,\quad D_i(\widehat{G}_*)-D_i(\widehat{G}) = (D_i(G)-D_i(\widehat{G}))/2.
$$
Combining all of the above results, we get
\begin{align}
    & \frac{1}{n}\sum_
    {i=1}^n \EE[G]{\sup_{z\in\mathbb R}
    \left|
    F_{G,i}(z\mid  \hat{b}_i)-F_{\widehat G,i}(z\mid  \hat{b}_i)
    \right|}\nonumber\\
    &\quad\leq \PP[G]{\mathcal{A}^c} +\frac{1}{n}\sum_
    {i=1}^n \EE[G]{\sup_{z\in\mathbb{R}}\left|\frac{N_i(z,G)-N_i(z,\widehat{G})}{D_i(\widehat{G}_*)}\right|\mathds{1}(\mathcal{A})}\nonumber\\
    &\quad\quad\quad +\frac{1}{n}\sum_
    {i=1}^n \EE[G]{\left|\frac{D_i(G)-D_i(\widehat{G})}{D_i(\widehat{G}_*)}\right|\mathds{1}(\mathcal{A})}\nonumber\\
    &\quad=:\PP[G]{\mathcal{A}^c}+\mathrm{I}+\mathrm{II}.\label{eq:hetero_final_goal}
\end{align}
 By Lemma~\ref{lemm:hetero_hellienger}, $\PP[G]{\mathcal{A}^c}\leq\exp{(-c_0\log{n})}$ with $c_0\geq2$. We consider the following two lemmas that bound terms $\mathrm{I}$ and $\mathrm{II}$.
\begin{lemm}
\label{lemm:hetero_n_ratio_bound}
    It holds that 
    $$
    \frac{1}{n}\sum_
    {i=1}^n \EE[G]{\sup_{z\in\mathbb{R}}\left|\frac{N_i(z,G)-N_i(z,\widehat{G})}{D_i(\widehat{G}_*)}\right|\mathds{1}(\mathcal{A})}\lesssim_{\Gamma,\underline{\sigma},\bar{\sigma},\underline{\gamma},\bar{\gamma}} \frac{(\log{n})^{3/2}}{n^{(1-\check{\gamma}^2)/2}}.
    $$
\end{lemm}
\begin{lemm}
\label{lemm:hetero_marginal_ratio_bound}
    It holds that 
    $$
    \frac{1}{n}\sum_
    {i=1}^n \EE[G]{\left|\frac{D_i(G)-D_i(\widehat{G})}{D_i(\widehat{G}_*)}\right|\mathds{1}(\mathcal{A})} \lesssim_{\Gamma,\underline{\sigma},\bar{\sigma}} \frac{\log{n}}{\sqrt{n}}.
    $$
\end{lemm}
Plugging in the conclusions of the above two lemmas in \eqref{eq:hetero_final_goal}, the assertion of the theorem follows.
To prove Lemma~\ref{lemm:hetero_n_ratio_bound}, we consider the following covering lemma.
\begin{lemm}
\label{lemm:hetero_cardinality_cover_N}
    For $U>0$, define the distance 
    \begin{align}
    \label{eq:hetero_metric_N}
        d_{N,U}(G_1,G_2)
        :=
        \max_{1\le i\le n}
        \bigg[
        \left\|f_{G_1}^{(i)}-f_{G_2}^{(i)}\right\|_{\infty,U}
        +
        \sup_{z\in\mathbb R}\sup_{|l|\le U}
        \left|
        N_{G_1}^{(i)}(z,l)-N_{G_2}^{(i)}(z,l)
        \right|
        \bigg].
    \end{align}
    Then for any $0<\varepsilon<1/\sqrt{2\pi}$ and $U>0$, it holds that
    $$
    \log \mathcal N
    \left(\varepsilon,\mathcal{G},d_{N,U}\right)
    \lesssim_{\underline{\sigma},\bar{\sigma},\bar{\gamma}}
    \left(\log{\left(\frac1\varepsilon\right)}\right)^2
    \max\left\{
    \frac{U}{\sqrt{|\log\varepsilon|}},1
    \right\}.
    $$
    Moreover, the same bound holds for a proper cover of any subset of
    $\mathcal G$, up to changing the covering radius by a universal constant.
\end{lemm}
\begin{proof}
Denote $a_i:=\Omega_{i,12}/\Omega_{i,22}$. Under Assumption~\ref{assm:subgaussian}, there exists a constant
    $C=C(\underline{\sigma},\bar{\sigma},\bar{\gamma})$
    such that
    \[
        |a_i|\le C,\qquad
        C(\Sigma_i)\le C,\qquad
        \Sigma_{i,11}\le C,\qquad
        \Sigma_{i,22}\in[C^{-1},C],
        \qquad
        \Omega_{i,22}^{-1}\in[C^{-1},C],
    \]
    for all $i\in\{1,\ldots,n\}$. Let
    $R_\varepsilon=C_1\sqrt{|\log\varepsilon|}$, with $C_1$ sufficiently large, and define $$U_\varepsilon:=U+C R_\varepsilon.$$ We claim that for all $l\in\mathbb{R}$,
    \begin{equation}
    \label{eq:hetero_entropy_tail}
        \sup_{\widetilde G\in\mathcal G}
        \sup_{1\le i\le n}
        C(\Sigma_i)
        \int_{|u|>R_\varepsilon}
        \exp\left(-\frac{u^2}{2\Sigma_{i,11}}\right)
        f_{\widetilde G}\left(a_i u+l;\Omega_{i,22}^{-1}\right)
        \dd u
        \le
        \frac{\varepsilon}{4}.
    \end{equation}
    In particular, since $\Omega_{i,22}^{-1} \ge C^{-1}$ and $\|\varphi(\cdot;\sigma^2)\|_\infty \le 1/\sqrt{2\pi\sigma^2}$, we have $f_{\widetilde G}(t;\Omega_{i,22}^{-1}) \le 1/\sqrt{2\pi C^{-1}}$. Then there is a constant $K_1=K_1(\underline{\sigma},\bar{\sigma},\bar{\gamma})$ such that
    $\sup_{\widetilde G}\,\sup_{i}\,\sup_{t \in \mathbb R}f_{\widetilde G}(t;\Omega_{i,22}^{-1}) \le K_1.$ Combining with a normal tail bound and $\Sigma_{i,11}\leq C$, we have
    $$
    \begin{aligned}
        \int_{|u|>R_\varepsilon}
        \exp\left(-\frac{u^2}{2\Sigma_{i,11}}\right)
        f_{\widetilde G}\left(a_i u+l;\Omega_{i,22}^{-1}\right)
        \dd u
        &\le K_1\int_{|u|>R_\varepsilon}\exp\left(-\frac{u^2}{2\Sigma_{i,11}}\right)\dd u\\
        &\leq 2K_1\sqrt{2\pi C}\exp\left(-\frac{C_1^2|\log \varepsilon|}{2C}\right).
    \end{aligned}
    $$
    Since $C(\Sigma_i)\leq C$, choose $K_2=K_2(\underline{\sigma},\bar{\sigma},\bar{\gamma}):=2CK_1\sqrt{2\pi C}$, and choose $C_1$ large enough so that $K_2 \varepsilon^{C_1^2/(2C)-1} \le 1/4$ uniformly for $\varepsilon \in (0, 1/\sqrt{2\pi})$; this is possible by taking $C_1^2/(2C)$ as large as needed. Then we plug back and prove \eqref{eq:hetero_entropy_tail}.

    Now consider the augmented class of normal mixture densities
    $$
    \begin{aligned}
        \mathcal F_n^{\mathrm{aug}}
        :=
        \bigg\{&
        \Big(
        f_{\widetilde G}(\cdot;\Sigma_{1,22}),\ldots,
        f_{\widetilde G}(\cdot;\Sigma_{n,22}),
        f_{\widetilde G}(\cdot;\Omega_{1,22}^{-1}),\ldots,
        f_{\widetilde G}(\cdot;\Omega_{n,22}^{-1})
        \Big):
        \widetilde G\in\mathcal{G}
        \bigg\}.
    \end{aligned}
    $$
    has $2n$ normal convolutions of $\widetilde{G}$, all with variances bounded above and below by constants depending only on $\underline{\sigma},\bar{\sigma},\bar{\gamma}$. Hence, the proof of Lemma \ref{lemm:hetero_cardinality_cover} applies verbatim to this augmented class. Therefore, there exist $G_1,...,G_J$ with
    $$
    \log J\lesssim_{\underline{\sigma},\bar{\sigma},\bar{\gamma}}\left(\log{\left(\frac1\varepsilon\right)}\right)^2\max\left\{\frac{U_\varepsilon}{\sqrt{|\log\varepsilon|}},1\right\}
    \lesssim\left(\log{\left(\frac1\varepsilon\right)}\right)^2\max\left\{\frac{U}{\sqrt{|\log\varepsilon|}},1\right\},
    $$
   such that for every $\widetilde G\in\mathcal{G}$, there is some
    $j\in\{1,\ldots,J\}$ satisfying
    $$
    \begin{aligned}
        \max_{1\le i\le n}
        \left\|f_{\widetilde G}(\cdot;\Sigma_{i,22})
        -f_{G_j}(\cdot;\Sigma_{i,22})\right\|_{\infty,U}
        +
        \max_{1\le i\le n}
        \left\|f_{\widetilde G}(\cdot;\Omega_{i,22}^{-1})
        -f_{G_j}(\cdot;\Omega_{i,22}^{-1})\right\|_{\infty,U_\varepsilon}
        \le c\varepsilon,
    \end{aligned}
    $$
    for a sufficiently small constant $c>0$. 
    
    We claim that the same $\{G_j\}$ also approximates
    $N_{\widetilde G}^{(i)}$ uniformly on $\mathbb R\times[-U,U]$. Indeed, by Proposition~\ref{prop:f_formula}, for $|l|\le U$,
    $$
    \begin{aligned}
        &\sup_{z\in\mathbb R}
        \left|
        N_{\widetilde G}^{(i)}(z,l)-N_{G_j}^{(i)}(z,l)
        \right|
        \\
        &\quad\le
        C(\Sigma_i)
        \int_{\mathbb R}
        \exp\left(-\frac{u^2}{2\Sigma_{i,11}}\right)
        \left|
        f_{\widetilde G}\left(a_i u+l;\Omega_{i,22}^{-1}\right)
        -
        f_{G_j}\left(a_i u+l;\Omega_{i,22}^{-1}\right)
        \right|
        \dd u
        \\
        &\quad\le
        Cc\varepsilon
        +
        C(\Sigma_i)
        \int_{|u|>R_\varepsilon}
        \exp\left(-\frac{u^2}{2\Sigma_{i,11}}\right)
        \left[
        f_{\widetilde G}\left(a_i u+l;\Omega_{i,22}^{-1}\right)
        +
        f_{G_j}\left(a_i u+l;\Omega_{i,22}^{-1}\right)
        \right]
        \dd u
        \\
        &\quad\le
        Cc\varepsilon+\epsilon/2\leq \varepsilon,
    \end{aligned}
    $$
    where we use the fact that $|a_i u+l|\le U_\varepsilon$ whenever $|u|\le R_\varepsilon$ and $|l|\le U$, combined with \eqref{eq:hetero_entropy_tail}, and by choosing a $c>0$ sufficiently small. Thus, the same cover is an $\varepsilon$-cover under $d_{N,U}$, after adjusting constants.

    Indeed, if we restrict to any subset of
    $\mathcal{G}$, an improper
    $\varepsilon/2$-cover can be converted into a proper
    $\varepsilon$-cover by replacing each center whose ball intersects the
    subset by one point in that intersection. This proves the last claim.
\end{proof}

\begin{proof}[Proof of Lemma~\ref{lemm:hetero_n_ratio_bound}]
Since $b_i\sim G\in\mathcal{G}_{\Gamma}$, and  $\hat{b}_i = b_i +\varepsilon_i$ with $\varepsilon_i\sim \mathrm{N}(0,\tau_i^2)$, then $\hat b_i$ is a sub-Gaussian with a proxy constant depending only on $(\Gamma,\bar\sigma)$ (since $\tau_i^2\leq\bar{\sigma}^2$). In particular, there exists a constant
$c>0$, depending only on $(\Gamma,\bar\sigma)$, such that $\mathbb P(|\hat b_i|>t)\le 2e^{-ct^2}$. Hence, we can choose $B_n:=C_B\sqrt{\log n}$ with $C_B$ sufficiently large such that 
$$
\PP[G]{|\hat b_i|>B_n}\le \frac{1}{n}.
$$
Also let $\eta:=n^{-2}$. Define the local class
$$
\mathcal G_n:=\left\{\widetilde G\in\mathcal{G}:\;\bar d(G,\widetilde G)<C_0\frac{\log n}{\sqrt n}\right\},
$$
where $C_0$ is the constant in Lemma~\ref{lemm:hetero_hellienger}. By applying Lemma~\ref{lemm:hetero_cardinality_cover_N} with $U=B_n$, $\varepsilon=\eta$, there exists a proper $\eta$-cover $\{G_j:j=1,\ldots,J\}\subseteq \mathcal G_n$ of $\mathcal G_n$ under $d_{N,B_n}$ satisfying
$$
\log J   \lesssim_{\underline{\sigma},\bar{\sigma},\bar{\gamma}} (\log n)^2 \max\left\{ \frac{B_n}{\sqrt{\log n}},1 \right\}\lesssim(\log n)^2 .
$$
 On the event $\mathcal A$, we have $\widehat G\in\mathcal G_n$, and
    hence there exists a random index $\widehat j\in\{1,\ldots,J\}$ such
    that $d_{N,B_n}(\widehat G,G_{\widehat j})\le \eta$, that is
    \begin{equation}
    \label{eq:hetero_cover_Ghat}
        \sup_{z\in\mathbb R}
        \left|
        N_{\widehat G}^{(i)}(z,l)-N_{G_{\widehat j}}^{(i)}(z,l)
        \right|
        \le \eta,
        \qquad
        \left|
        f_{\widehat G}^{(i)}(l)-f_{G_{\widehat j}}^{(i)}(l)
        \right|
        \le \eta \qquad \text{for all }|l| \le B_n,\ i \le n.
    \end{equation}
    For each deterministic $G_j$, define
    $$
    V_{ij}(l)
        :=
        \frac{
        \sup_{z\in\mathbb R}
        \left|
        N_G^{(i)}(z,l)-N_{G_j}^{(i)}(z,l)
        \right|
        }{
        f_G^{(i)}(l)+f_{G_j}^{(i)}(l)
        },
    $$
    and $V_{ij}(l)\in[0,1]$ since $0\le N_{\widetilde G}^{(i)}(z,l)\le f_{\widetilde G}^{(i)}(l)$ for
    any $\widetilde G$.

    We claim that on $\mathcal{A}$,
    \begin{align}
    \label{eq:hetero_ratio_cover_reduction}
        &\sup_{z\in\mathbb R}
        \left|
        \frac{
        N_i(z,G)-N_i(z,\widehat G)
        }{
        D_i(\widehat G_*)
        }
        \right|
        \lesssim
        \mathds 1(|\hat b_i|>B_n)
        +
        V_{i\widehat j}(\hat b_i)
        +
        \eta
        \frac{\mathds 1(|\hat b_i|\le B_n)}
        {f_G^{(i)}(\hat b_i)}.
    \end{align}
    We prove this by case analysis. For every $z\in\mathbb{R}$, write
    $$
    \frac{|N_i(z,G) - N_i(z,\widehat{G})|}{D_i(\widehat{G}_*)} \le \frac{|N_i(z,G) - N_i(z, G_{\hat j})|}{D_i(\widehat{G}_*)} + \frac{|N_i(z, G_{\hat j}) - N_i(z, \widehat{G})|}{D_i(\widehat{G}_*)}.
    $$
    For the second term, by \eqref{eq:hetero_cover_Ghat} applied at $l =\hat{b}_i$ when $|\hat b_i| \le B_n$, the numerator is $\le \eta$. The denominator satisfies $D_i(\widehat{G}_*) = (f_G^{(i)}(\hat b_i) + f_{\widehat{G}}^{(i)}(\hat b_i))/2 \ge f_G^{(i)}(\hat b_i)/2$, hence
    $$
    \frac{|N_i(z, G_{\hat j}) - N_i(z, \widehat{G})|}{D_i(\widehat{G}_*)} \le \frac{2\eta}{f_G^{(i)}(\hat b_i)} \quad \text{on } \{|\hat b_i| \le B_n\}.
    $$
    For the first term, for $|l|\le B_n$. If
    $f_G^{(i)}(l)+f_{G_{\widehat j}}^{(i)}(l)\ge 2\eta$, then
    $D_i(\widehat G_*)\ge
    \{f_G^{(i)}(l)+f_{G_{\widehat j}}^{(i)}(l)\}/4$, and hence
    $$
        \frac{
        \sup_z|N_G^{(i)}(z,l)-N_{G_{\widehat j}}^{(i)}(z,l)|
        }{
        D_i(\widehat G_*)
        }
        \le
        4V_{i\widehat j}(l).
    $$
    On the other hand, if
    $f_G^{(i)}(l)+f_{G_{\widehat j}}^{(i)}(l)<2\eta$, then both $f_G^{(i)}(l)<2\eta$ and $f_{G_{\widehat j}}^{(i)}(l)<2\eta$, then $\sup_z|N_G^{(i)}(z,l)-N_{G_{\widehat j}}^{(i)}(z,l)|\le 2\eta$. Combined with $D_i(\widehat G_*)\ge f_G^{(i)}(l)/2$, then 
    $$
    \frac{|N_i(z, G) - N_i(z, G_{\hat j})|}{D_i(\widehat{G}_*)} \le \frac{4\eta}{f_G^{(i)}(l)}.
    $$
    For $|\hat{b}_i|>B_n$, then $|N_i(z,G) - N_i(z, \widehat{G})| \le f_G^{(i)}(\hat b_i) + f_{\widehat G}^{(i)}(\hat b_i) = 2D_i(\widehat{G}_*)$, so
    $$
    \frac{|N_i(z, G) - N_i(z, \widehat G)|}{D_i(\widehat{G}_*)} \le 2\, \mathds{1}(|\hat b_i| > B_n).
    $$
    Combining all together claims \eqref{eq:hetero_ratio_cover_reduction}. We further take expectations in
    \eqref{eq:hetero_ratio_cover_reduction},
    \begin{align}
    \label{eq:hetero_bound_N_replacement}
        \mathrm I
        &\lesssim
        \frac1n
        +
        \eta\cdot \frac1n\sum_{i=1}^n
        \EE[G]{
        \frac{\mathds 1(|\hat b_i|\le B_n)}
        {f_G^{(i)}(\hat b_i)}
        }
        +
        \EE[G]{
        \max_{1\le j\le J}
        \frac1n\sum_{i=1}^n V_{ij}(\hat b_i)
        }
        \nonumber\\
        &=
        \frac1n
        +
        2B_n\eta
        +
        \EE[G]{
        \max_{1\le j\le J}
        \frac1n\sum_{i=1}^n V_{ij}(\hat b_i)
        },
    \end{align}
    where the second term is justified by
    $$
    \EE[G]{\frac{\mathds{1}(|\hat b_i| \le B_n)}{f_G^{(i)}(\hat b_i)}}= \int_{-B_n}^{B_n} \frac{1}{f_G^{(i)}(t)} f_G^{(i)}(t)\, \dd t = 2 B_n.
    $$
    For a fixed $j$, the random variables
    $V_{ij}(\hat b_i)$ are independent and take values in $[0,1]$.
    Hence, by Hoeffding's lemma and a union bound,
    \begin{align}
        &\mathbb E_G\left[
\max_{1\le j\le J}
\frac1n\sum_{i=1}^n V_{ij}(\hat b_i)
\right] \nonumber\\
&\quad\le
\max_{1\le j\le J}
\frac1n\sum_{i=1}^n
\mathbb E_G\left[V_{ij}(\hat b_i)\right] +
\mathbb E_G\left[
\max_{1\le j\le J}
\frac1n\sum_{i=1}^n
\left\{
V_{ij}(\hat b_i)
-
\mathbb E_G[V_{ij}(\hat b_i)]
\right\}
\right]  \nonumber\\
&\quad\le
\max_{1\le j\le J}
\frac1n\sum_{i=1}^n
\mathbb E_G\left[V_{ij}(\hat b_i)\right]
+
\inf_{\lambda>0}
\left\{
\frac{\log J}{\lambda}
+
\frac{\lambda}{8n}
\right\} \nonumber \\
&\quad\le
\max_{1\le j\le J}
\frac1n\sum_{i=1}^n
\mathbb E_G\left[V_{ij}(\hat b_i)\right]
+
C\sqrt{\frac{\log J}{n}}. \label{eq:hetero_finite_maximal}
    \end{align}
    We next control the deterministic means in
    \eqref{eq:hetero_finite_maximal}. For fixed $i$ and $j$, set
    $$
        S_{ij}
        :=
        \sup_{z\in\mathbb R}\sup_{|l|\le B_n}
        \left|
        N_G^{(i)}(z,l)-N_{G_j}^{(i)}(z,l)
        \right|.
    $$
    Since $V_{ij}\le 1$ and
    $f_G^{(i)}(l)/\{f_G^{(i)}(l)+f_{G_j}^{(i)}(l)\}\le 1$, we have
    \begin{align}
        \EE[G]{V_{ij}(\hat b_i)}
        &\le
        \PP[G]{|\hat b_i|>B_n}
        +
        \int_{-B_n}^{B_n}
        \sup_{z\in\mathbb R}
        \left|
        N_G^{(i)}(z,l)-N_{G_j}^{(i)}(z,l)
        \right|
        \dd l
        \nonumber\\
        &\le
        \frac{1}{n}
        +
        2B_nS_{ij}.
        \label{eq:hetero_mean_V_bound}
    \end{align}
    It remains to bound $S_{ij}$. By Proposition~\ref{prop:f_formula}, 
   \begin{align}
    & \left\lvert N_{G}^{(i)}(z,l) - N_{G_j}^{(i)}(z,l)\right\rvert\big/C(\Sigma_i) \nonumber \\
    =&\left\lvert\int_{-\infty}^ze^{-u^2/2\Sigma_{i,11}}\left(f_{G}\left(\frac{\Omega_{i,12}}{\Omega_{i,22}}u+l;\Omega_{i,22}^{-1}\right)-f_{G_j}\left(\frac{\Omega_{i,12}}{\Omega_{i,22}}u+l;\Omega_{i,22}^{-1}\right)\right)~\dd u\right\rvert \nonumber \\
    \leq& \left[\int_{-\infty}^z e^{-u^2/\Sigma_{i,11}}~\dd u\right]^\frac{1}{2}\left[\int_{-\infty}^z\left[f_{G}\left(\frac{\Omega_{i,12}}{\Omega_{i,22}}u+l;\Omega_{i,22}^{-1}\right)-f_{G_j}\left(\frac{\Omega_{i,12}}{\Omega_{i,22}}u+l;\Omega_{i,22}^{-1}\right)\right]^2\dd u\right]^{\frac{1}{2}}. \nonumber
    \end{align}
    Let
    $$
        h_{ij}(t):=f_G(t;\Omega_{i,22}^{-1})-f_{G_j}(t;\Omega_{i,22}^{-1}),
        \qquad
        a_i:=\frac{\Omega_{i,12}}{\Omega_{i,22}}=-\gamma_i\frac{\tau_i}{\tilde{\sigma}_i},
    $$
    and by Assumption~\ref{assm:subgaussian}, $|a_i|$ is bounded away from zero,
    $$
    \begin{aligned}
        \int_{-\infty}^z\left(f_{G}\left(\frac{\Omega_{i.12}}{\Omega_{i,22}}u+l;\Omega_{i,22}^{-1}\right)-f_{G_j}\left(\frac{\Omega_{i,12}}{\Omega_{i,22}}u+l;\Omega_{i,22}^{-1}\right)\right)^2\dd u &= \int_{-\infty}^z h_{ij}(a_iu+l)^2\dd u \\&\leq \frac{1}{|a_i|} \int_{-\infty}^{\infty} h_{ij}(t)^2 \dd t,
    \end{aligned}
    $$

    Therefore
    \begin{equation}
    \label{eq:bound_all_N}
         \sup_{z\in\mathbb R}\sup_{l\in\mathbb R}
        \left|
        N_G^{(i)}(z,l)-N_{G_j}^{(i)}(z,l)
        \right|
        \lesssim_{\underline\sigma,\bar\sigma,\underline\gamma,\bar\gamma}
        \left(
        \int_{\mathbb R}h_{ij}(t)^2\dd t
        \right)^{1/2}.
    \end{equation}
     Notice that $\Omega_{i,22}^{-1} = \Sigma_{i,22}(1-\gamma_i^2)$, and denote $f^*$ as the Fourier transform of $f$, then
     $$
    f^*_{G}(\omega;\Omega_{i,22}^{-1}) = G^*(\omega)e^{-\frac{1}{2}\Omega_{i,22}^{-1}\omega^2} = \left(G^*(\omega)e^{-\frac{1}{2}\Sigma_{i,22}\omega^2}\right)e^{\frac{1}{2}\Sigma_{i,22}\gamma_i^2\omega^2} = f^*_{G}(\omega;\Sigma_{i,22})e^{\frac{1}{2}\Sigma_{i,22}\gamma_i^2\omega^2}.
     $$
     Similar results also hold for $\{G_j\}$. By Plancherel's identity, and a non-random $U_i>0$ to be chosen later,
     $$
    \begin{aligned}
        \int_{-\infty}^{\infty} h_{ij}(t)^2 \dd t\quad&=\frac{1}{2\pi}\int_{-\infty}^{\infty}\left\vert f_{G}^{*}(\omega;\Omega_{i,22}^{-1})-f_{G_j}^*(\omega;\Omega_{i,22}^{-1})\right\rvert^2\dd\omega \\
        &\quad=\frac{1}{2\pi}\int_{-\infty}^{\infty}e^{\Sigma_{i,22}\gamma_i^2\omega^2}\left\vert f_{G}^{*}(\omega;\Sigma_{i,22})-f_{G_j}^*(\omega;\Sigma_{i,22})\right\vert^2 \dd\omega \\
        &\quad= \frac{1}{2\pi}\int_{[- U_i,U_i]}e^{\Sigma_{i,22}\gamma_i^2\omega^2}\left\vert f_{G}^{*}(\omega;\Sigma_{i,22})-f_{G_j}^*(\omega;\Sigma_{i,22})\right\vert^2\dd\omega\\
        &\quad\quad\quad+\frac{1}{2\pi}\int_{[-U_i,U_i]^c}e^{\Sigma_{i,22}\gamma_i^2\omega^2}\left\vert f_{G}^{*}(\omega;\Sigma_{i,22})-f_{G_j}^*(\omega;\Sigma_{i,22})\right\vert^2\dd\omega\\
        &\quad:= I_{ij1}+I_{ij2}.
    \end{aligned}
    $$
    For $I_{ij1}$,
    $$
    \begin{aligned}
    I_{ij1} &\leq e^{\Sigma_{i,22}\gamma_i^2U_i^2}\cdot \frac{1}{2\pi}\int_{-\infty}^{\infty}\left\vert f_{G}^{*}(\omega;\Sigma_{i,22})-f_{G_j}^*(\omega;\Sigma_{i,22})\right\vert^2\dd\omega \\
    &=e^{\Sigma_{i,22}\gamma_i^2U_i^2}\int_{-\infty}^{\infty}\left(f_{G}(t;\Sigma_{i,22})-f_{G_j}(t;\Sigma_{i,22})\right)^2\dd t\\
    &=e^{\Sigma_{i,22}\gamma_i^2U_i^2}\int_{-\infty}^{\infty}\left(\sqrt{f_{G}(t;\Sigma_{i,22})}-\sqrt{f_{G_j}(t;\Sigma_{i,22})}\right)^2\left(\sqrt{f_{G}(t;\Sigma_{i,22})}+\sqrt{f_{G_j}(t;\Sigma_{i,22})}\right)^2\dd t \\
    &\leq 2e^{\Sigma_{i,22}\gamma_i^2U_i^2}\left\Vert\sqrt{f_{G}(t;\Sigma_{i,22})}+\sqrt{f_{G_j}(t;\Sigma_{i,22})}\right\Vert_{\infty}^{2}\mathfrak{H}^2(f_{G}(\cdot;\Sigma_{i,22}),f_{G_j}(\cdot;\Sigma_{i,22}))\\
    &\lesssim e^{\Sigma_{i,22}\gamma_i^2U_i^2}\mathfrak{H}^2(f_{G}(\cdot;\Sigma_{i,22}),f_{G_j}(\cdot;\Sigma_{i,22})).
    \end{aligned}
    $$
    For $I_{ij2}$, first noticing that
    $$
    \left\vert f_{G}^{*}(\omega;\Sigma_{i,22})-f_{G_j}^*(\omega;\Sigma_{i,22})\right\vert^2 = \left\vert G^*(\omega)-G_j^*(\omega)\right\vert^2e^{-\Sigma_{i,22}\omega^2}\leq 4e^{-\Sigma_{i,22}\omega^2},
    $$
    then applying the normal tail bound,
    $$
    I_{i2} \leq \frac{4}{2\pi}\int_{[-U_i,U_i]^c} e^{-\Sigma_{i,22}\omega^2}e^{\Sigma_{i,22}\gamma_i^2\omega^2}~\dd\omega\lesssim \frac{1}{\Sigma_{i,22}(1-\gamma_i^2)U_i}e^{-\Sigma_{i,22}(1-\gamma_i^2)U_i^2}.
    $$
    Take $U_i=(\log{n}/\Sigma_{i,22})^{1/2}$, then we have
    $$
    \begin{aligned}
        \int_{-\infty}^{\infty} h_{ij}(t)^2\dd t \lesssim n^{\gamma_i^2}\mathfrak{H}^2(f_G^{(i)},f_{G_j}^{(i)})+\frac{n^{-(1-\gamma_i^2)}}{\sqrt{\log{n}}}.
    \end{aligned}
    $$
    By using the inequality $\sqrt{a+b} \leq \sqrt{a}+\sqrt{b}$ for $a , b>0$, and combined with \eqref{eq:bound_all_N}, we have that
    \begin{equation}
    \label{eq:hetero_deterministic_N_bound}
        S_{ij}
        \lesssim_{\underline{\sigma},\bar{\sigma}\underline{\gamma},\bar{\gamma}}
        n^{\gamma_i^2/2}
        \mathfrak H(f_G^{(i)},f_{G_j}^{(i)})
        +
        n^{-(1-\gamma_i^2)/2}(\log n)^{-1/4}.
    \end{equation}
    Since the cover is proper, each $G_j$ belongs to $\mathcal G_n$.
    Hence, uniformly over $j\in\{1,\ldots,J\}$,
    \begin{align}
        \frac1n\sum_{i=1}^n S_{ij}
        &\lesssim
        n^{\check\gamma^2/2}
        \frac1n\sum_{i=1}^n
        \mathfrak H(f_G^{(i)},f_{G_j}^{(i)})
        +
        n^{-(1-\check\gamma^2)/2}(\log n)^{-1/4}
        \nonumber\\
        &\le
        n^{\check\gamma^2/2}
        \bar d(G,G_j)
        +
        n^{-(1-\check\gamma^2)/2}(\log n)^{-1/4}
        \nonumber\\
        &\lesssim
        \frac{\log n}{n^{(1-\check\gamma^2)/2}}.
        \label{eq:hetero_avg_Sij_bound}
    \end{align}
    Combining \eqref{eq:hetero_mean_V_bound} and
    \eqref{eq:hetero_avg_Sij_bound}, we get
    $$
        \max_{1\le j\le J}
        \frac1n\sum_{i=1}^n
        \EE[G]{V_{ij}(\hat b_i)}
        \lesssim
        \frac1n
        +
        \frac{B_n\log n}{n^{(1-\check\gamma^2)/2}}
        \lesssim
        \frac{(\log n)^{3/2}}{n^{(1-\check\gamma^2)/2}}.
    $$
    Also, combined with the Hoeffding term,
    $$
        \sqrt{\frac{\log J}{n}}
        \lesssim
        \frac{\log n}{\sqrt n}
        \lesssim
        \frac{(\log n)^{3/2}}{n^{(1-\check\gamma^2)/2}}.
    $$
    Plugging these bounds into \eqref{eq:hetero_bound_N_replacement},
    $$
        \mathrm I
        \lesssim_{\Gamma,\underline{\sigma},\bar{\sigma},\underline{\gamma},\bar{\gamma}}
        \frac{(\log n)^{3/2}}{n^{(1-\check\gamma^2)/2}}.
    $$
    This completes the proof.

\end{proof}
\begin{proof}[Proof of Lemma~\ref{lemm:hetero_marginal_ratio_bound}]
    Let $\eta=1/n$. Consider the following class of densities
    $$
    \mathcal{F}_n^{-1}:=\left\{\left(f_{\widetilde{G}}^{(1)},...,f_{\widetilde{G}}^{(n)}\right): \widetilde{G}\in\mathcal{G},\;\bar{d}(G,{\widetilde{G}})<\frac{C_0\log{n}}{\sqrt{n}}\right\}. 
    $$
    The constant $C_0$ is in Lemma \ref{lemm:hetero_hellienger}. Again, we take $B_n := C_B\sqrt{\log{n}}$. Let $\mathcal{S}=\{(f^{(1)}_{G_j},...,f^{(n)}_{G_j}):j\in\mathcal{J}\}\subseteq \mathcal{F}_n^{-1}$  , $\mathcal{J}=\{1,...,J\}$, $J= \#S$ be a proper $(\Vert\cdot\rVert_{\infty,B_n},\eta)$-cover of $\mathcal{F}_n^{-1}$. Here, a proper cover means that the centers of the cover are themselves elements of $\mathcal{F}_n^{-1}$. Lemma \ref{lemm:hetero_cardinality_cover} provides a cover for a larger class of functions, however, it is not a proper cover for $\mathcal{F}_n^{-1}$. By a standard argument, an $\eta/2$-cover in Lemma \ref{lemm:hetero_cardinality_cover} yields a proper $\eta$-cover for $\mathcal{F}_n^{-1}$. Therefore, we get that 
   $$\log{J} \lesssim_{\underline{\sigma},\bar{\sigma}}(\log{n})^2\max\left\{\frac{B_n}{\sqrt{\log{n}}},1\right\}.$$
   By the definition of $B_n$, we further conclude that $\log{J} \lesssim_{\Gamma,\underline{\sigma},\bar{\sigma}}(\log{n})^2.$ Now consider the main argument of the lemma. On the event $\mathcal{A}$, there must exist a (random) index \smash{$\hat{j}$} such that 
   $$\max_{1\leq i \leq n}\max_{|t|\leq B_n}\left|f_{\widehat{G}}^{(i)}(t)-f_{G_{\hat{j}}}^{(i)}(t)\right|\leq\eta.$$
   In particular, on $\{|\hat{b}_i|\leq B_n\}$, $\left|D_i(\widehat{G})-D_i(G_{\hat{j}})\right|\leq \eta.$ Also, we have $\PP[G]{|\hat{b}_i|>B_n}\leq 1/n$ for $n\geq1$. Then,
    \begin{align}
    &\frac{1}{n}\sum_
    {i=1}^n \EE[G]{\left|\frac{D_i(G)-D_i(\widehat{G})}{D_i(\hat{G}_*)}\right|\mathds{1}(\mathcal{A})} \nonumber\\
    &\quad = \frac{2}{n}\sum_
    {i=1}^n \EE[G]{\left|\frac{D_i(G)-D_i(\widehat{G})}{D_i(G)+D_i(\widehat{G})}\right|\mathds{1}(\mathcal{A})}\nonumber\\
    &\quad\leq\frac{2}{n}\sum_{i=1}^n\left\{\EE[G]{\left|\frac{D_i(G)-D_i(\widehat{G})}{D_i(G)+D_i(\widehat{G})}\right|\mathds{1}(\mathcal{A})\mathds{1}(|\hat{b}_i|\leq B_n)}+\PP{|\hat{b}_i|> B_n)}\right\}\nonumber\\
    &\quad\leq\frac{2}{n}+\frac{2}{n}\sum_{i=1}^n\EE[G]{\left|\frac{D_i(G)-D_i(G_{\hat{j}})}{D_i(G)+D_i(G_{\hat{j}})}\right|\mathds{1}(\mathcal{A})\mathds{1}(|\hat{b}_i|\leq B_n)}\nonumber\\
    &\quad\quad\quad\quad+\frac{2}{n}\sum_{i=1}^n\EE[G]{\left|\frac{D_i(G)-D_i(\widehat{G})}{D_i(G)+D_i(\widehat{G})}-\frac{D_i(G)-D_i(G_{\hat{j}})}{D_i(G)+D_i(G_{\hat{j}})}\right|\mathds{1}(\mathcal{A})\mathds{1}(|\hat{b}_i|\leq B_n)}.\label{eq:density_transform}
    \end{align}
    Proceeding similarly as in the proof of Lemma S10 of \cite{ignatiadis2025empirical}, we can show that
    $$
    \begin{aligned}
        &\frac{1}{n}\sum_{i=1}^n\EE[G]{\left|\frac{D_i(G)-D_i(\widehat{G})}{D_i(G)+D_i(\widehat{G})}-\frac{D_i(G)-D_i(G_{\hat{j}})}{D_i(G)+D_i(G_{\hat{j}})}\right|\mathds{1}(\mathcal{A})\mathds{1}(|\hat{b}_i|\leq B_n)}\\
        &\quad\leq2\eta\frac{1}{n}\sum_{i=1}^n\EE[G]{\frac{\mathds{1}(|\hat{b}_i|\leq B_n)}{D_i(G)}}\leq4\eta B_n.
    \end{aligned}
    $$
    For the remaining term, 
    $$
    \begin{aligned}
    &\frac{1}{n}\sum_{i=1}^n\EE[G]{\left|\frac{D_i(G)-D_i(G_{\hat{j}})}{D_i(G)+D_i(G_{\hat{j}})}\right|\mathds{1}(\mathcal{A})\mathds{1}(|\hat{b}_i|\leq B_n)}\\
    &\quad\leq\EE[G]{\sup_{j\in\mathcal{J}}\left\{\frac{1}{n}\sum_{i=1}^n\left|\frac{D_i(G)-D_i(G_{j})}{D_i(G)+D_i(G_{j})}\right|\right\}}\\
    &\quad\leq \EE[G]{\sup_{j\in\mathcal{J}}\left\{\frac{1}{n}\sum_{i=1}^n\left(\left|\frac{D_i(G)-D_i(G_{j})}{D_i(G)+D_i(G_{j})}\right|-\EE[G]{\left|\frac{D_i(G)-D_i(G_{j})}{D_i(G)+D_i(G_{j})}\right|}\right)\right\}}\\
    &\quad\quad\quad + \sup_{j\in\mathcal{J}}\left\{\EE[G]{\frac{1}{n}\sum_{i=1}^n\left|\frac{D_i(G)-D_i(G_{j})}{D_i(G)+D_i(G_{j})}\right|}\right\}.
    \end{aligned}
    $$
    Both terms in the last step can be controlled using the same techniques applied in the proof of Lemma S10 of \cite{ignatiadis2025empirical}. In fact, using Lemma \ref{lemm:hetero_cardinality_cover} and Lemma \ref{lemm:hetero_hellienger}, we can show that
    $$
    \sup_{j\in\mathcal{J}}\left\{\EE[G]{\frac{1}{n}\sum_{i=1}^n\left|\frac{D_i(G)-D_i(G_{j})}{D_i(G)+D_i(G_{j})}\right|}\right\}\lesssim_{\Gamma,\underline{\sigma},\bar{\sigma}} \frac{\log n}{\sqrt{n}},
    $$
    and
    $$
    \EE[G]{\sup_{j\in\mathcal{J}}\left\{\frac{1}{n}\sum_{i=1}^n\left(\left|\frac{D_i(G)-D_i(G_{j})}{D_i(G)+D_i(G_{j})}\right|-\EE[G]{\left|\frac{D_i(G)-D_i(G_{j})}{D_i(G)+D_i(G_{j})}\right|}\right)\right\}}\lesssim_{\Gamma,\underline{\sigma},\bar{\sigma}} \frac{\log n}{\sqrt{n}}.
    $$
    Therefore, combining above results into \eqref{eq:density_transform}, we have
    $$
        \frac{1}{n}\sum_
        {i=1}^n \EE[G]{\left|\frac{D_i(G)-D_i(\widehat{G})}{D_i(\widehat{G}_*)}\right|\mathds{1}(\mathcal{A})} \lesssim_{\Gamma,\underline{\sigma},\bar{\sigma}} \frac{\log{n}}{\sqrt{n}}.
    $$
\end{proof}

\subsection{Proof of Proposition~\ref{prop:minimax}}

\begin{proof}
Since $\tilde{\sigma}_i^2,\tau_i^2,\gamma_i$ do not depend on $i$, we drop the subscript $i$ throughout the proof. Write
$
\nu=\tau^2(1-\gamma^2).
$
Recall from Proposition~\ref{prop:f_formula} that, in this homoscedastic case,
$
F_G(z\mid l)=N_G(z,l)/D_G(l),
$
where
$
D_G(l)=f_G(l;\tau^2)
$
and
$$
N_G(z,l)
=
C\int_{-\infty}^{z}
\exp\left(-\frac{u^2}{2\tilde\sigma^2}\right)
f_G\left(l-\frac{\gamma\tau}{\tilde\sigma}u;\nu\right)\,\dd u,
\qquad
C=(2\pi\tilde\sigma^2)^{-1/2}.
$$
Here we used that
$$
\frac{\Omega_{12}}{\Omega_{22}}
=
-\frac{\gamma\tau}{\tilde\sigma},
\qquad
\Omega_{22}^{-1}=\tau^2(1-\gamma^2)=\nu .
$$
Our strategy is to use Le Cam's two-point lemma. Define
$$
r_A
=
\frac{\nu(A+\tau^2)}{\tau^2(A+\nu)}
=
\frac{(1-\gamma^2)(A+\tau^2)}
{A+\tau^2(1-\gamma^2)}.
$$
Since $\gamma\ne0$, we have $r_A<1$, and also $r_A\to1-\gamma^2$ as $A\to\infty$. Since $\beta>(1-\gamma^2)/2$, we may choose $A$ large enough so that $r_A/2<\beta$, and then choose $\Gamma>A$. Below we fix these values of $A$ and $\Gamma$.

As the first point, take $G_0=\mathrm{N}(0,A)$, and write $g_0$ for its density. Then $G_0$ is $A$-sub-Gaussian, and hence $G_0\in\mathcal G_\Gamma$. For the second point, let
$$
\omega^2
=
\log n\cdot\frac{A+\tau^2}{A\tau^2}
$$
and define
$$
h(b)
=
g_0(b)\cb{\cos(\omega b)-\exp(-A\omega^2/2)}.
$$
For $\varepsilon\in(0,1/4)$, to be chosen sufficiently small below, let $G_1$ have density
$$
g_1(b)=g_0(b)+\varepsilon h(b).
$$
We first check that $g_1$ is a density. Indeed,
$$
g_1(b)
\ge
g_0(b)\cb{1-\varepsilon\left(|\cos(\omega b)|+\exp(-A\omega^2/2)\right)}
\ge
g_0(b)(1-2\varepsilon)>0,
$$
and
$$
\int h(b)\,\dd b
=
\int g_0(b)\cb{\cos(\omega b)-\exp(-A\omega^2/2)}\,\dd b
=0.
$$
Thus $g_1$ is nonnegative and integrates to one.

We next check the sub-Gaussian property. Since $g_1$ is even,
$\EEInline[G_1]{b}=0$. Moreover, for any $s\in\RR$,
$$
\begin{aligned}
\EE[G_1]{\exp(sb)}
&=
\exp(As^2/2)
+
\varepsilon
\exp(As^2/2-A\omega^2/2)
\cb{\cos(As\omega)-1}  \\
&\le
\exp(As^2/2)
\le
\exp(\Gamma s^2/2),
\end{aligned}
$$
where we used $\cos(As\omega)-1\le0$ and $\Gamma>A$. Hence $G_1\in\mathcal G_\Gamma$.

We will use the following elementary normal convolution identity: for any $v>0$,
$$
\begin{aligned}
&\int
\varphi(x-b;v)\,
\varphi(b;A)
\cb{\cos(\omega b)-\exp(-A\omega^2/2)}\,\dd b  \\
&\qquad =
\varphi(x;A+v)
\left[
\exp\left\{-\frac{1}{2}\omega^2\frac{Av}{A+v}\right\}
\cos\left(\omega\frac{A}{A+v}x\right)
-
\exp(-A\omega^2/2)
\right].
\end{aligned}
$$
Let $\Delta D(l)=D_{G_1}(l)-D_{G_0}(l)$ and
$\Delta N(z,l)=N_{G_1}(z,l)-N_{G_0}(z,l)$. Applying the identity with
$v=\tau^2$ gives
$$
\frac{\Delta D(l)}{D_{G_0}(l)}
=
\varepsilon
\left[
n^{-1/2}
\cos\left(\omega\frac{A}{A+\tau^2}l\right)
-
n^{-\frac{A+\tau^2}{2\tau^2}}
\right].
$$
In particular,
$$
\sup_{l\in\RR}
\left|
\frac{\Delta D(l)}{D_{G_0}(l)}
\right|
\le
2\varepsilon n^{-1/2}.
$$
Applying the same identity with $v=\nu$ gives
$$
\begin{aligned}
\Delta N(t,l)
&=
\varepsilon C
\int_{-\infty}^{t} e^{-\frac{u^2}{2\tilde\sigma^2}}
\varphi\left(l-\frac{\gamma\tau}{\tilde\sigma}u;A+\nu\right) 
\left[
n^{-\frac{r_A}{2}}
\cos\left(
\frac{\omega A}{A+\nu}
\left(l-\frac{\gamma\tau}{\tilde\sigma}u\right)
\right)
-
n^{-\frac{A+\tau^2}{2\tau^2}}
\right]\dd u .
\end{aligned}
$$
Thus
$$
\frac{\Delta N(t,l)}{D_{G_0}(l)}
=
\varepsilon n^{-r_A/2}B_n(l)
-
\varepsilon n^{-\frac{A+\tau^2}{2\tau^2}}F_{G_0}(t\mid l),
$$
where
$$
B_n(l)
:=
\int_{-\infty}^{t}
W_l(u)
\cos\left(
\omega\frac{A}{A+\nu}
\left(l-\frac{\gamma\tau}{\tilde\sigma}u\right)
\right)\dd u
$$
and
$$
W_l(u)
:=
\frac{C}{D_{G_0}(l)}
\exp\left(-\frac{u^2}{2\tilde\sigma^2}\right)
\varphi\left(l-\frac{\gamma\tau}{\tilde\sigma}u;A+\nu\right).
$$
We now lower bound $B_n(l)$ on a set of $l$'s that has probability bounded away from zero. Fix a compact interval $I=[-1,1]$. Uniformly over $l\in I$, the function $u\mapsto W_l(u)$ is smooth and has normal tails; moreover, for constants $0<c<C<\infty$ not depending on $n$,
$$
0<c\le W_l(t)\le C,
\qquad
|W_l'(t)|+\int_{-\infty}^{t}|W_l''(u)|\,\dd u\le C .
$$
Write, only for this calculation,
$$
\phi_l(u)
=
\omega\frac{A}{A+\nu}
\left(l-\frac{\gamma\tau}{\tilde\sigma}u\right).
$$
Since
$$
\frac{\dd}{\dd u}\sin\{\phi_l(u)\}
=
-\frac{\omega A\gamma\tau}{(A+\nu)\tilde\sigma}
\cos\{\phi_l(u)\},
$$
integration by parts gives
$$
B_n(l)
=
-
\frac{W_l(t)\sin\{\phi_l(t)\}}
{\omega A\gamma\tau/\{(A+\nu)\tilde\sigma\}}
+
\frac{1}
{\omega A\gamma\tau/\{(A+\nu)\tilde\sigma\}}
\int_{-\infty}^{t}W_l'(u)\sin\{\phi_l(u)\}\,\dd u .
$$
Applying integration by parts once more to the remaining integral, using
\[
\frac{\dd}{\dd u}\cos\{\phi_l(u)\}
=
\frac{\omega A\gamma\tau}{(A+\nu)\tilde\sigma}
\sin\{\phi_l(u)\},
\]
we get, uniformly over $l\in I$,
$$
\begin{aligned}
B_n(l)
&=
-
\frac{W_l(t)\sin\{\phi_l(t)\}}
{\omega A\gamma\tau/\{(A+\nu)\tilde\sigma\}}  
+
\frac{W_l'(t)\cos\{\phi_l(t)\}
-
\int_{-\infty}^{t}W_l''(u)\cos\{\phi_l(u)\}\,\dd u}
{\left(\omega A\gamma\tau/\{(A+\nu)\tilde\sigma\}\right)^2}  \\
&=
-
\frac{W_l(t)\sin\{\phi_l(t)\}}
{\omega A\gamma\tau/\{(A+\nu)\tilde\sigma\}}
+
O(\omega^{-2}).
\end{aligned}
$$
Here the boundary terms at $-\infty$ vanish because $W_l$ and $W_l'$ have normal tails.

The sine is not uniformly large, so define
$$
S_n
=
\left\{
l\in I:
\left|
\sin\left(
\omega\frac{A}{A+\nu}
\left(l-\frac{\gamma\tau}{\tilde\sigma}t\right)
\right)
\right|
\ge \frac12
\right\}.
$$
On $S_n$, the leading term has size at least a constant multiple of $\omega^{-1}$, while the remainder is $O(\omega^{-2})$. Hence, for all sufficiently large $n$,
$$
|B_n(l)|
\gtrsim
\omega^{-1}
\gtrsim
\frac{1}{\sqrt{\log n}},
\qquad l\in S_n.
$$
Finally, as a function of $l$, the sine above has period of order $\omega^{-1}$, and on each full period the set where its absolute value is at least $1/2$ occupies a fixed positive fraction of the period. Since $\omega\to\infty$, $I$ contains many periods for large $n$. Since under $G_0$ we have $\hat b_i\sim \mathrm N(0,A+\tau^2)$ and this density is bounded below on $I$, there is a constant $p>0$ such that
$$
\PP[G_0]{\hat b_i\in S_n}\ge p
$$
for all sufficiently large $n$.

Combining the expressions for $\Delta D(l)$ and $\Delta N(t,l)$, we get
$$
\begin{aligned}
F_{G_1}(t\mid l)-F_{G_0}(t\mid l)
&=
\frac{
\Delta N(t,l)/D_{G_0}(l)
-
F_{G_0}(t\mid l)\Delta D(l)/D_{G_0}(l)
}{
1+\Delta D(l)/D_{G_0}(l)
} \\
&=
\frac{
\varepsilon n^{-r_A/2}B_n(l)
-
\varepsilon n^{-1/2}
F_{G_0}(t\mid l)
\cos\left(\omega\frac{A}{A+\tau^2}l\right)
}{
1+\Delta D(l)/D_{G_0}(l)
}.
\end{aligned}
$$
The denominator is bounded away from zero for all sufficiently large $n$, while the second term in the numerator is negligible relative to
$n^{-r_A/2}/\sqrt{\log n}$ because $r_A<1$. Hence, for all sufficiently large $n$ and all $l\in S_n$,
$$
\left|
F_{G_1}(t\mid l)-F_{G_0}(t\mid l)
\right|
\gtrsim
\varepsilon
\frac{n^{-r_A/2}}{\sqrt{\log n}}.
$$
It remains to verify that the two experiments are close. Let $P_j$ denote the
joint law of $(\hat b_1,\ldots,\hat b_n)$ under $G_j$, $j=0,1$, and let
$P_{j,1}$ denote the corresponding one-coordinate marginal law. The density of
$P_{j,1}$ is $D_{G_j}$. Therefore,
$$
\chi^2(P_{1,1},P_{0,1})
=
\int
\left(
\frac{D_{G_1}(l)}{D_{G_0}(l)}-1
\right)^2
D_{G_0}(l)\,\dd l .
$$
From the expression above for $\Delta D(l)=D_{G_1}(l)-D_{G_0}(l)$,
$$
\frac{D_{G_1}(l)}{D_{G_0}(l)}-1
=
\varepsilon
\left[
n^{-1/2}
\cos\left(\omega\frac{A}{A+\tau^2}l\right)
-
n^{-\frac{A+\tau^2}{2\tau^2}}
\right].
$$
Under $P_{0,1}$, we have $\hat b_i\sim \mathrm N(0,A+\tau^2)$. Hence,
using $\EE{\cos(sX)}=\exp(-s^2\VarInline{X}/2)$ for a centered normal random
variable $X$,
$$
\EE[P_{0,1}]{
\cos\left(\omega\frac{A}{A+\tau^2}\hat b_i\right)
}
=
\exp\left\{
-\frac12
\omega^2\frac{A^2}{A+\tau^2}
\right\}
=
n^{-A/(2\tau^2)}.
$$
Similarly,
$$
\EE[P_{0,1}]{
\cos^2\left(\omega\frac{A}{A+\tau^2}\hat b_i\right)
}
=
\frac12
+
\frac12
\EE[P_{0,1}]{
\cos\left(2\omega\frac{A}{A+\tau^2}\hat b_i\right)
}  
=
\frac12
\left(
1+n^{-2A/\tau^2}
\right).
$$
Therefore,
$$
\begin{aligned}
\chi^2(P_{1,1},P_{0,1})
&=
\varepsilon^2
\EE[P_{0,1}]{
\left[
n^{-1/2}
\cos\left(\omega\frac{A}{A+\tau^2}\hat b_i\right)
-
n^{-\frac{A+\tau^2}{2\tau^2}}
\right]^2
} \\
&=
\varepsilon^2
\left[
\frac{1}{2n}\left(1+n^{-2A/\tau^2}\right)
-
2n^{-1/2-\frac{A+\tau^2}{2\tau^2}}n^{-A/(2\tau^2)}
+
n^{-\frac{A+\tau^2}{\tau^2}}
\right] \\
&=
\frac{\varepsilon^2}{2n}
\left(1-n^{-A/\tau^2}\right)^2
\le
\frac{\varepsilon^2}{2n}.
\end{aligned}
$$
Since the coordinates are independent under both $P_0$ and $P_1$,
$$
1+\chi^2(P_1,P_0)
=
\left\{1+\chi^2(P_{1,1},P_{0,1})\right\}^n.
$$
Thus
$$
\chi^2(P_1,P_0)
\le
\left(1+\frac{\varepsilon^2}{2n}\right)^n-1
\le
\exp(\varepsilon^2/2)-1.
$$
Consequently,
$$
\mathrm{TV}(P_0,P_1)
\le
\left\{\chi^2(P_1,P_0)\right\}^{1/2}
\le
\left\{\exp(\varepsilon^2/2)-1\right\}^{1/2}.
$$
Choosing $\varepsilon>0$ small enough ensures that
$\mathrm{TV}(P_0,P_1)\le p/2$.

Now lift $S_n$ to an event in the full $n$-dimensional sample space:
$
\mathcal E_{n,i}
=
\left\{
(l_1,\ldots,l_n)\in\RR^n:l_i\in S_n
\right\}.
$
Then $P_0(\mathcal E_{n,i})\ge p$. We write $P_0\wedge P_1$ for the common part of the two measures; if $p_0,p_1$ are their densities, then $P_0\wedge P_1$ has density $\min\{p_0,p_1\}$. Hence
$$
(P_0\wedge P_1)(\mathcal E_{n,i})
\ge
P_0(\mathcal E_{n,i})-\mathrm{TV}(P_0,P_1)
\ge
p/2.
$$
Let $\widehat\psi(t)$ be any measurable function of
$(\hat b_1,\ldots,\hat b_n)$. Pointwise,
$$
\left|\widehat\psi(t)-F_{G_0}(t\mid l_i)\right|
+
\left|\widehat\psi(t)-F_{G_1}(t\mid l_i)\right|
\ge
\left|F_{G_1}(t\mid l_i)-F_{G_0}(t\mid l_i)\right|.
$$
By Le Cam's two-point argument,
$$
\begin{aligned}
&\max_{j\in\{0,1\}}
\EE[G_j]{
\left|
\widehat\psi(t)-F_{G_j}(t\mid \hat b_i)
\right|
}  \\
&\qquad\ge
\frac12
\int
\left|
F_{G_1}(t\mid l_i)-F_{G_0}(t\mid l_i)
\right|
\,\dd(P_0\wedge P_1)(l_1,\ldots,l_n) \\
&\qquad\ge
\frac12
\int_{\mathcal E_{n,i}}
\left|
F_{G_1}(t\mid l_i)-F_{G_0}(t\mid l_i)
\right|
\,\dd(P_0\wedge P_1)(l_1,\ldots,l_n) \\
&\qquad\gtrsim
\varepsilon
\frac{n^{-r_A/2}}{\sqrt{\log n}}
(P_0\wedge P_1)(\mathcal E_{n,i})
\gtrsim
\frac{n^{-r_A/2}}{\sqrt{\log n}}.
\end{aligned}
$$
Since $r_A/2<\beta$,
$$
\frac{n^{-r_A/2}}{\sqrt{\log n}}
\gtrsim
n^{-\beta}
$$
for all sufficiently large $n$. Since $G_0,G_1\in\mathcal G_\Gamma$, this implies
$$
\inf_{\widehat\psi(t)}
\sup_{G\in\mathcal G_\Gamma}
\EE[G]{
\left|
\widehat\psi(t)-F_{G,i}(t\mid\hat b_i)
\right|
}
\ge
c n^{-\beta}
$$
for some $c>0$, after reducing $c$ if necessary to handle finitely many small values of $n$. This proves the claim.
\end{proof}

\subsection{Proof of Theorem~\ref{thm:rate_coverage}}

    Fix $i\in\{1,...,n\}$. We consider the testing problem $H_{0i}:\theta_i=\theta_{i0}$, and take $\theta_{i0}=\theta_{i}$ to be the true value. Under $H_{0i}$, we have 
    $$
    \binom{\hat{\theta}^{\mathrm{db}}_i-\theta_i}{\hat{b}_i}\;\bigg|\; b_i
    \sim
    \mathrm{N}\left\{
    \binom{0}{b_i},
    \Sigma_i
    \right\}.
    $$
    By definition of $F_{G,i}(\cdot\mid \hat{b}_i)$, the conditional CDF of $\hat{\theta}^{\mathrm{db}}_i-\theta_i$ given $\hat{b}_i$ under the oracle model is $F_{G,i}(\cdot\mid \hat{b}_i)$. Therefore, by the probability integral transform, 
    $$
    U_i:=F_{G,i}(\hat{\theta}^{\mathrm{db}}_i-\theta_i\mid \hat{b}_i) \sim\mathrm{Unif}(0,1).
    $$
    Since $\widehat G$ is computed from $\{\hat{b}_1,\ldots,\hat{b}_n\}$, we may also condition $ \mathcal L_n:=\sigma(\hat{b}_1,\ldots,\hat{b}_n)$, and $U_i\mid \mathcal L_n\sim \mathrm{Unif}(0,1).$

    Now we define the plug-in conditional probability transform
    $$
     \widehat U_i
    :=
    F_{\widehat G,i}(\hat{\theta}^{\mathrm{db}}_i-\theta_i\mid \hat{b}_i),
    $$
    and we have that 
    $$
    |\widehat{U}_i-U_i| \leq \sup_{z\in\mathbb{R}}\left|F_{\widehat{G},i}(z\mid\hat{b}_i)-F_{G,i}(z\mid\hat{b}_i)\right|=:\Delta_i.
    $$
    By the definition of $\mathcal{I}_{\widehat{G},i}(1-\alpha)$, 
    $$
    \{\theta_i\in \mathcal{I}_{\widehat{G},i}(1-\alpha)\}
    =
    \left\{
    \frac{\alpha}{2}
    \le
    F_{\widehat G,i}(\hat\theta_i^{\mathrm{db}}-\theta_i\mid \hat b_i)
    \le
    1-\frac{\alpha}{2}
    \right\} =
    \left\{
    \frac{\alpha}{2}
    \le
    \widehat U_i
    \le
    1-\frac{\alpha}{2}
    \right\}.
    $$
    Combined with $|\widehat{U}_i-U_i| \leq\Delta_i$, we have the following set inclusions
    $$
    \left\{
    \frac{\alpha}{2}+\Delta_i\le U_i\le 1-\frac{\alpha}{2}-\Delta_i\right\}
    \subseteq
    \{\theta_i\in \mathcal I_{\widehat{G},i} (1-\alpha)\}\subseteq\left\{\frac{\alpha}{2}-\Delta_i\le U_i \le 1-\frac{\alpha}{2}+\Delta_i
    \right\}.
    $$
    Taking conditional probabilities given \(\mathcal L_n\), and using
    \(U_i\mid\mathcal L_n\sim \mathrm{Unif}(0,1)\), we obtain
    \[
        1-\alpha-2\Delta_i
        \le
        \mathbb P_{\boldsymbol{\theta},G}\left[
        \theta_i\in \mathcal{I}_{\widehat{G},i}(1-\alpha)
        \mid \mathcal L_n
        \right]
        \le
        1-\alpha+2\Delta_i .
    \]
    Hence
    \[
        \left|\PP[\boldsymbol{\theta},G]{\theta_i\in \mathcal{I}_{\widehat{G},i}(1-\alpha)\mid \mathcal L_n}-(1-\alpha)\right|\le2\Delta_i .
    \]
    Then by taking the expectation on both sides, we have
    $$
    \begin{aligned}
        &\left|\PP[\boldsymbol{\theta},G]{
        \theta_i\in \mathcal{I}_{\widehat{G},i}(1-\alpha)}-(1-\alpha)\right|\\
        &~~=\left|\EE[\boldsymbol{\theta},G]{\mathbb P_{\boldsymbol{\theta},G}\left[
        \theta_i\in \mathcal{I}_{\widehat{G},i}(1-\alpha)
        \mid \mathcal L_n
        \right]-(1-\alpha)}\right|\\
        &~~\leq\EE[\boldsymbol{\theta},G]{\left|\mathbb P_{\boldsymbol{\theta},G}\left[
        \theta_i\in \mathcal{I}_{\widehat{G},i}(1-\alpha)
        \mid \mathcal L_n
        \right]
        -
        (1-\alpha)\right|}\leq2\EE[\boldsymbol{\theta},G]{\Delta_i}.
    \end{aligned}
    $$
    Averaging over $i$, we get
    $$
    \begin{aligned}
    \frac1n\sum_{i=1}^n
    \left|
    \PP[\boldsymbol{\theta},G]{
    \theta_i\in \mathcal{I}_{\widehat{G},i}(1-\alpha)} -(1-\alpha)\right| \le\frac2n\sum_{i=1}^n\mathbb E_{\boldsymbol{\theta},G}[\Delta_i]
    \le 2C\frac{(\log{n})^{3/2}}{n^{(1-\check{\gamma}^2)/2}},
    \end{aligned}
    $$
    where the last inequality comes from the proof of Theorem \ref{thm:covergence_pvalue}. This holds for all $\alpha\in(0,1)$, then the following completes the proof:
     $$
     \sup_{\alpha\in(0,1)}
    \begin{aligned}
    \frac1n\sum_{i=1}^n
    \left|\PP[\boldsymbol{\theta},G]{\theta_i\in \mathcal{I}_{\widehat{G},i}(1-\alpha)}
    -(1-\alpha)\right| \le C^{\prime}\frac{(\log{n})^{3/2}}{n^{(1-\check{\gamma}^2)/2}}.
    \end{aligned}
    $$

\subsection{Proof of Theorem~\ref{thm:rate_coverage_joint}}
The argument follows the proof of Theorem~\ref{thm:rate_coverage}. Under the generalized model, conditional on $(\theta_i,b_i)$, the joint distribution of
$(\hat\theta_i^{\db}-\theta_i,\hat b_i)$
depends on $(\theta_i,b_i)$ only through $b_i$. Since the marginal distribution of $b_i$ under $H_i$ is the common $G$, the marginal joint distribution of
$(\hat\theta_i^{\mathrm{db}}-\theta_i,\hat b_i)$
is identical to that under the original partially Bayes model.

Consequently, the oracle conditional probability integral transform remains uniform given
$\mathcal L_n=\sigma(\hat b_1,\ldots,\hat b_n)$.
Moreover, the joint distribution of the bias estimates, and hence the distribution of $\Delta_i$ defined in the proof of Theorem~\ref{thm:rate_coverage}, is unchanged. Applying the same argument and the conditional CDF estimation bound established in the proof of Theorem~\ref{thm:covergence_pvalue} yields
\[
\sup_{\alpha\in(0,1)}
\frac1n\sum_{i=1}^n
\left|
\mathbb P_{\boldsymbol H}
\left[
\theta_i\in\mathcal I_{\widehat G,i}^{\mathrm{rb}}(1-\alpha)
\right]-(1-\alpha)
\right|
\le
C'\frac{(\log n)^{3/2}}{n^{(1-\check\gamma^2)/2}},
\]
as claimed.
\subsection{Proof of Theorem~\ref{thm:rate_average_coverage}}
\label{app:proof_compound}
For notational simplicity, we denote
\[
\Sigma=(\Sigma_{jq}):=
\begin{pmatrix}
\widetilde\sigma^2 & \gamma\widetilde\sigma\tau\\
\gamma\widetilde\sigma\tau & \tau^2
\end{pmatrix},
\quad
\Omega=(\Omega_{jq}):=\Sigma^{-1}.
\]
Define the empirical distribution of the fixed biases $\boldsymbol{b}=(b_1,...,b_n)$ by 
$$
G_{n}=\frac1n\sum_{i=1}^n\delta_{b_i}.
$$
Under the homogeneous model, we define
$$
d(G_n,\widehat{G}) := \mathfrak{H}\left(f_{G_n}(\cdot;\Sigma_{22}),f_{\wh G}(\cdot;\Sigma_{22})\right).
$$
\begin{lemm}[Theorem 1 of \cite{zhang2009generalized}] \label{lemm:compound_hellienger}
Suppose that $\hat{b}_1,...,\hat{b}_n$ are independent, $\hat b_i\sim \mathrm{N}(b_i,\tau^2)$, and $|b_i|\leq M$. Fix $c_0\geq 2$, then there exists a constant $C_0$ depending only on $M,\tau,c_0$ such that,
$$
    \PP[\boldb]{d(G_n,\widehat{G})\geq C_0\frac{\log{n}}{\sqrt{n}}}\leq \exp{\left(-c_0\log{n}\right)} \quad\text{for all } n \in \mathbb{N}_{\geq2}.
$$
\end{lemm}
\begin{proof}
   We verify the rate obtained from Theorem 1 of \cite{zhang2009generalized}. In the notation of that theorem, the observations $X_i$, conditional means $\theta_i$, and common standard deviation $\sigma$ correspond to $\hat b_i$, $b_i$, and $\tau$, respectively. Since the $\hat b_i$ are independent with $\hat b_i\sim \mathrm{N}(b_i, \tau^2)$, the sampling model satisfies (2.1) of \cite{zhang2009generalized}, with the average mixing distribution corresponding to $G_n$. Since $|b_i|\le M$ for all $i$, we have $G_n([-M,M])=1.$ Thus, applying the bounded-support case of Theorem 1 of \cite{zhang2009generalized}, with $M_n=M/\tau$ after standardization, yields
    $$
    \epsilon_n \lesssim n^{-1/2}(\log n)^{3/4}\left\{\left(\frac{M}{\tau}\right)^{1/2}\vee(\log n)^{1/4}\right\}\le C_1(M,\tau)\frac{\log n}{\sqrt n}.
    $$
    Therefore, by the concentration inequality (2.11) of \cite{zhang2009generalized}, for any fixed $c_0\ge2$, choosing $t\ge\max\{t_*,\sqrt{c_0}\}$, there exist a constant $C_0$ that depends only on $M,\tau,c_0$ such that
    $$
    \PP[\boldb]{d(G_n,\widehat G)\ge C_0\frac{\log n}{\sqrt n}}\le\exp(-c_0\log n)\quad \text{for all } n\in\mathbb{N}_{\geq2}.
    $$
    The difference in the normalization of the Hellinger distance is absorbed into the constant $C_0$.
\end{proof}
\begin{proof}[Proof of Theorem~\ref{thm:rate_average_coverage}]
As in the proof of Theorem~\ref{thm:covergence_pvalue}, we should first try to bound 
$$
\frac1n\sum_{i=1}^n
\mathbb E_{\boldb}[\Delta_i],\quad 
\Delta_i:=\sup_{z\in\mathbb R}\left|F_{\widehat G}(z\mid\hat b_i)-F_{G_n}(z\mid\hat b_i)\right|.
$$
Following the definition in Proposition~\ref{prop:f_formula}, under the homogeneous covariance parameters, for any probability measure $\widetilde{G}$,
$$
F_{\widetilde{G}}(z\mid l):=\frac{N_{\widetilde{G}}(z,l)}{D_{\widetilde{G}}(l)},
\qquad
D_{\widetilde{G}}(l):=f_{\widetilde{G}}(l;\Sigma_{22}),
$$
where 
\begin{equation}
\label{eq:homo_n_formula}
N_{\widetilde{G}}(z,l)= C(\Sigma)\int_{-\infty}^{z}\exp\left(-\frac{u^2}{2\Sigma_{11}}\right)f_{\widetilde{G}}\left(\frac{\Omega_{12}}{\Omega_{22}}u+l;\Omega_{22}^{-1}\right) \dd u,
\end{equation}
where $C(\Sigma)=(2\pi|\Sigma|\Omega_{22})^{-1/2}.$ Also note that $0\le N_{\widetilde{G}}(z,l)\le D_{\widetilde{G}}(l)$. For all $z,l\in\mathbb{R}$,
\[
\begin{aligned}
&|N_G(z,l)-N_{G_n}(z,l)|\\
&\quad\le C(\Sigma)\int_{-\infty}^{z}\exp\left(-\frac{u^2}{2\Sigma_{11}}\right)
\left|f_G\left(\frac{\Omega_{12}}{\Omega_{22}}u+l;\Omega_{22}^{-1}\right)-f_{G_n}\left(\frac{\Omega_{12}}{\Omega_{22}}u+l;\Omega_{22}^{-1}\right)\right|\dd u\\
&\quad \le C(\Sigma)\left(\int_{\mathbb R}\exp\left(-\frac{u^2}{\Sigma_{11}}\right)\dd u\right)^{1/2}\left(\int_{\mathbb R}
\left|f_G\left(\frac{\Omega_{12}}{\Omega_{22}}u+l;\Omega_{22}^{-1}\right)-f_{G_n}\left(\frac{\Omega_{12}}{\Omega_{22}}u+l;\Omega_{22}^{-1}\right)\right|^2\dd u\right)^{1/2}\\
&\quad\lesssim_{\Sigma}
\left\|f_G(\cdot;\Omega_{22}^{-1})-f_{G_n}(\cdot;\Omega_{22}^{-1})\right\|_2,
\end{aligned}
\]
where the last inequality follows from a change of variables, using $\Omega_{12}\ne0$. Similarly, proceeding as in the proof of Lemma~\ref{lemm:hetero_n_ratio_bound}, for every $T>0$, by the normal convolution identity and Plancherel's theorem
$$
\begin{aligned}
    &\left\|f_G(\cdot;\Omega_{22}^{-1})-f_{G_n}(\cdot;\Omega_{22}^{-1})\right\|_2^2 =\frac{1}{2\pi}\int_{-\infty}^{\infty}\left\vert f_{G}^{*}(\omega;\Omega_{22}^{-1})-f_{G_n}^*(\omega;\Omega_{22}^{-1})\right\rvert^2\dd\omega \\
    &=\frac{1}{2\pi}\int_{-\infty}^{\infty}e^{\Sigma_{22}\gamma^2\omega^2}\left\vert f_{G}^{*}(\omega;\Sigma_{22})-f_{G_n}^*(\omega;\Sigma_{22})\right\vert^2\dd\omega \\
    &\leq \frac{e^{\Sigma_{22}\gamma^2T^2}}{2\pi}\int_{\mathbb{R}}\left\vert f_{G}^{*}(\omega;\Sigma_{22})-f_{G_n}^*(\omega;\Sigma_{22})\right\vert^2\dd \omega +\frac{1}{2\pi}\int_{[-T,T]^c}e^{\Sigma_{22}\gamma^2\omega^2-\Sigma_{22}\omega^2}|G^*(\omega)-G_n^*(\omega)|^2 \dd\omega\\
    &\leq e^{\Sigma_{22}\gamma^2T^2}\left\|f_G(\cdot;\Sigma_{22})-f_{G_n}(\cdot;\Sigma_{22})\right\|_2^2 +\frac4{2\pi}\int_{[-T,T]^c}e^{-\Sigma_{22}(1-\gamma^2)\omega^2}\dd\omega,
\end{aligned}
$$
since $|G^*(\omega)|\le1$ and $|G_n^*(\omega)|\le1$. Also, we have
$$
\left\|f_G(\cdot;\Sigma_{22})-f_{G_n}(\cdot;\Sigma_{22})\right\|_2^2\lesssim_{\Sigma} \mathfrak{H}^2(f_{G}(\cdot;\Sigma_{22}),f_{G_n}(\cdot;\Sigma_{22})),
$$
and
$$
\int_{[-T,T]^c}e^{-\Sigma_{22}(1-\gamma^2)\omega^2}\,\dd\omega\le\frac{e^{-\Sigma_{22}(1-\gamma^2)T^2}}{\Sigma_{22}(1-\gamma^2)T}.
$$
Taking $T=(\log n/\Sigma_{22})^{1/2}$, then
$$
\sup_{z,l\in\mathbb R}|N_G(z,l)-N_{G_n}(z,l)| \lesssim_{\Sigma} n^{\gamma^2/2} \mathfrak{H}(f_{G}(\cdot;\Sigma_{22}),f_{G_n}(\cdot;\Sigma_{22})) + \frac{n^{-(1-\gamma^2)/2}}{(\log{n})^{1/4}},
$$
and it holds uniformly over $G$. Define $\mathcal A=\{d(G_n,\widehat G)\le C_0\log n/\sqrt n\}$, by Lemma~\ref{lemm:compound_hellienger}, then on $\mathcal{A}$,
$$
\sup_{z,l\in\mathbb R}
|N_{\widehat G}(z,l)-N_{G_n}(z,l)|
\lesssim_{\Sigma}\frac{\log n}{n^{(1-\gamma^2)/2}}.
$$
Since $N_G(+\infty,l)=D_G(l)$, then for any fixed $l$,
$$
|D_{\widehat G}(l)-D_{G_n}(l)|=\lim_{z\to+\infty}
|N_{\widehat G}(z,l)-N_{G_n}(z,l)|\le\sup_{z\in\mathbb R}|N_{\widehat G}(z,l)-N_{G_n}(z,l)|,
$$
we also obtain
$$
\sup_{l\in\mathbb R}|D_{\widehat G}(l)-D_{G_n}(l)| \lesssim_{\Sigma}\frac{\log n}{n^{(1-\gamma^2)/2}}.
$$
Following the decomposition in Appendix~\ref{app:proof_theorem2}, define $\wh G_{*}=(\wh G + G_n)/2$, and 
$$
\begin{aligned}
    \left|F_{\wh G}(z\mid l)-F_{G_n}(z\mid l)\right| \le \frac{|N_{\widehat G}(z,l)-N_{G_n}(z,l)|+|D_{\widehat G}(l)-D_{G_n}(l)|}{D_{\widehat G_*}(l)}.
\end{aligned}
$$
Since $D_{\wh G_*}(l)=(D_{\widehat G}(l)+f_{G_n}(l;\Sigma_{22}))/2\ge f_{G_n}(l;\Sigma_{22})/2$, on the event $\mathcal{A}$,
$$
\sup_{z\in\mathbb R}
|F_{\widehat G}(z\mid l)-F_{G_n}(z\mid l)|
\lesssim_\Sigma \frac{\log n}{n^{(1-\gamma^2)/2}}\frac1{f_{G_n}(l;\Sigma_{22})}.
$$
Taking $B_n = M+\sqrt{2\Sigma_{22}\log n}$, we have
\[
\begin{aligned}
\frac1n\sum_{i=1}^n
\mathbb E_{\boldb}[\Delta_i]
&\lesssim_{\Sigma}
\PP[\boldb]{\mathcal A^c}
+\frac1n\sum_{i=1}^n
\PP[\boldb]{|\hat b_i|>B_n}
+\frac{\log n}{n^{(1-\gamma^2)/2}}
\frac1n\sum_{i=1}^n
\mathbb E_{\boldb}
\left[
\frac{\mathds 1(|\hat b_i|\le B_n)}
{f_{G_n}(\hat b_i;\Sigma_{22})}
\right].
\end{aligned}
\]
Notice that the average marginal density of $\hat{b}_i$ is $f_{G_n}(\cdot;\Sigma_{22})$, then
\[
\frac1n\sum_{i=1}^n\mathbb E_{\boldb}
\left[\frac{\mathds1(|\hat b_i|\le B_n)}{f_{G_n}(\hat b_i;\Sigma_{22})}
\right]=\int_{-B_n}^{B_n}\frac{n^{-1}\sum_{i=1}^n\varphi(l-b_i;\Sigma_{22})}{f_{G_n}(l;\Sigma_{22})}\,\dd l=2B_n.
\]
Moreover, since $|b_i|\leq M$, $\mathbb P_{\boldb}[|\hat b_i|>B_n]\le2/n.$ Hence,
$$
\frac1n\sum_{i=1}^n
\mathbb E_{\boldb}[\Delta_i] \lesssim_{\Sigma,M} \frac{1}{n^{c_0}} + \frac2n + \frac{B_n\log n}{n^{(1-\gamma^2)/2}} \lesssim_{\Sigma,M}  \frac{(\log n)^{3/2}}{n^{(1-\gamma^2)/2}}.
$$
Now we need to move to the coverage result. Following the proof of Theorem~\ref{thm:rate_coverage}, we define
\[
U_i:=F_{G_n}(\hat\theta_i^{\db}-\theta_i\mid\hat b_i),
\qquad \widehat U_i:=F_{\widehat G}(\hat\theta_i^{\db}-\theta_i\mid\hat b_i).
\]
By definition, $|\widehat U_i-U_i|\le\Delta_i$. For any $t\in[0,1]$,
\[
\begin{aligned}
\frac1n\sum_{i=1}^n\PP[\boldtheta,\boldb]{U_i\le t}
&=\frac1n\sum_{i=1}^n
\PP[\theta_i,b_i]{
F_{G_n}(\hat\theta_i^{\db}-\theta_i\mid\hat b_i)\le t}\\
&=\int\PP[\theta,b]{F_{G_n}(\hat\theta^{\db}-\theta
\mid\hat b)\le t} G_n(\dd b)=t,
\end{aligned}
\]
where the last equality follows from the probability integral transform under the mixture model with mixing distribution $G_n$. Similar to the set inclusions in the proof of Theorem~\ref{thm:rate_coverage}, for any fixed $\delta>0$, we have
\[
\begin{aligned}
&\left\{\frac{\alpha}{2}+\delta\le U_i\le1-\frac{\alpha}{2}-\delta\right\}\cap\{\Delta_i\le\delta\}\\
&\qquad\subseteq
\left\{\theta_i\in\mathcal I_{\widehat G,i}^{\mathrm{rb}}(1-\alpha)\right\}\\
&\qquad\subseteq\left\{\frac{\alpha}{2}-\delta\le U_i\le1-\frac{\alpha}{2}+\delta\right\}\cup\{\Delta_i>\delta\}.
\end{aligned}
\]
Taking probabilities, averaging over $i$, and using the average uniformity, we have
\[
\begin{aligned}
&\left|
\frac1n\sum_{i=1}^n\PP[\boldtheta,\boldb]{\theta_i\in\mathcal I_{\widehat G,i}^{\mathrm{rb}}(1-\alpha)}-(1-\alpha)
\right|\le
2\delta+\frac1n\sum_{i=1}^n\PP[\boldb]{\Delta_i>\delta}.
\end{aligned}
\]
Consequently, by Markov's inequality,
$$
\left|
\frac1n\sum_{i=1}^n\PP[\boldtheta,\boldb]{\theta_i\in\mathcal I_{\widehat G,i}^{\mathrm{rb}}(1-\alpha)}-(1-\alpha)
\right|\leq 2\delta+\frac{1}{n\delta}\sum_{i=1}^n\EE[\boldb]{\Delta_i}\lesssim_{\Sigma,M}2\delta+\frac{1}{\delta} \frac{(\log n)^{3/2}}{n^{(1-\gamma^2)/2}}.
$$
Taking $\delta = (\log n)^{3/4}/n^{(1-\gamma^2)/4}$, we obtain
$$
\sup_{\alpha\in(0,1)}
\left|
\frac1n\sum_{i=1}^n
\mathbb P_{\boldtheta,\boldb}
\left[
\theta_i\in
\mathcal I_{\widehat G,i}^{\mathrm{rb}}(1-\alpha)
\right]
-(1-\alpha)
\right|
\le
C'\frac{(\log n)^{3/4}}
{n^{(1-\gamma^2)/4}}.
$$
\end{proof}

\subsection{Proof of Theorem~\ref{thm:robust_w_exchangeability}}

For notation simplicity, we denote $Y_i =\hat{\theta}_i^{\db}-\theta_i$, $A_i=\hat{b}_i$. For $\widetilde G\in\mathcal{P}([-M,M])$ and an admissible covariance matrix $S$, the
\emph{Normal working model} is
$$
b_i\sim \widetilde G,\qquad (Y_i^G,A_i^G)\mid b_i
\sim \mathrm{N}_2((0,b_i)^\top,S).
$$
In particular, taking $\widetilde G=G$ and $S=\Sigma_i$ yields the matched normal experiment corresponding to the reference distribution in $\mathcal W_n$.
Let \(F_{\widetilde G}(y\mid a;S)\) denote the conditional CDF of \(Y_i^G\mid A_i^G=a\) under the normal working model with mixing distribution \(\widetilde G\) and covariance matrix \(S\). We write \(q_{\widetilde G,\alpha}(a;S)\) for its \(\alpha\)-quantile, then the rebiased interval proposed in \eqref{eq:CI_npmle} can be expressed as
$$
\mathcal I^{\rb}_{\widehat G,i}(1-\alpha)
=\bigl[\hat\theta_i^{\mathrm{db}}-
q_{\widehat G,1-\alpha/2}(A_i;\widehat\Sigma_i),
\hat\theta_i^{\mathrm{db}}-
q_{\widehat G,\alpha/2}(A_i;\widehat\Sigma_i)\bigr].
$$

We first collect several direct consequences of the assumptions that will be used repeatedly in the proof.
\begin{lemm}\label{lemm:consequence_assm}
Let $\mu_i=\mathcal L_{\boldsymbol{\theta},G}(A_i)$, and let $\nu_i$ have density
$f_G(\cdot;\Sigma_{i,22})$. Note that
$\tau_i^2=\Sigma_{i,22}$, and let
$\hat\tau_i^2=\widehat\Sigma_{i,22}$. Then under assumptions of Theorem~\ref{thm:robust_w_exchangeability},

\begin{align}
V_n:=\frac1n\sum_{i=1}^n W_1(\mu_i,\nu_i)&\le\mathcal W_n\to0,
\label{eq:Vbound}\\
\lim_{T\to\infty}\limsup_{n\to\infty}\frac1n\sum_{i=1}^n
\EE[\boldsymbol{\theta},G]{A_i^2\mathds{1}(|A_i|> T)}&=0,
\label{eq:tailbound}\\
\frac1n\sum_{i=1}^n|\widehat\tau_i-\tau_i|&\xrightarrow{\mathbb{P}_{\boldsymbol{\theta},G}}0.
\label{eq:stdev}
\end{align}
Here \(W_p\) denotes the Wasserstein-\(p\) distance.
\end{lemm}
\begin{proof}
For each $i$, conditionally on $b_i=b$, take an optimal bivariate $W_2$ coupling between $\mathcal L_{\boldsymbol\theta,G}((Y_i,A_i)\mid b_i=b)$ and $\mathrm N_2((0,b)^\top,\Sigma_i)$. Let $(Y_i^G,A_i^G)$ denote the normal random vector in this coupling, and let $d_i(b)$ be the corresponding Wasserstein-2 distance. Projection gives
\begin{align*}
|A_i-A_i^G|^2 &\le\|(Y_i,A_i)-(Y_i^G,A_i^G)\|^2,\\
W_1(\mu_i,\nu_i) &\le\int W_1\bigl(\mathcal L(A_i\mid b_i=b),\mathrm{N}(b,\tau_i^2)\bigr)\,G(\dd b) \le\int d_i(b)\,G(\dd b),
\end{align*}
since $\EE{|A_i-A_i^G|}\le\sqrt{\EE{|A_i-A_i^G|^2}}\le d_i(b).$ Cauchy--Schwarz now yields
\[
V_n\le\sqrt{\frac1n\sum_{i=1}^n\int d_i(b)^2\,G(\dd b)}=\mathcal W_n.
\]
Under the same individual couplings, $n^{-1}\sum_{i=1}^n\EE{|A_i-A_i^G|^2}\le\mathcal W_n^2$. Since $A_i^G\mid b_i\sim\mathrm N(b_i,\tau_i^2)$, its marginal distribution admits the representation $A_i^G = b_i+\tau_i Z_i$, where $Z_i\sim\mathrm N(0,1)$ is independent of $b_i$, and $|b_i|\leq M$ almost surely. Hence
$$
\begin{aligned}
\frac1n\sum_{i=1}^n\EE{A_i^2\mathds 1(|A_i|>T)}
&\le\frac4n\sum_{i=1}^n\EE{|A_i-A_i^G|^2}
+\frac2n\sum_{i=1}^n\EE{(A_i^G)^2\mathds 1(|A_i^G|>T/2)}\\
&\le4\mathcal W_n^2+\frac8{T^2}\sup_i\EE{(A_i^G)^4}
\le4\mathcal W_n^2+\frac{64(M^4+3\bar\lambda^2)}{T^2},
\end{aligned}
$$
where the first inequality follows from $x^2\mathds 1(|x|>T)\le4(x-z)^2+2z^2\mathds1(|z|>T/2)$. The last uses $(u+v)^4\le8(u^4+v^4)$, therefore $\EE{(A_i^G)^4} \le 8(\EE{b_i^4}+\tau_i^4 \EE{Z^4}) \leq 8(M^4 +  3\bar\lambda^2).$ Taking $\limsup_n$, then $T\to\infty$, proves \eqref{eq:tailbound}.
Finally,
\[
\frac1n\sum_{i=1}^n|\widehat\tau_i-\tau_i|
=\frac1n\sum_{i=1}^n
\frac{|\widehat\Sigma_{i,22}-\Sigma_{i,22}|}
{\widehat\tau_i+\tau_i}
\le\frac1{2\sqrt{\underline\lambda}\,n}
\sum_{i=1}^n\|\widehat\Sigma_i-\Sigma_i\|
\xrightarrow{\mathbb{P}_{\boldsymbol{\theta},G}}0.
\]

\end{proof}
We next establish the mixing-distribution consistency needed for the coverage proof.
\begin{lemm}\label{lemm:consistency}
Under assumptions of Theorem~\ref{thm:robust_w_exchangeability}, we have $W_1(\widehat G,G)\xrightarrow{\mathbb{P}_{\boldsymbol{\theta},G}}0$. Here $W_1$ denotes the Wasserstein-1 distance on $\mathcal P ([-M,M])$.
\end{lemm}
\begin{proof}
Let 
$$
L_n^{\mathrm{true}}(\widetilde G)=\frac1n\sum_{i=1}^n\log f_{\widetilde G}(A_i;\Sigma_{i,22}),\qquad
L_n^{\mathrm{est}}(\widetilde G)=\frac1n\sum_{i=1}^n\log f_{\widetilde G}(A_i;\widehat \Sigma_{i,22}).
$$
By the global optimality of $\widehat G$ and since $G\in\mathcal P([-M,M])$, we have
$$
L_n^{\mathrm{est}}(\widehat G)\ge L_n^{\mathrm{est}}(G).
$$For $v\in[\underline\lambda,\overline\lambda]$,
\[
\partial_v\log f_{\widetilde G}(a;v)
=\EE[b\sim\widetilde G_{a,v}]{-\frac1{2v}+\frac{(a-b)^2}{2v^2}},
\quad \widetilde G_{a,v}(\dd b)=\frac{\varphi(a-b;v)\widetilde G(\dd b)}{f_{\widetilde G}(a;v)}.
\]
Since \(\widetilde G\in\mathcal P([-M,M])\) and \(v\in[\underline\lambda,\overline\lambda]\), we have $\left|\partial_v\log f_{\widetilde G}(a;v)\right|\le C(1+a^2).$ For every realized $(A_i,\widehat \Sigma_{i,22})$ and $T>1$,
$$
\begin{aligned}
R_n&:=\sup_{\widetilde G\in\mathcal{P}([-M,M])}
|L_n^{\mathrm{est}}(\widetilde G)-L_n^{\mathrm{true}}(\widetilde G)|\\
&\le\frac Cn\sum_{i=1}^n(1+A_i^2)|\widehat \Sigma_{i,22}-\Sigma_{i,22}|\\
&\le C(1+T^2)\frac1n\sum_{i=1}^n|\widehat \Sigma_{i,22}-\Sigma_{i,22}|
+\frac{2C\overline\lambda}{n}\sum_{i=1}^n(1+A_i^2)
\mathds 1(|A_i|>T).
\end{aligned}
$$
By Lemma~\ref{lemm:consequence_assm}, for every fixed \(T>1\), the first term is \(o_{\mathbb P}(1)\). For the second term, notice that $(1+A_i^2)
\mathds 1(|A_i|>T)\leq 2A_i^2\mathds 1(|A_i|>T)$ and use Markov's inequality and \eqref{eq:tailbound}, we have
\[
\lim_{T\to\infty}\limsup_{n\to\infty}
\PP{\frac1n\sum_i(1+A_i^2)\mathds 1(|A_i|>T)>\varepsilon}=0
\]
for every \(\varepsilon>0\). Choosing \(T\) sufficiently large and then letting \(n\to\infty\) yields \(R_n=o_{\mathbb P_{\boldsymbol{\theta},G}}(1)\). Consequently, by the global optimality of $\widehat G$, $L_n^{\mathrm{true}}(\widehat G)-L_n^{\mathrm{true}}(G)
\ge -2R_n=-o_{\mathbb P_{\boldsymbol\theta,G}}(1).$

Set $g_{i,\widetilde G}(a)=\log\{f_{\widetilde G}(a;\Sigma_{i,22})/f_G(a;\Sigma_{i,22})\}$ and
$m_{\widetilde G,a,v}=\EE[\widetilde G_{a,v}]{b}$. Posterior differentiation gives
\[
\partial_a\log f_{\widetilde G}(a;v)=\frac{m_{\widetilde G,a,v}-a}{v},
\qquad
\partial_a^2\log f_{\widetilde G}(a;v)
=\frac{\operatorname{Var}_{\widetilde G_{a,v}}(b)}{v^2}-\frac1v.
\]
Since \(\widetilde G_{a,v}\) is supported on \([-M,M]\), we have $|m_{\widetilde G,a,v}|\le M$, $\operatorname{Var}_{\widetilde G_{a,v}}(b)\le M^2.$ Therefore, for every \(\widetilde G\in\mathcal P([-M,M])\),
\[
|g'_{i,\widetilde G}(0)| \le \frac{2M}{\underline\lambda},
\qquad \|g'_{i,\widetilde G}\|_\infty \le \frac{2M}{\underline\lambda},
\qquad \|g''_{i,\widetilde G}\|_\infty \le \frac{2M^2}{\underline\lambda^2}.
\]
Since Assumption~\ref{assm:robust_independence} allows deterministic triangular arrays of functions, it applies to $g_{i,\widetilde G}/C_0$ for each fixed $\widetilde G$, with $C_0=\max\{1,2M/\underline\lambda,2M^2/\underline\lambda^2\}$, and it gives 
\begin{equation}
\label{eq:pointwise_convergence}
    Z_n(\widetilde G):=\frac1n\sum_{i=1}^n\left(g_{i,\widetilde G}(A_i)-\EE[\boldsymbol{\theta},G]{g_{i,\widetilde G}(A_i)}\right)
\xrightarrow{\mathbb{P}_{\boldsymbol{\theta},G}}0.
\end{equation}
We next upgrade the pointwise convergence in \eqref{eq:pointwise_convergence} to uniform convergence over \(\widetilde G\in\mathcal P([-M,M])\), we separately control the compact region \(|A_i|\le T\) and the tails. On the compact region $|a|\le T$, when $|b|\le M$, $v\in[\underline\lambda,
\overline\lambda]$, set
$d_T=\inf\varphi(a-b;v)>0$ and
$C_T=\sup|\partial_b\varphi(a-b;v)|<\infty$. Then, by the Kantorovich–Rubinstein duality for \(W_1\), $|f_{\widetilde G}(a;v)-f_{\widetilde G'}(a;v)|
\le
C_T W_1(\widetilde G,\widetilde G').$ Since \(f_{\widetilde G}(a;v),f_{\widetilde G'}(a;v)\ge d_T\),
\[
\left|\log f_{\widetilde G}(a;v)-\log f_{\widetilde G'}(a;v)
\right|\le
\frac1{d_T}\left|f_{\widetilde G}(a;v)-f_{\widetilde G'}(a;v)\right|\le
\frac{C_T}{d_T}W_1(\widetilde G,\widetilde G').
\]
Hence
\[
\sup_{i,|a|\le T}|g_{i,\widetilde G}(a)-g_{i,\widetilde G'}(a)|
\le\frac{C_T}{d_T}W_1(\widetilde G,\widetilde G').
\]
For the tail region, since \(\Sigma_{i,22}\in[\underline\lambda,\overline\lambda]\) and \(G,\widetilde G\in\mathcal P([-M,M])\), we have $
\sup_{i,\widetilde G}|g_{i,\widetilde G}(0)|<\infty,
$ and $\sup_{i,\widetilde G}\|g'_{i,\widetilde G}\|_\infty
\le 2M/\underline\lambda.
$ Hence, for some constant \(C_1<\infty\),
\[
|g_{i,\widetilde G}(a)| \le |g_{i,\widetilde G}(0)| + \|g'_{i,\widetilde G}\|_\infty |a| \le C_1(1+|a|)
\]
uniformly in \(i\), \(\widetilde G\), and \(a\). Therefore, on \(\{|A_i|>T\}\), $|g_{i,\widetilde G}(A_i)|\le C_1(1+|A_i|)$, and for $T\geq 1$,
\[
\begin{aligned}
\frac1n\sum_{i=1}^n
\mathbb E_{\boldsymbol{\theta},G}\!\left[(1+|A_i|)\mathds 1(|A_i|>T)\right]
&\le\left(\frac1{T^2}+\frac1T\right) \frac1n\sum_{i=1}^n \mathbb E_{\boldsymbol{\theta},G}\!\left[ A_i^2\mathds 1(|A_i|>T)\right],
\end{aligned}
\]
combining with Lemma~\ref{lemm:consequence_assm},
\begin{equation}
\label{eq:ulln_tail_bound}
\lim_{T\to\infty} \limsup_{n\to\infty}
\frac1n\sum_{i=1}^n \mathbb E_{\boldsymbol{\theta},G}\!\left[(1+|A_i|)\mathds 1(|A_i|>T)\right]=0.
\end{equation}
We can now combine the compact-region continuity bound and the tail bound with a finite \(W_1\)-net of \(\mathcal P([-M,M])\). Let $\widetilde G_1,\ldots,\widetilde G_J$ be a finite $W_1$-\!$\delta$-net of
$\mathcal{P}([-M,M])$, i.e. for any \(\widetilde G\in\mathcal{P}([-M,M])\) we can find a $j\leq J$ such that $W_1(\widetilde G,\widetilde G_j)\le\delta,$ and we have
\[
|Z_n(\widetilde G)| \le |Z_n(\widetilde G_j)| + |Z_n(\widetilde G)-Z_n(\widetilde G_j)|.
\]
By taking supremum over all possible $\widetilde G$,
\[
\sup_{\widetilde G\in\mathcal P([-M,M])}|Z_n(\widetilde G)| \le
\max_{j\le J}|Z_n(\widetilde G_j)| + \sup_{\widetilde G\in\mathcal P([-M,M])} |Z_n(\widetilde G)-Z_n(\widetilde G_j)|.
\]
For the second term,
\[
\begin{aligned}
&\left|Z_n(\widetilde G)-Z_n(\widetilde G_j)\right|\\
&\le\frac1n\sum_{i=1}^n\left|g_{i,\widetilde G}(A_i)-g_{i,\widetilde G_j}(A_i)\right|
+\frac1n\sum_{i=1}^n\EE[\boldsymbol{\theta},G]{\left|g_{i,\widetilde G}(A_i)-g_{i,\widetilde G_j}(A_i)\right|}\\
&\le
\frac1n\sum_{i=1}^n
\left|g_{i,\widetilde G}(A_i)-g_{i,\widetilde G_j}(A_i)\right|\mathds 1(|A_i|\le T)+\frac1n\sum_{i=1}^n
\left|g_{i,\widetilde G}(A_i)-g_{i,\widetilde G_j}(A_i)\right|\mathds 1(|A_i|> T)\\
&\qquad+
\frac1n\sum_{i=1}^n\mathbb E_{\boldsymbol{\theta},G}\!\left[\left|g_{i,\widetilde G}(A_i)-g_{i,\widetilde G_j}(A_i)\right|\mathds 1(|A_i|\le T)\right] +
\frac1n\sum_{i=1}^n\mathbb E_{\boldsymbol{\theta},G}\!\left[\left|g_{i,\widetilde G}(A_i)-g_{i,\widetilde G_j}(A_i)\right|\mathds 1(|A_i|> T)\right]\\
&\le 
2\frac{C_T}{d_T}\delta+
\frac{2C_1}{n}\sum_{i=1}^n(1+|A_i|)\mathds 1(|A_i|> T)+
\frac{2C_1}{n}\sum_{i=1}^n\mathbb E_{\boldsymbol{\theta},G}\!\left[(1+|A_i|)\mathds 1(|A_i|> T)\right].
\end{aligned}
\]
For $\epsilon>0$, choose $T$ large enough such that the last deterministic term is at most $\epsilon/4$; then choose
$\delta>0$ such that $2(C_T/d_T)\delta<\epsilon/4$. For these fixed
$T,\delta$, and all sufficiently large $n$,
$$
\begin{aligned}
\PP[\boldsymbol{\theta},G]{\sup_{\widetilde G\in\mathcal P([-M,M])}|Z_n(\widetilde G)|>\epsilon}
&\le\PP[\boldsymbol{\theta},G]{\max_{j\le J}|Z_n(\widetilde G_j)|>\epsilon/4}\\
&\quad+\PP[\boldsymbol{\theta},G]{\frac{2C_1}{n}\sum_{i=1}^n
(1+|A_i|)\mathds 1(|A_i|> T)>\epsilon/4}.
\end{aligned}
$$
For the first probability term, since the \(W_1\)-net is finite, \(J<\infty\), and \(Z_n(\widetilde G_j)\xrightarrow{\mathbb P_{\boldsymbol{\theta},G}}0\) for each fixed \(j\) by \eqref{eq:pointwise_convergence}, a finite union bound gives
\[
\PP[\boldsymbol{\theta},G]{\max_{j\le J}|Z_n(\widetilde G_j)|>\varepsilon/4}
\le\sum_{j=1}^J\PP[\boldsymbol{\theta},G]{|Z_n(\widetilde G_j)|>\varepsilon/4}
\longrightarrow 0.
\]
For the second probability term, Markov's inequality gives
\[
\PP[\boldsymbol{\theta},G]{\frac{2C_1}{n}\sum_{i=1}^n(1+|A_i|)\mathds 1(|A_i|> T)
>\frac{\varepsilon}{4}}
\le\frac{8C_1}{\varepsilon}\frac1n\sum_{i=1}^n\EE[\boldsymbol{\theta},G]{(1+|A_i|)\mathds 1(|A_i|> T)},
\]
and then combine it with \eqref{eq:ulln_tail_bound}, we have
\[
\lim_{T\to\infty}\limsup_{n\to\infty}
\mathbb P_{\boldsymbol{\theta},G}\!\left(\frac{2C_1}{n}\sum_{i=1}^n(1+|A_i|)\mathds 1(|A_i|> T)
>\frac{\varepsilon}{4}\right)=0.
\]
Combining the two bounds yields
\begin{equation}
\label{eq:ulln}
    \sup_{\widetilde G\in\mathcal P([-M,M])}\left|\frac1n\sum_{i=1}^n
\left(g_{i,\widetilde G}(A_i)-\EE[\boldsymbol{\theta},G]{g_{i,\widetilde G}(A_i)}\right)\right|
\xrightarrow{\mathbb{P}_{\boldsymbol{\theta},G}}0.
\end{equation}
With $\mu_i,\nu_i$ from Lemma~\ref{lemm:consequence_assm}, and $\|g'_{i,\widetilde G}\|_\infty\le 2M/\underline\lambda$, by the Kantorovich–Rubinstein duality for $W_1$, we have $
\left|\mathbb E_{\mu_i}g_{i,\widetilde G}-\mathbb E_{\nu_i}g_{i,\widetilde G}\right|
\le(2M/\underline\lambda)\cdot W_1(\mu_i,\nu_i),$ and consequently
\begin{equation}
\label{eq:error_actual_to_gaussian}
    \sup_{\widetilde G\in\mathcal P([-M,M])}
\left|\frac1n\sum_{i=1}^n
\left(\mathbb E_{\mu_i}g_{i,\widetilde G}-\mathbb E_{\nu_i}g_{i,\widetilde G}\right)
\right|
\le \frac{2M}{\underline\lambda} \frac1n\sum_{i=1}^n W_1(\mu_i,\nu_i) =\frac{2M}{\underline\lambda}V_n=o(1).
\end{equation}
Also notice that $\mathbb{E}_{\nu_i}g_{i,\widetilde G}
=-\operatorname{KL}\bigl(f_G(\cdot;\Sigma_{i,22})\|f_{\widetilde G}(\cdot;\Sigma_{i,22})\bigr)$.
Defining $$K_n(\widetilde G)=\frac{1}{n}\sum_{i=1}^n\operatorname{KL}
\bigl(f_G(\cdot;\Sigma_{i,22})\|f_{\widetilde G}(\cdot;\Sigma_{i,22})\bigr)=-\frac{1}{n}\sum_{i=1}^n\mathbb{E}_{\nu_i}g_{i,\widetilde G},$$ we have
$$
\begin{aligned}
L_n^{\mathrm{true}}(\widetilde G)-L_n^{\mathrm{true}}(G)+K_n(\widetilde G)
&=\left[L_n^{\mathrm{true}}(\widetilde G)-L_n^{\mathrm{true}}(G)-\frac1n\sum_{i=1}^n\mathbb E_{\mu_i}g_{i,\widetilde G}\right]\\
&\qquad+ \frac1n\sum_{i=1}^n \left(\mathbb E_{\mu_i}g_{i,\widetilde G}-\mathbb E_{\nu_i}g_{i,\widetilde G}\right),
\end{aligned}
$$
then combining \eqref{eq:ulln} and \eqref{eq:error_actual_to_gaussian}, we obtain 
\[
\sup_{\widetilde G\in\mathcal P([-M,M])}
\left|L_n^{\mathrm{true}}(\widetilde G)-L_n^{\mathrm{true}}(G)+K_n(\widetilde G)\right|
\xrightarrow{\mathbb P_{\boldsymbol{\theta},G}}0.
\]
Since $\widehat{G}\in \mathcal{P}([-M,M])$, and $L_n^{\mathrm{true}}(\widehat G)-
L_n^{\mathrm{true}}(G)\ge-2R_n=-o_{\mathbb P_{\boldsymbol{\theta},G}}(1),$ it follows that 
\[
\begin{aligned}
0\le K_n(\widehat G) &\le \sup_{\widetilde G\in\mathcal P([-M,M])}
\left|L_n^{\mathrm{true}}(\widetilde G)-L_n^{\mathrm{true}}(G)+K_n(\widetilde G)\right|+2R_n=o_{\mathbb P_{\boldsymbol{\theta},G}}(1).
\end{aligned}
\]
Therefore, since \(\mathfrak H^2(p,q)\le \operatorname{KL}(p\|q)\),
\[
\frac1n\sum_{i=1}^n \mathfrak H^2
\left(f_{\widehat G}(\cdot;\Sigma_{i,22}),f_G(\cdot;\Sigma_{i,22})\right)
\le K_n(\widehat G)=o_{\mathbb P_{\boldsymbol{\theta},G}}(1).
\]
It remains to translate the mixture Hellinger consistency into consistency of the mixing distribution. If $\varepsilon>2M$, then $W_1(\widehat G,G)<\varepsilon$ automatically. Otherwise, for $0<\varepsilon\le2M$, define
$$
c_\varepsilon :=
\inf_{\substack{G_1,G_2\in\mathcal P([-M,M]),\,
W_1(G_1,G_2)\ge\varepsilon,\\
v\in[\underline\lambda,\overline\lambda]}}
\mathfrak H^2\left(f_{G_1}(\cdot;v),f_{G_2}(\cdot;v)\right).
$$
The index set $\{(G_1,G_2,v): G_1,G_2\in\mathcal P([-M,M]),W_1(G_1,G_2)\ge\varepsilon,\,
v\in[\underline\lambda,\overline\lambda]\}$ is compact, and the map
\[
(G_1,G_2,v) \mapsto \mathfrak H^2 \left(f_{G_1}(\cdot;v),f_{G_2}(\cdot;v)\right)
\]
is continuous. Suppose for contradiction that $c_\varepsilon=0$. Then the infimum is attained at some $(G_1^\star,G_2^\star,v^\star)$ such that $\mathfrak H^2 \left(f_{G_1^\star}(\cdot;v^\star),f_{G_2^\star}(\cdot;v^\star)\right) =0$, therefore $f_{G_1^\star}(\cdot;v^\star)=f_{G_2^\star}(\cdot;v^\star)$.
Taking characteristic functions gives 
\[
\left(\int e^{itb}\,\dd G_1^\star(b)\right)e^{-v^\star t^2/2}
= \left(\int e^{itb}\,\dd G_2^\star(b)\right)e^{-v^\star t^2/2},
\]
and hence $G_1^\star=G_2^\star$, contradicting $W_1(G_1^\star ,G_2^\star)\ge\varepsilon$. Therefore, $c_\varepsilon>0$ and it is independent of $n$. Since $\wh G$ and $G$ belong to $\mathcal P([-M,M])$, it follows that
\[
c_\varepsilon \mathds 1(W_1(\widehat G,G)\ge\varepsilon)
\le \frac1n\sum_{i=1}^n \mathfrak H^2
\left( f_{\widehat G}(\cdot;\Sigma_{i,22}), f_G(\cdot;\Sigma_{i,22})\right)
= o_{\mathbb P_{\boldsymbol{\theta},G}}(1),
\]
which implies $W_1(\widehat G,G)\xrightarrow{\mathbb P_{\boldsymbol{\theta},G}}0.$ 
\end{proof}
The following lemma gives a uniform Lipschitz bound for the normal working-model conditional CDF with respect to the observation pair.
\begin{lemm}
\label{lemm:obs_perturbation}
For $\widetilde G\in\mathcal P([-M,M])$ and an admissible covariance matrix $S$ defined below, we define the following
$$
S=\begin{pmatrix}
v_Y & c\\
c & v_A
\end{pmatrix},
$$
$$
\widetilde G_{a,S}(\dd b)= \frac{\varphi(a-b;v_A)\widetilde G(\dd b)}{f_{\widetilde G}(a;v_A)},\quad
F_{\widetilde G}(y\mid a;S)=\int
\Phi\!\left(\frac{y-\frac{c}{v_A}(a-b)}{\sqrt{v_Y-c^2/v_A}}\right)
\widetilde G_{a,S}(\dd b),
$$
where $\underline\lambda I_2 \preceq S \preceq \bar{\lambda}I_2$. Then there is $L<\infty$, depending only on $M,\underline\lambda,
\overline\lambda$, such that uniformly in $\widetilde G,S,y,a,y',a'$,
    \begin{equation}
|F_{\widetilde G}(y\mid a;S)-F_{\widetilde G}(y'\mid a';S)|
\le L\sqrt{(y-y')^2+(a-a')^2}.
\label{eq:gap_f_func}
\end{equation}
\end{lemm}
\begin{proof}
Conditional on \(b\),  normal conditioning gives
\[
Y^G\mid(A^G=a,b)\sim
\mathrm{N}\!\left(\frac{c}{v_A}(a-b),v_Y-\frac{c^2}{v_A}\right).
\]
Under the normal working model, Bayes' rule gives the conditional distribution of \(b\) given \(A^G=a\) is \(\widetilde G_{a,S}\). Hence, integrating the preceding normal CDF with respect to \(\widetilde G_{a,S}\) gives the stated expression for \(F_{\widetilde G}(y\mid a;S)\). To establish the uniform Lipschitz bound, note that
$$
v_Y-\frac{c^2}{v_A}=\left(1,-\frac c{v_A}\right)
S\left(1,-\frac c{v_A}\right)^\top
\ge\underline\lambda\left(1+\frac{c^2}{v_A^2}\right)\ge\underline\lambda,\qquad 
\left|\frac c{v_A}\right|\le\frac{\bar\lambda}{\underline\lambda}.
$$
For brevity, set
\[
q(y,a,b)=\Phi\!\left(\frac{y-\frac{c}{v_A}(a-b)}{\sqrt{v_Y-c^2/v_A}}\right).
\]
Differentiating under the posterior integral gives
\[
\partial_y F_{\widetilde G}(y\mid a;S)
=\int \varphi\left(y-\frac{c}{v_A}(a-b);v_Y-\frac{c^2}{v_A}\right)\widetilde G_{a,S}(\dd b),
\]
and
\[
\begin{aligned}
\partial_a F_{\widetilde G}(y\mid a;S)
&=-\frac{c}{v_A}\int \varphi\left(y-\frac{c}{v_A}(a-b);v_Y-\frac{c^2}{v_A}\right)\widetilde G_{a,S}(\dd b)\\
&\quad+\frac1{v_A}\operatorname{Cov}_{\widetilde G_{a,S}}\left(q(y,a,b),b\right).
\end{aligned}
\]
Since $\varphi(x;v)\le 1/\sqrt{2\pi v}$ and $v_Y-c^2/v_A\ge \underline\lambda,$ we have $\left|\partial_y F_{\widetilde G}(y\mid a;S)\right|\le1/\sqrt{2\pi\underline\lambda}.$ For the second term, since $0\le q(y,a,b)\le1$ and $|b|\le M$, it holds that $\left|\operatorname{Cov}_{\widetilde G_{a,S}}(q(y,a,b),b)\right|\le 2M.$ Therefore,
\[
\left|\partial_aF_{\widetilde G}(y\mid a;S)\right|
\le\frac{\overline\lambda}{\sqrt{2\pi}\,\underline\lambda^{3/2}}
+\frac{2M}{\underline\lambda}.
\]
Thus the Euclidean norm of the gradient of $(y,a)\mapsto F_{\widetilde G}(y\mid a;S)$ is bounded uniformly over all admissible \(S\) and \(\widetilde G\in\mathcal P([-M,M])\). Integrating the gradient along the line segment joining \((y,a)\) and \((y',a')\) yields
\[
\left|F_{\widetilde G}(y\mid a;S)-F_{\widetilde G}(y'\mid a';S)\right|
\le L\sqrt{(y-y')^2+(a-a')^2},
\]
for some \(L<\infty\) depending only on \(M,\underline\lambda,\overline\lambda\).
\end{proof}
We next complement the preceding Lipschitz bound with stability with respect to the mixing distribution and covariance matrix.
\begin{lemm}
\label{lemm:param_perturbation}
For every $K>0$ there is a finite $C_K$ such that, for admissible
$S,S'$ and $\widetilde G,\widetilde G'\in\mathcal P([-M,M])$,
\begin{align}
\sup_{|a|\le K,y\in\mathbb R}|F_{\widetilde G}(y\mid a;S)-F_{\widetilde G'}(y\mid a;S)|
&\le C_KW_1(\widetilde G,\widetilde G'),\label{eq:priorLip}\\
\sup_{\widetilde G,|a|\le K,y\in\mathbb R}
|F_{\widetilde G}(y\mid a;S)-F_{\widetilde G}(y\mid a;S')|
&\le C_K\|S-S'\|.\label{eq:matrixLip}
\end{align}
\end{lemm}
\begin{proof}
Constants denoted by \(C_K\) below depend only on \(K,M,\underline\lambda,\overline\lambda\). For fixed $y,a,S$ ($S$ defined in Lemma~\ref{lemm:obs_perturbation}), define
\[
N_{\widetilde G}(y,a;S):=F_{\widetilde G}(y\mid a;S)f_{\widetilde G}(a;v_A),
\qquad D_{\widetilde G}(a;S):=f_{\widetilde G}(a;v_A).
\]
Equivalently,
\[
N_{\widetilde G}(y,a;S)=\int\Phi\!\left(\frac{y-\frac c{v_A}(a-b)}{\sqrt{v_Y-c^2/v_A}}\right)
\varphi(a-b;v_A)\,\widetilde G(\dd b),
\quad D_{\widetilde G}(a;S)=\int\varphi(a-b;v_A)\,\widetilde G(\dd b).\]
For \(|a|\le K\), define
\[
d_K:=\inf_{\substack{|a|\le K,\ |b|\le M,\\
v\in[\underline\lambda,\overline\lambda]
}} \varphi(a-b;v)>0.
\]
Then, uniformly in admissible \(S\), it holds that $D_{\widetilde G}(a;S)\ge d_K$ and $D_{\widetilde G'}(a;S)\ge d_K$. We first have on \(|a|\le K\) and \(|b|\le M\),
$$
\left|\partial_b \varphi(a-b;v_A)\right| = \left|\frac{a-b}{v_A}\varphi(a-b;v_A)\right|
\le\frac{K+M}{\underline\lambda}(2\pi\underline\lambda)^{-1/2}=:C_{K,1},
$$
and
$$
\left|\partial_b q(y,a,b)\right|=
\left|\frac{c}{v_A}\right|
\varphi\left(y-\frac{c}{v_A}(a-b);v_Y-\frac{c^2}{v_A}\right)\leq 
\frac{\overline\lambda}{\sqrt{2\pi}\,\underline\lambda^{3/2}}=:C_2,
$$
uniformly in \(y\) and admissible \(S\). Since
$$
\begin{aligned}
\left|\partial_b\{q(y,a,b)\varphi(a-b;v_A)\}\right|
&\le|\partial_b q(y,a,b)|\,\varphi(a-b;v_A)+|q(y,a,b)|\,|\partial_b \varphi(a-b;v_A)|\\
&\le C_2(2\pi\underline\lambda)^{-1/2}+C_{K,1}\\
&=:C_{K,2},
\end{aligned}
$$
then both
\[
b\mapsto\varphi(a-b;v_A),\quad b\mapsto q(y,a,b)\varphi(a-b;v_A)
\]
are uniformly Lipschitz on \(|a|\le K\), \(|b|\le M\), over admissible \(S\) and all \(y\in\mathbb R\). Therefore, by Kantorovich–Rubinstein duality, there exists a constant $C_K>0$ such that
\[
|D_{\widetilde G}(a;S)-D_{\widetilde G'}(a;S)| +
|N_{\widetilde G}(y,a;S)-N_{\widetilde G'}(y,a;S)|
\le C_K W_1(\widetilde G,\widetilde G').
\]
Since $0\le N_{\widetilde G'}(y,a;S)/D_{\widetilde G'}(a;S)=F_{\widetilde G'}(y\mid a;S)\le1$,
\[
\begin{aligned}
&
\left|F_{\widetilde G}(y\mid a;S)-F_{\widetilde G'}(y\mid a;S)\right|=
\left|\frac{N_{\widetilde G}(y,a;S)}{D_{\widetilde G}(a;S)}-\frac{N_{\widetilde G'}(y,a;S)}{D_{\widetilde G'}(a;S)}\right|\\
&\quad= \left|
\frac{N_{\widetilde G}(y,a;S)}{D_{\widetilde G}(a;S)}
-
\frac{N_{\widetilde G'}(y,a;S)}{D_{\widetilde G}(a;S)}
+
\frac{N_{\widetilde G'}(y,a;S)}{D_{\widetilde G}(a;S)}
-
\frac{N_{\widetilde G'}(y,a;S)}{D_{\widetilde G'}(a;S)}\right| \\
&\quad\le
\frac{\left|N_{\widetilde G}(y,a;S)-N_{\widetilde G'}(y,a;S)\right|}{D_{\widetilde G}(a;S)} +
\frac{N_{\widetilde G'}(y,a;S)}{D_{\widetilde G'}(a;S)}
\frac{\left|D_{\widetilde G}(a;S)-D_{\widetilde G'}(a;S)\right|}{D_{\widetilde G}(a;S)}\\
&\le
\frac{\left|N_{\widetilde G}(y,a;S)-N_{\widetilde G'}(y,a;S)\right|+\left|D_{\widetilde G}(a;S)-D_{\widetilde G'}(a;S)\right|}{d_K}
\\
&\le C_K W_1(\widetilde G,\widetilde G'),
\end{aligned}
\]
which proves \eqref{eq:priorLip}. For \eqref{eq:matrixLip}, let $S_t=(1-t)S+tS',t\in[0,1],$ and write \[
S_t=
\begin{pmatrix}
v_{Y,t} & c_t\\
c_t & v_{A,t}
\end{pmatrix},
\]
and the class of admissible covariance matrices is convex, $\underline\lambda I_2\preceq S_t\preceq\overline\lambda I_2$ for every \(t\in[0,1]\). In particular,
\[
v_{A,t}\ge \underline\lambda,
\qquad v_{Y,t}-\frac{c_t^2}{v_{A,t}}\ge \underline\lambda,
\qquad
|c_t|\le \overline\lambda,\qquad 
|\partial_t v_{Y,t}|,
\ |\partial_t v_{A,t}|,
\ |\partial_t c_t|
\le \|S'-S\|.
\]
Define
\[
z_t(y,a,b)= \frac{y-\frac{c_t}{v_{A,t}}(a-b)}{\sqrt{v_{Y,t}-c_t^2/v_{A,t}}},
\]
then on \(|a|\le K\) and \(|b|\le M\), we have
\[
\begin{aligned}
|\partial_t z_t(y,a,b)|
&\le \frac{|a-b|}{\sqrt{v_{Y,t}-c_t^2/v_{A,t}}}
\left|\partial_t\!\left(\frac{c_t}{v_{A,t}}\right)\right|+
|z_t(y,a,b)|\left|\partial_t\log\sqrt{v_{Y,t}-\frac{c_t^2}{v_{A,t}}}\right|\\
&\le C_K(1+|z_t(y,a,b)|)\|S-S'\|,
\end{aligned}
\]
and $\left|\partial_t\Phi(z_t(y,a,b))\right|=\varphi(z_t(y,a,b);1)
|\partial_t z_t(y,a,b)|\le C_K\|S-S'\|,$ since $\sup_z\varphi(z;1)<\infty$,  $\sup_z|z|\varphi(z;1)<\infty.$
For the normal marginal factor,
\[
\partial_t \varphi(a-b;v_{A,t}) = \varphi(a-b;v_{A,t})
\left[ -\frac{1}{2v_{A,t}} + \frac{(a-b)^2}{2v_{A,t}^2} \right]
\partial_t v_{A,t},
\]
since \(v_{A,t}\ge\underline\lambda\), then on \(|a|\le K\) and \(|b|\le M\), it holds that $\left|\partial_t\varphi(a-b;v_{A,t})\right|\le C_K\|S-S'\|,$ and then
\[
\begin{aligned}
\left|
\partial_t N_{\widetilde G}(y,a;S_t)
\right|
&\le\int\left|\partial_t\Phi(z_t(y,a,b))\right|\varphi(a-b;v_{A,t})\,\widetilde G(\dd b) \\
&\qquad+\int\Phi(z_t(y,a,b))\left|\partial_t\varphi(a-b;v_{A,t})\right|\,\widetilde G(\dd b)\\
&\le C_K\|S-S'\|.
\end{aligned}
\]
Similarly, $\left|\partial_tD_{\widetilde G}(a;S_t)\right|\le C_K\|S-S'\|.$ Then, integrating over \(t\in[0,1]\) gives
\[
|N_{\widetilde G}(y,a;S)-N_{\widetilde G}(y,a;S')|
+
|D_{\widetilde G}(a;S)-D_{\widetilde G}(a;S')|\le
C_K\|S-S'\|.
\]
Finally, since for \(|a|\le K\), $D_{\widetilde G}(a;S),
\ D_{\widetilde G}(a;S')\geq d_K > 0$, the similar quotient bound used above yields
\[|F_{\widetilde G}(y\mid a;S)-F_{\widetilde G}(y\mid a;S')|\le C_K\|S-S'\|,\]
uniformly over \(\widetilde G\in\mathcal P([-M,M])\), \(|a|\le K\), and \(y\in\mathbb R\).
\end{proof}
The final lemma combines the preceding bounds to compare the estimated conditional CDF with its oracle and matched-normal counterparts.
\begin{lemm}
\label{lemm:gauss_to_plugin}
Let $V_i=F_{G}(Y_i\mid A_i;\Sigma_i)$, $U_i^G=F_{G}(Y_i^G\mid A_i^G;\Sigma_i)$, and
$\widehat U_i=F_{\widehat G}(Y_i\mid A_i;\widehat\Sigma_i)$. On the actual data space, $n^{-1}\sum_{i=1}^n\EE[\boldsymbol{\theta},G]{|\widehat U_i-V_i|}\to0$; under separate per-task normal couplings,
$n^{-1}\sum_{i=1}^n\EE{|V_i-U_i^G|}\to0$.
\end{lemm}
\begin{proof}
For each $i$, conditionally on $b_i$, couple the actual pair $(Y_i,A_i)$ with its matched normal pair $(Y_i^G,A_i^G)$ by an optimal $W_2$ coupling. By Lemma~\ref{lemm:obs_perturbation},
\[
|V_i-U_i^G|\le L\|(Y_i,A_i)-(Y_i^G,A_i^G)\|.
\]
Under the optimal \(W_2\) coupling,
\[
\EE{\|(Y_i,A_i)-(Y_i^G,A_i^G)\|^2\mid b_i}=
W_2^2\!\left( \mathcal L_{\boldsymbol{\theta},G}((Y_i,A_i)\mid b_i), \mathrm{N}_2((0,b_i)^\top,\Sigma_i) \right),
\]
and then averaging first over \(b_i\sim G\) and then over \(i\), and applying Cauchy–Schwarz,
\[
\begin{aligned}
\frac1n\sum_{i=1}^n \mathbb E|V_i-U_i^G|
&\le\frac Ln \sum_{i=1}^n \EE{\left\|(Y_i,A_i)-(Y_i^G,A_i^G)\right\|}\\
&\le
L\left\{\frac1n\sum_{i=1}^n\EE{\left\|(Y_i,A_i)-(Y_i^G,A_i^G)\right\|^2}\right\}^{1/2}\\
&=L\left\{\frac1n\sum_{i=1}^n\EE[b_i\sim G]{W_2^2\!\left(\mathcal L((Y_i,A_i)\mid b_i),\mathrm  N_2((0,b_i)^\top,\Sigma_i)\right)}
\right\}^{1/2}\\
&=L\mathcal W_n\longrightarrow0.
\end{aligned}
\]
For $K>0$, it holds that
\[
\begin{aligned}
D_n &:=
\frac1n\sum_{i=1}^n
\left|F_{\widehat G}(Y_i\mid A_i;\widehat\Sigma_i)-F_G(Y_i\mid A_i;\Sigma_i)\right|\\
&\le\frac1n\sum_{i=1}^n\left|F_{\widehat G}(Y_i\mid A_i;\widehat\Sigma_i)-F_G(Y_i\mid A_i;\widehat\Sigma_i)\right|\mathds 1(|A_i|\le K) \\
&\qquad +\frac1n\sum_{i=1}^n\left|F_G(Y_i\mid A_i;\widehat\Sigma_i)-F_G(Y_i\mid A_i;\Sigma_i)\right|
\mathds 1(|A_i|\le K)
\\
&\qquad+\frac1n\sum_{i=1}^n\left|F_{\widehat G}(Y_i\mid A_i;\widehat\Sigma_i)
-F_G(Y_i\mid A_i;\Sigma_i)\right|
\mathds 1(|A_i|>K)
\\
&\le C_KW_1(\widehat G,G)
+\frac{C_K}{n}\sum_{i=1}^n\|\widehat\Sigma_i-\Sigma_i\|+ \frac1n\sum_{i=1}^n \mathds 1(|A_i|>K),
\end{aligned}
\]
where on $\{|A_i|\leq K\}$ we apply Lemma~\ref{lemm:param_perturbation}, on $\{|A_i|> K\}$ we use $|F_{\widehat G}-F_G|\le1$. By Lemma~\ref{lemm:consequence_assm}, \ref{lemm:consistency}, Assumption~\ref{assm:robust_second_moment}, we have that
$$
W_1(\widehat G,G)\xrightarrow{\mathbb{P}_{\boldsymbol{\theta},G}}0,\quad \frac{1}{n}\sum_{i=1}^n
\|\widehat\Sigma_i-\Sigma_i\|
\xrightarrow{\mathbb P_{\boldsymbol{\theta},G}}0, \quad 
\lim_{K\to\infty} \limsup_n \frac1n\sum_{i=1}^n \PP[\boldsymbol{\theta},G]{|A_i|>K}=0.
$$
Hence, for every \(\eta>0\),
\[
\begin{aligned}
\lim_{K\to\infty}\limsup_{n\to\infty}
\PP[\boldsymbol{\theta},G]{\frac1n\sum_{i=1}^n\mathds 1(|A_i|>K)>\eta}
&\le
\lim_{K\to\infty}\limsup_{n\to\infty}
\frac1{\eta n}
\sum_{i=1}^n
\PP[\boldsymbol{\theta},G]{|A_i|>K}=0.
\end{aligned}
\]
Therefore, choosing \(K\) sufficiently large and then letting \(n\to\infty\), we obtain $D_n\xrightarrow{\mathbb{P}_{\boldsymbol{\theta},G}}0$ and $\EE[\boldsymbol{\theta},G]{D_n}\to0$, since $0\le D_n\le1$. 
\end{proof}
We now combine the preceding results to provide coverage guarantee.
\begin{proof}[Proof of Theorem~\ref{thm:robust_w_exchangeability}] Take $C_\alpha=[\alpha/2,1-\alpha/2],$ and by definition
\begin{equation}
\label{eq:invariant_transformation}
\theta_i\in\mathcal I_{\widehat G,i}^{\rb}(1-\alpha)
\iff \frac\alpha2 \le F_{\widehat G} \left( \widehat\theta_i^{\db}-\theta_i \mid \widehat b_i;\widehat\Sigma_i\right) \le 1-\frac\alpha2
\iff \widehat U_i\in C_\alpha.
\end{equation}
It remains to show 
$$
\frac1n\sum_{i=1}^n\left|\PP[\boldsymbol{\theta},G]{\widehat U_i\in C_\alpha}-(1-\alpha)\right|\to0.
$$
First, since $F_G(\cdot\mid A_i^G;\Sigma_i)$ is the conditional CDF of $Y_i^G\mid A_i^G$ under the matched normal model, the probability integral transform gives
$$
U_i^G\mid A_i^G\sim\mathrm{Unif}(0,1),\qquad U_i^G\sim\mathrm{Unif}(0,1).
$$
For \(\varepsilon>0\), define
\[
C_\alpha^{+\varepsilon}=\left[\frac{\alpha}{2}-\varepsilon,1-\frac{\alpha}{2}+\varepsilon\right]\cap[0,1],
\quad
C_\alpha^{-\varepsilon} =\left[\frac{\alpha}{2}+\varepsilon,1-\frac{\alpha}{2}-\varepsilon\right],
\]
with \(C_\alpha^{-\varepsilon}=\varnothing\) if the left endpoint exceeds the right endpoint.
On the actual data space,
\[
\begin{aligned}
\{\widehat U_i\in C_\alpha\}
&\subseteq \left\{V_i\in C_\alpha^{+\varepsilon/2}\right\}
\cup \left\{|\widehat U_i-V_i|>\frac{\varepsilon}{2}\right\},\\
\left\{V_i\in C_\alpha^{-\varepsilon/2}\right\}
&\subseteq \{\widehat U_i\in C_\alpha\}
\cup \left\{|\widehat U_i-V_i|>\frac{\varepsilon}{2}\right\}.
\end{aligned}
\]
Since the marginal law of \(V_i\) is preserved under the chosen per-task coupling,
\[
\PP[\boldsymbol{\theta},G]{V_i\in C_\alpha^{+\varepsilon/2}} = \PP{V_i\in C_\alpha^{+\varepsilon/2}},
\]
and similarly for \(C_\alpha^{-\varepsilon/2}\). Under the per-task coupling,
\[
\begin{aligned}
\left\{V_i\in C_\alpha^{+\varepsilon/2}\right\}
&\subseteq \left\{U_i^G\in C_\alpha^{+\varepsilon}\right\}
\cup \left\{|V_i-U_i^G|>\frac{\varepsilon}{2}\right\}\\
\left\{U_i^G\in C_\alpha^{-\varepsilon}\right\}
&\subseteq \left\{V_i\in C_\alpha^{-\varepsilon/2}\right\}
\cup \left\{|V_i-U_i^G|>\frac{\varepsilon}{2}\right\}.
\end{aligned}
\]
Hence
\[
\begin{aligned}
\PP[\boldsymbol{\theta},G]{\widehat U_i\in C_\alpha}
&\le
\PP{U_i^G\in C_\alpha^{+\varepsilon}}
+\PP[\boldsymbol{\theta},G]{|\widehat U_i-V_i|>\frac{\varepsilon}{2}}
+\PP{|V_i-U_i^G|>\frac{\varepsilon}{2}},\\
\PP[\boldsymbol{\theta},G]{\widehat U_i\in C_\alpha}
&\ge
\PP{U_i^G\in C_\alpha^{-\varepsilon}}
-\PP[\boldsymbol{\theta},G]{|\widehat U_i-V_i|>\frac{\varepsilon}{2}}
-\PP{|V_i-U_i^G|>\frac{\varepsilon}{2}}.
\end{aligned}
\]
Since \(U_i^G\sim\operatorname{Unif}(0,1)\),
\[
\PP{U_i^G\in C_\alpha^{+\varepsilon}}\le 1-\alpha+2\varepsilon, \quad
\PP{U_i^G\in C_\alpha^{-\varepsilon}}\ge1-\alpha-2\varepsilon.
\]
Therefore,
\[
\begin{aligned}
\left|\PP[\boldsymbol{\theta},G]{\widehat U_i\in C_\alpha}-(1-\alpha)\right|
\le\;&
\PP[\boldsymbol{\theta},G]{|\widehat U_i-V_i|>\frac{\varepsilon}{2}}+\PP{|V_i-U_i^G|>\frac{\varepsilon}{2}}+2\varepsilon.
\end{aligned}
\]
Averaging over \(i\) and applying Markov's inequality gives
\[
\begin{aligned}
\frac1n\sum_{i=1}^n\left|\PP[\boldsymbol{\theta},G]{\widehat U_i\in C_\alpha}-(1-\alpha)\right|\le\frac{2}{n\varepsilon}\sum_{i=1}^n\left\{\mathbb E_{\boldsymbol{\theta},G}|\widehat U_i-V_i|+\mathbb E|V_i-U_i^G|\right\}+2\varepsilon.
\end{aligned}
\]
By Lemma~\ref{lemm:gauss_to_plugin},
$$
\frac1n\sum_{i=1}^n \mathbb E_{\boldsymbol{\theta},G}|\widehat U_i-V_i|\longrightarrow0,\quad
\frac1n\sum_{i=1}^n\mathbb E|V_i-U_i^G|\longrightarrow0.
$$
Thus, for every fixed \(\varepsilon>0\),
\[
\limsup_{n\to\infty}\frac1n\sum_{i=1}^n\left|\PP[\boldsymbol{\theta},G]{\widehat U_i\in C_\alpha}-(1-\alpha)\right|\le2\varepsilon.
\]
Letting \(\varepsilon\downarrow0\) and combining with \eqref{eq:invariant_transformation}, we conclude that
\[
\frac1n\sum_{i=1}^n\left|\PP[\boldsymbol{\theta},G]{\theta_i\in\mathcal I_{\widehat G,i}^{\mathrm{rb}}(1-\alpha)}-(1-\alpha)\right|\longrightarrow0.
\]
\end{proof}
\subsection{Proof of Theorem~\ref{thm:robust_joint}}
The argument follows the proof of Theorem~\ref{thm:robust_w_exchangeability}, with all expectations and convergence statements taken under the hierarchical sampling law $\mathbb P_{\boldsymbol H}$. 
Since each $b_i$ has marginal distribution $G$, and Assumption~\ref{assm:robust_normality} controls the conditional law of $(Y_i,A_i)$ given $b_i$, the proof of Lemma~\ref{lemm:consequence_assm} carries over under $\mathbb P_{\boldsymbol H}$ using the same second-moment conditions. The proof of Lemma~\ref{lemm:consistency} then applies under Assumptions~\ref{assm:robust_normality}--\ref{assm:robust_independence}, yielding $W_1(\widehat G,G)\xrightarrow{\mathbb P_{\boldsymbol H}}0$, which is used in the conditional CDF approximation below.
Lemma~\ref{lemm:obs_perturbation} and \ref{lemm:param_perturbation} can be applied directly. It remains to verify the conditional CDF approximation used in Lemma~\ref{lemm:gauss_to_plugin}. For each $i$, conditionally on $b_i$, couple the actual pair $(Y_i,A_i)$, whose conditional law marginalizes over $\theta_i \mid b_i$, with its matched normal pair $(Y_i^G,A_i^G)$ using an optimal $W_2$ coupling. By Lemma~\ref{lemm:obs_perturbation} and Cauchy-Schwarz,
\[
\begin{aligned}
\frac1n\sum_{i=1}^n
\EE{|V_i-U_i^G|}
&\leq \frac{L}{n}\sum_{i=1}^n
\EE{\left\|(Y_i,A_i)-(Y_i^G,A_i^G)\right\|}\\
&\leq
L\left\{
\frac1n\sum_{i=1}^n
\EE[b_i\sim G]{W_2^2(\mathcal L_{\boldsymbol H}
((Y_i,A_i)\mid b_i),
\mathrm N_2
((0,b_i)^{\top},\Sigma_i))}\right\}^{1/2}\\
&=L\mathcal W_n\longrightarrow0,
\end{aligned}
\]
Moreover, Lemma~\ref{lemm:consequence_assm}, \ref{lemm:consistency}, and \ref{lemm:param_perturbation}, together with Assumption~\ref{assm:robust_second_moment}, yield
\[
\frac1n\sum_{i=1}^n\EE[\boldsymbol H]{|\widehat U_i-V_i|}\longrightarrow0
\]
by the same argument as in Lemma~\ref{lemm:gauss_to_plugin}. Finally, the interval inversion identity \eqref{eq:invariant_transformation} holds pointwise, regardless of whether $\theta_i$ is fixed or random. Also, the probability integral transform under the matched normal working model gives $U_i^G\sim \mathrm{Unif}(0,1)$, then the same coverage-transfer argument as in the proof of Theorem~\ref{thm:robust_w_exchangeability} yields, for every fixed $\varepsilon>0$, 
\[
\frac1n\sum_{i=1}^n\left|\PP[\boldsymbol{H}]{\widehat U_i\in C_\alpha}-(1-\alpha)\right|
\le \frac{2}{n\varepsilon} \sum_{i=1}^n\left\{\mathbb E_{\boldsymbol{H}}|\widehat U_i-V_i|+\mathbb E|V_i-U_i^G|\right\}+ 2\varepsilon.
\]
The two averaged error terms converge to zero by the preceding arguments. Consequently,
\[
\limsup_{n\to\infty}\frac1n\sum_{i=1}^n\left|\PP[\boldsymbol{H}]{\widehat U_i\in C_\alpha}-(1-\alpha)\right|\le2\varepsilon.
\]
Letting \(\varepsilon\downarrow0\) and applying the interval inversion identity, we conclude that
\[
\frac1n\sum_{i=1}^n\left|\PP[\boldsymbol{H}]{\theta_i\in\mathcal I_{\widehat G,i}^{\mathrm{rb}}(1-\alpha)}-(1-\alpha)\right|\longrightarrow0.
\]
\subsection{
  Proof of Theorem~\ref{thm:robust_compound}
}
Let $G_n=n^{-1}\sum_{i=1}^n\delta_{b_i}$ denote the empirical distribution of the fixed biases. The proof largely follows the structure of Theorem~\ref{thm:robust_w_exchangeability}. Lemmas~\ref{lemm:obs_perturbation} and~\ref{lemm:param_perturbation} apply directly with \(\widetilde G=G_n\),and the arguments of Lemmas~\ref{lemm:consequence_assm}, \ref{lemm:consistency}, and \ref{lemm:gauss_to_plugin} are adapted to the compound setting as follows.
\begin{lemm}
\label{lemm:cp_consequence_assm}
Let $\mu_i=\mathcal L_{\boldsymbol{\theta},\boldsymbol{b}}(A_i)$ and
$\nu_i=\mathrm{N}(b_i,\Sigma_{22})$. Under the assumptions of Theorem~\ref{thm:robust_compound},
\[
\frac1n\sum_{i=1}^n W_1(\mu_i,\nu_i)\leq\mathcal W_n\to0,
\qquad
\lim_{T\to\infty}\limsup_{n\to\infty}
\frac1n\sum_{i=1}^n
\mathbb E_{\boldsymbol{\theta},\boldsymbol{b}}[A_i^2\mathbf1\{|A_i|>T\}]=0,
\]
and, on writing $\widehat\tau_i^2=\widehat\Sigma_{i,22}$, $\tau^2=\Sigma_{22}$,
\[
\frac1n\sum_{i=1}^n|\widehat\tau_i-\tau|
\xrightarrow{\mathbb P_{\boldsymbol{\theta},\boldsymbol{b}}}0.
\]
\end{lemm}

\begin{proof}
Let $d_i=W_2\left(\mathcal L_{\boldsymbol\theta,\boldsymbol b}((Y_i,A_i)),\mathrm N_2((0,b_i)^\top,\Sigma)\right)$. Projection gives $W_1(\mu_i,\nu_i)\le d_i,$ and hence, by Cauchy–Schwarz,
\[
\frac1n\sum_iW_1(\mu_i,\nu_i)
\le
\left(\frac1n\sum_i d_i^2\right)^{1/2}
=\mathcal W_n.
\]
For the tail bound, the proof of Lemma~\ref{lemm:consequence_assm} applies directly. In the present setting, $b_i$ is deterministic and $A_i^G=b_i+\tau Z_i$, where $Z_i\sim\mathrm N(0,1)$, $|b_i|\le M$, and $\tau^2\le\overline\lambda$. Hence the same uniform bound as in Lemma~\ref{lemm:consequence_assm} holds, and therefore
\[
\frac1n\sum_{i=1}^n \mathbb E_{\boldsymbol{\theta},\boldb}\!\left[
A_i^2\mathbf 1\{|A_i|>T\} \right]
\le 4\mathcal W_n^2+\frac{C}{T^2}.
\]
Letting \(n\to\infty\) and then \(T\to\infty\) proves the second claim. The final claim follows exactly as in Lemma~\ref{lemm:consequence_assm} with \(\tau_i\) replaced by the common \(\tau\) and \(\Sigma_i\) by \(\Sigma\).
\end{proof}

\begin{lemm}
\label{lemm:cp_consistency}
Under the assumptions of Theorem~\ref{thm:robust_compound}, we have $W_1(\widehat G,G_n)\xrightarrow{\mathbb P_{\boldsymbol{\theta},\boldsymbol{b}}}0.$
\end{lemm}
\begin{proof}
The proof follows that of Lemma~\ref{lemm:consistency}. For \(\widetilde G\in\mathcal P([-M,M])\), define
\[
L_n^{\mathrm{true}}(\widetilde G)
=\frac1n\sum_{i=1}^n\log f_{\widetilde G}(A_i;\Sigma_{22}),
\qquad
L_n^{\mathrm{est}}(\widetilde G)=\frac1n\sum_{i=1}^n\log f_{\widetilde G}(A_i;\widehat\Sigma_{i,22}).
\]    
We record only the changes specific to the compound setting. As in Lemma~\ref{lemm:consistency}, the variance-derivative bound together with Lemma~\ref{lemm:cp_consequence_assm} gives
\[
R_n:=
\sup_{\widetilde G\in\mathcal P([-M,M])}
\left|
L_n^{\mathrm{est}}(\widetilde G)-L_n^{\mathrm{true}}(\widetilde G)\right|=o_{\mathbb P_{\boldsymbol{\theta},\boldsymbol{b}}}(1).
\]
Since $\widehat G$ is a global maximizer of $L_n^{\mathrm{est}}$ over $\mathcal P([-M,M])$ and $G_n\in\mathcal P([-M,M])$, we have $L_n^{\mathrm{est}}(\widehat G)\ge L_n^{\mathrm{est}}(G_n)$. Consequently,
$$
L_n^{\mathrm{true}}(\widehat G)-L_n^{\mathrm{true}}(G_n)
\ge-2R_n=-o_{\mathbb P_{\boldsymbol{\theta},\boldsymbol{b}}}(1).
$$
Set $g_{n,\widetilde G}(a)=\log\{f_{\widetilde G}(a;\Sigma_{22})/f_{G_n}(a;\Sigma_{22})\}$. 
Since $G_n$ is deterministic and both $G_n$ and $\widetilde G$ belong to $\mathcal P([-M,M])$, the derivative bounds, truncation argument, and finite-$W_1$-net argument from Lemma~\ref{lemm:consistency} carry over. In particular, Assumption~\ref{assm:robust_independence} applies to the triangular sequence $g_{n,\widetilde G}$ for each fixed $\widetilde G$, yielding
\[
\sup_{\widetilde G}
\left|\frac1n\sum_{i=1}^n
\left(g_{n,\widetilde G}(A_i)-\EE[\boldsymbol{\theta},\boldsymbol{b}]{g_{n,\widetilde G}(A_i)}\right)
\right|
= o_{\mathbb P_{\boldsymbol{\theta},\boldsymbol{b}}}(1).
\]
Likewise, using the same Lipschitz bound as in Lemma~\ref{lemm:consistency} together with Lemma~\ref{lemm:cp_consequence_assm},
\[
\sup_{\widetilde G\in\mathcal P([-M,M])}
\left| \frac1n\sum_{i=1}^n \left(\mathbb E_{\mu_i}g_{n,\widetilde G}-\mathbb E_{\nu_i}g_{n,\widetilde G}\right) \right|
\le
\frac{2M}{\underline\lambda} \frac1n\sum_{i=1}^n W_1(\mu_i,\nu_i) = o(1).
\]
Since \(\nu_i=\mathrm{N}(b_i,\Sigma_{22})\) and $n^{-1}\sum_{i=1}^n \nu_i$ has density \(f_{G_n}(\cdot;\Sigma_{22})\),
\[\frac1n\sum_{i=1}^n\mathbb E_{\nu_i}g_{n,\widetilde G}
= \int f_{G_n}(a;\tau^2)
\log\frac{f_{\widetilde G}(a;\tau^2)}
     {f_{G_n}(a;\tau^2)}\,\dd a= -\mathrm{KL}\!\left(f_{G_n}(\cdot;\tau^2)\Vert f_{\widetilde G}(\cdot;\tau^2)\right).
\]
Together with $L_n^{\mathrm{true}}(\widehat G)-L_n^{\mathrm{true}}(G_n)\ge-2R_n$, this yields
$$
\begin{aligned}
0\le&\ \mathrm{KL}\!\left(f_{G_n}(\cdot;\tau^2)\Vert f_{\widehat G}(\cdot;\tau^2)\right)\\
\le&\ \sup_{\widetilde G\in\mathcal P([-M,M])}\left|L_n^{\mathrm{true}}(\widetilde G)-L_n^{\mathrm{true}}(G_n)
+\mathrm{KL}\!\left(f_{G_n}(\cdot;\tau^2)\Vert f_{\widetilde G}(\cdot;\tau^2)\right)\right|+2R_n=o_{\mathbb P_{\boldsymbol\theta,\boldsymbol b}}(1),
\end{aligned}
$$
and therefore
\[
\mathfrak H^2\!\left(f_{\widehat G}(\cdot;\tau^2),f_{G_n}(\cdot;\tau^2)\right)
\le\mathrm{KL}\!\left(f_{G_n}(\cdot;\tau^2)\Vert f_{\widehat G}(\cdot;\tau^2)\right)
=o_{\mathbb P_{\boldsymbol{\theta},\boldsymbol{b}}}(1).
\]
To account for the moving target $G_n$, define
\[
c_\varepsilon=\inf_{\substack{G,\widetilde G\in\mathcal P([-M,M])\\
W_1(G,\widetilde G)\ge\varepsilon,\;v\in[\underline\lambda,\overline\lambda]
}}
\mathfrak H^2\!\left(f_G(\cdot;v),f_{\widetilde G}(\cdot;v)\right).
\]
Compactness, continuity, and identifiability of normal convolution imply \(c_\varepsilon>0\) (similar to Lemma~\ref{lemm:consistency}). Hence
\[
c_\varepsilon\mathds1\!(W_1(\widehat G,G_n)\ge\varepsilon)
\le\mathfrak H^2\!\left(f_{\widehat G}(\cdot;\tau^2), f_{G_n}(\cdot;\tau^2) \right) = o_{\mathbb P_{\boldsymbol{\theta},\boldsymbol{b}}}(1),
\]
which proves
\[
W_1(\widehat G,G_n) \xrightarrow{\mathbb P_{\boldsymbol{\theta},\boldsymbol{b}}}0.
\]
\end{proof}
\begin{lemm}
\label{lemm:cp_gauss_to_plugin}
Define $V_i=F_{G_n}(Y_i\mid A_i;\Sigma)$, $U_i^{G_n}=F_{G_n}(Y_i^G\mid A_i^G;\Sigma)$, and $\widehat U_i=F_{\widehat G}(Y_i\mid A_i;\widehat\Sigma_i)$. Then, on the actual data space, $n^{-1}\sum_{i=1}^n\EE[\boldsymbol{\theta},\boldsymbol{b}]{|\widehat U_i-V_i|}\to0$; under separate per-task optimal normal couplings, $n^{-1}\sum_{i=1}^n\EE{|V_i-U_i^{G_n}|}\to0$.
\end{lemm}
\begin{proof}
The second statement follows exactly as in Lemma~\ref{lemm:gauss_to_plugin} by applying Lemma~\ref{lemm:obs_perturbation} with \(\widetilde G=G_n\) and \(S=\Sigma\):
\[
\frac1n\sum_{i=1}^n
\mathbb{E}|V_i-U_i^{G_n}|
\le \frac Ln\sum_{i=1}^n \EE{\|(Y_i,A_i)-(Y_i^G,A_i^G)\|}
\le L\mathcal W_n\to0.
\]
For the first statement, the same truncation argument as in Lemma~\ref{lemm:gauss_to_plugin} together with Lemma~\ref{lemm:param_perturbation} gives, for every \(K>0\),
\[
D_n:=\frac1n\sum_{i=1}^n|\widehat U_i-V_i|
\le C_KW_1(\widehat G,G_n)
+
\frac{C_K}{n}\sum_{i=1}^n\|\widehat\Sigma_i-\Sigma\|
+
\frac1n\sum_{i=1}^n1\{|A_i|>K\}.
\]
By Lemmas \ref{lemm:cp_consequence_assm} and \ref{lemm:cp_consistency}, together with Assumption~\ref{assm:robust_second_moment}, the right-hand side converges to zero in \(\mathbb P_{\boldsymbol{\theta},\boldsymbol{b}}\)-probability after first letting \(n\to\infty\) and then \(K\to\infty\). Hence \(D_n\to0\) in \(\mathbb P_{\boldsymbol{\theta},\boldsymbol{b}}\)-probability. Since \(0\le D_n\le1\), uniform integrability gives $\EE[\boldsymbol{\theta},\boldsymbol{b}]{D_n}\to0$.
\end{proof}
\begin{proof}[Proof of Theorem~\ref{thm:robust_compound}]
Fix \(\alpha\in(0,1)\) and let $C_\alpha=[\alpha/2,1-\alpha/2].$ As in the proof of Theorem~\ref{thm:robust_w_exchangeability},
\[
\theta_i\in\mathcal I^{\rm rb}_{\widehat G,i}(1-\alpha)
\quad\Longleftrightarrow\quad
\widehat U_i\in C_\alpha.
\]
It remains to establish the oracle calibration appropriate for the compound setting. Let
\[
J\sim {\rm Unif}\{1,\ldots,n\},\qquad
(Y_J^G,A_J^G)\mid(J=i) \sim \mathrm{N}_2\!\left(\begin{pmatrix}0\\ b_i\end{pmatrix},\Sigma\right).
\]
Since \(J\) is uniform, \(b_J\sim G_n\), the marginal law of \((Y_J^G,A_J^G)\) is the normal working model with mixing distribution \(G_n\) and covariance \(\Sigma\). Hence the probability integral transform gives
\[
U_J^{G_n}=F_{G_n}(Y_J^G\mid A_J^G;\Sigma) \sim {\rm Unif}(0,1).
\]
Equivalently, for every measurable \(C\subseteq[0,1]\),
\[
\frac1n\sum_{i=1}^n \PP{U_i^{G_n}\in C}=\operatorname{Leb}(C),
\]
where \(\operatorname{Leb}(C)\) denotes the Lebesgue measure of \(C\). Applying the same enlargement and shrinkage argument as in the proof of Theorem~\ref{thm:robust_w_exchangeability}, with $V_i=F_{G_n}(Y_i\mid A_i;\Sigma)$, for every sufficiently small \(\varepsilon>0\),
\[
\left|\frac1n\sum_{i=1}^n\PP[\boldsymbol{\theta},\boldsymbol{b}]{\widehat U_i\in C_\alpha}-(1-\alpha)\right|
\le2\varepsilon
+\frac{2}{n\varepsilon}\sum_{i=1}^n\left\{\mathbb E_{\boldsymbol{\theta},\boldsymbol{b}}|\widehat U_i-V_i|
+\mathbb E|V_i-U_i^{G_n}|\right\}.
\]
By Lemma~\ref{lemm:cp_gauss_to_plugin},
\[
\frac1n\sum_{i=1}^n\mathbb E_{\boldsymbol{\theta},\boldsymbol{b}}|\widehat U_i-V_i|\to0,
\qquad
\frac1n\sum_{i=1}^n\mathbb E|V_i-U_i^{G_n}|\to0.
\]
Therefore, letting \(n\to\infty\) and then \(\varepsilon\downarrow0\),
\[
\frac1n\sum_{i=1}^n \mathbb P_{\boldsymbol{\theta},\boldsymbol{b}}
\left[\theta_i\in\mathcal I^{\rm rb}_{\widehat G,i}(1-\alpha)\right]
\to 1-\alpha.
\]
\end{proof}

\section{Details on numerical studies (Section~\ref{sec:numerical})}\label{app:experiment}

This appendix collects the dataset construction, hyperparameter choices, and
supplementary tables for the numerical studies in
\S\ref{sec:numerical}.

\paragraph{Computational resources.} We run the downstream numerical studies for the PPI setting (LM Arena, Amazon Fine Food
Reviews, and Synthetic Amazon) on an HPC cluster with 16-core Intel Xeon CPUs and 32GB memory. For the PPI setting (LM Arena, Amazon Fine Food Reviews, and Synthetic Amazon), we used $K=200$ random splits for the two real-data experiments and $K=500$ Monte Carlo replications for the synthetic Amazon experiments. Each downstream interval experiment was completed within
10 minutes. The height GWAS analysis and realistic-LD null simulation were run on a MacBook Pro with a 15-core Apple M5 Pro CPU and 24 GB of memory, taking approximately 6 minutes and 4.2 hours, respectively. 

These timings exclude upstream predictor training and scoring, whose hardware requirements and execution instructions are documented in the released code repository\footnote{\url{https://github.com/wyling01/eb-rebiasing}.}.

\paragraph{Empirical Bayes confidence intervals.}
For the EBCI baselines~\citep{armstrong2022robust} used in Section~\ref{subsec:numerical_ppi}, we use the implementation in the \texttt{ebci} R package with its default fourth-moment adjustment. Given task-level estimates \(Y_i\) with standard errors \(\sigma_i\), the procedure estimates the across-task mean \(\delta\), residual second moment \(\mu_2\), and kurtosis $\kappa$, and forms the empirical Bayes estimator \(\hat\theta_i^{\mathrm{EB}}=\delta+w_i(Y_i-\delta)\), where \(w_i=\mu_2/(\mu_2+\sigma_i^2)\). Thus, when the primary parameters are tightly concentrated across tasks, the procedure shrinks more strongly toward the common mean, while the shrinkage diminishes as across-task heterogeneity increases. EBCIs are then constructed around the shrunken estimates using critical values that incorporate the estimated kurtosis to account for possible non-normality of the primary parameter distribution. In our implementation, we use an intercept-only specification with inverse-variance weights \(1/\sigma_i^2\) and the PMT finite-sample correction. We apply the procedure to \((\bar Y_i,\widehat{\mathrm{se}}(\bar Y_i))\) for Classical EBCI and to \((\hat\theta_i^{\mathrm{PT}},\widehat{\mathrm{se}}(\hat\theta_i^{\mathrm{PT}}))\) for PT EBCI.

We additionally provide a Python implementation of the same EBCI construction, including the moment estimators and robust critical-value calculations, for use in our Python-based experiments. In both this Python implementation and the supplementary wrapper around the R \texttt{ebci} implementation, tasks with exactly zero plug-in standard error are excluded from the inverse-variance-weighted cross-task moment fit, since $1/\sigma_i^2$ is undefined at $\sigma_i=0$. Their intervals are instead left as the corresponding unshrunk Wald intervals, which reduce to the point intervals $[\bar{Y}_i,\bar{Y}_i]$ for Classical EBCI and $[\hat\theta_i^{\mathrm{PT}},\hat\theta_i^{\mathrm{PT}}]$ for EBCI PT when $\sigma_i=0$.
\subsection{LMArena}
\label{app:experiment-lmarena}

\paragraph{Dataset construction.} We begin with the public dataset \url{https://huggingface.co/datasets/lmarena-ai/arena-human-preference-140k}, which contains 140k pairwise LLM comparison records, where each row consists of two LLMs' responses to a given prompt, and a human's decision of which response is better. Our data cleaning pipeline filters out samples that (i) involve a multi-turn conversation (i.e., the user asks a follow-up question), (ii) have a tied human decision, or (iii) could not be successfully parsed or scored by the reward-model pipeline.
Then we group the rows by the pair of LLMs involved (we take LLM A to be the one whose model name is alphabetically smaller), this gives us $n = 298$ LLM pairs that have at least 100 data points (i.e. $M_i + m_i \geq 100$) for our numerical experiments (covering a total of \textasciitilde 41k pairwise comparisons).

\paragraph{Reward model predictions.} For each comparison $X_{ij}$, the reward model \texttt{Skywork-Reward-V2} directly outputs two scalar scores $s^A_j$ and $s^B_j$ for the responses of LLM A and B, respectively. We then convert them into a probabilistic prediction via the Bradley-Terry model
$$
  \hat p_j \;=\; \frac{1}{1 + \exp(s^B_j - s^A_j)}.
$$
The biased estimator $\hat\theta_i^{\mathrm{ML}}$ is the average of the Bradley--Terry probabilities $\hat p_j$ over the unlabeled comparisons for model pair $i$. Across all comparisons, the MSE of these probabilities against the observed binary preferences is $0.337$.

\begin{table}[H]
\scriptsize
\setlength{\tabcolsep}{3pt}
\renewcommand{\arraystretch}{1.2}
\caption{LMArena results across $\alpha\in\{0.05,0.10,0.20,0.30\}$, averaged over $K=200$ random labeled/unlabeled splits with $n=298$ pairwise LLM problems. Each cell reports mean $\pm$ 1 Monte-Carlo SE. \emph{Mis-cov} denotes the miscoverage rate. \emph{Classical} is the interval without ML information; \emph{Pred Mean} is the prediction-only interval; \emph{PT} denotes the power-tuned PPI baseline; \emph{RB Normal} and \emph{RB NPMLE} are the rebiased PT estimators with the bias distribution estimated by normal and NPMLE respectively. \emph{EBCI Classical} and \emph{EBCI PT} apply the robust empirical Bayes intervals of \cite{armstrong2022robust} to the Classical and PT estimators, respectively. The width-ratio is computed taskwise relative to the Classical interval within each split, excluding task/split pairs with 0 classical width, and is then averaged across tasks within each split and across Monte-Carlo splits.}
\label{tbl:lmarena_ci}
\makebox[\textwidth][c]{

\begin{tabular}{ccccccccc}
\toprule
$\alpha$
&
& \textbf{Classical}
& \textbf{Pred Mean}
& \textbf{PT}
& \textbf{Classical EBCI}
& \textbf{PT EBCI}
& \textbf{RB Normal}
& \textbf{RB NPMLE} \\
\hline

\multirow{3}{*}{0.05}
& mis-cov (\%)
& $4.4{\scriptstyle\,\pm\,0.1}$
& $58.1{\scriptstyle\,\pm\,0.1}$
& $4.9{\scriptstyle\,\pm\,0.1}$
& $3.7{\scriptstyle\,\pm\,0.1}$
& $4.0{\scriptstyle\,\pm\,0.1}$
& $6.1{\scriptstyle\,\pm\,0.1}$
& $9.6{\scriptstyle\,\pm\,0.2}$ \\
& width
& $0.363{\scriptstyle\,\pm\,0.001}$
& $0.150{\scriptstyle\,\pm\,0.000}$
& $0.354{\scriptstyle\,\pm\,0.001}$
& $0.322{\scriptstyle\,\pm\,0.001}$
& $0.316{\scriptstyle\,\pm\,0.001}$
& $0.273{\scriptstyle\,\pm\,0.001}$
& $0.258{\scriptstyle\,\pm\,0.001}$ \\
& width-ratio
& $1.000{\scriptstyle\,\pm\,0.000}$
& $0.416{\scriptstyle\,\pm\,0.001}$
& $0.975{\scriptstyle\,\pm\,0.001}$
& $0.891{\scriptstyle\,\pm\,0.001}$
& $0.875{\scriptstyle\,\pm\,0.001}$
& $0.758{\scriptstyle\,\pm\,0.001}$
& $0.717{\scriptstyle\,\pm\,0.002}$ \\
\hline

\multirow{3}{*}{0.10}
& mis-cov (\%)
& $8.0{\scriptstyle\,\pm\,0.1}$
& $64.6{\scriptstyle\,\pm\,0.1}$
& $8.9{\scriptstyle\,\pm\,0.1}$
& $7.5{\scriptstyle\,\pm\,0.1}$
& $8.0{\scriptstyle\,\pm\,0.1}$
& $11.3{\scriptstyle\,\pm\,0.2}$
& $15.5{\scriptstyle\,\pm\,0.2}$ \\
& width
& $0.304{\scriptstyle\,\pm\,0.000}$
& $0.125{\scriptstyle\,\pm\,0.000}$
& $0.297{\scriptstyle\,\pm\,0.000}$
& $0.270{\scriptstyle\,\pm\,0.001}$
& $0.265{\scriptstyle\,\pm\,0.001}$
& $0.229{\scriptstyle\,\pm\,0.001}$
& $0.216{\scriptstyle\,\pm\,0.001}$ \\
& width-ratio
& $1.000{\scriptstyle\,\pm\,0.000}$
& $0.416{\scriptstyle\,\pm\,0.001}$
& $0.975{\scriptstyle\,\pm\,0.001}$
& $0.892{\scriptstyle\,\pm\,0.001}$
& $0.876{\scriptstyle\,\pm\,0.001}$
& $0.758{\scriptstyle\,\pm\,0.001}$
& $0.714{\scriptstyle\,\pm\,0.002}$ \\
\hline

\multirow{3}{*}{0.20}
& mis-cov (\%)
& $16.3{\scriptstyle\,\pm\,0.2}$
& $72.0{\scriptstyle\,\pm\,0.1}$
& $17.4{\scriptstyle\,\pm\,0.2}$
& $16.3{\scriptstyle\,\pm\,0.2}$
& $16.8{\scriptstyle\,\pm\,0.2}$
& $21.4{\scriptstyle\,\pm\,0.2}$
& $26.4{\scriptstyle\,\pm\,0.3}$ \\
& width
& $0.237{\scriptstyle\,\pm\,0.000}$
& $0.098{\scriptstyle\,\pm\,0.000}$
& $0.231{\scriptstyle\,\pm\,0.000}$
& $0.212{\scriptstyle\,\pm\,0.000}$
& $0.208{\scriptstyle\,\pm\,0.000}$
& $0.178{\scriptstyle\,\pm\,0.001}$
& $0.168{\scriptstyle\,\pm\,0.001}$ \\
& width-ratio
& $1.000{\scriptstyle\,\pm\,0.000}$
& $0.416{\scriptstyle\,\pm\,0.001}$
& $0.975{\scriptstyle\,\pm\,0.001}$
& $0.896{\scriptstyle\,\pm\,0.001}$
& $0.879{\scriptstyle\,\pm\,0.001}$
& $0.758{\scriptstyle\,\pm\,0.001}$
& $0.710{\scriptstyle\,\pm\,0.002}$ \\
\hline

\multirow{3}{*}{0.30}
& mis-cov (\%)
& $25.7{\scriptstyle\,\pm\,0.2}$
& $77.0{\scriptstyle\,\pm\,0.1}$
& $26.7{\scriptstyle\,\pm\,0.2}$
& $25.4{\scriptstyle\,\pm\,0.2}$
& $26.2{\scriptstyle\,\pm\,0.2}$
& $31.2{\scriptstyle\,\pm\,0.2}$
& $36.2{\scriptstyle\,\pm\,0.3}$ \\
& width
& $0.192{\scriptstyle\,\pm\,0.000}$
& $0.079{\scriptstyle\,\pm\,0.000}$
& $0.187{\scriptstyle\,\pm\,0.000}$
& $0.172{\scriptstyle\,\pm\,0.000}$
& $0.168{\scriptstyle\,\pm\,0.000}$
& $0.144{\scriptstyle\,\pm\,0.001}$
& $0.135{\scriptstyle\,\pm\,0.001}$ \\
& width-ratio
& $1.000{\scriptstyle\,\pm\,0.000}$
& $0.416{\scriptstyle\,\pm\,0.001}$
& $0.975{\scriptstyle\,\pm\,0.001}$
& $0.899{\scriptstyle\,\pm\,0.001}$
& $0.882{\scriptstyle\,\pm\,0.001}$
& $0.758{\scriptstyle\,\pm\,0.001}$
& $0.708{\scriptstyle\,\pm\,0.002}$ \\
\bottomrule
\end{tabular}

}
\end{table}

\subsection{Amazon Fine Food reviews}
\label{app:experiment-amazon}

\paragraph{Dataset construction.} We use the publicly available \emph{Amazon Fine Food Reviews} dataset~\citep{mcauley2013amateurs}, hosted by the Stanford Network Analysis Project (SNAP) on Kaggle.\footnote{\url{https://www.kaggle.com/datasets/snap/amazon-fine-food-reviews/data}}
Duplicate reviews, defined as reviews sharing the same \texttt{(title, text, rating, UserID)} quadruplet, can appear under different \texttt{ProductID}s. We therefore deduplicate in three steps: (i) we merge products whose reviews are all identical, since this usually means several \texttt{ProductID}s point to the same underlying product; (ii) if two products share only some of their reviews and the overlap is at most $10$, we remove the shared reviews from both products; and (iii) if the overlap exceeds $10$, we drop the product with fewer total reviews. To avoid data contamination when training our predictor, after selecting the held-out products (the 200 products with most reviews; see below), we further remove 18 pairs of reviews from the training and held-out splits whose tokenized model inputs overlap (i.e. reviews whose title plus text look the same up to capitalization or white spaces). After filtering, $67{,}619$ products with $393{,}818$ reviews remain.

Reviews are grouped by their \texttt{ProductID}. For each review we form the covariate $X_{ij}$ by concatenating the review title and body text, and we let the response $Y_{ij} \in \{1,2,3,4,5\}$ be the integer rating (higher is better) chosen by the $j$-th reviewer of this product. Following the experimental design of \citet{li2025predictionpowered}, we restrict attention to the $n = 200$ products with the most reviews (every product has at least $151$ reviews, totaling $52{,}706$ reviews). Selecting the top-reviewed products mitigates extreme heteroscedasticity across the per-task variances $\sigma_i^2$ that would otherwise dominate any cross-task comparison. For each product we randomly split its reviews into a labeled and unlabeled partition with a $20/80$ ratio, repeating the random split for $K = 200$ trials; the results reported in~\cref{fig:ppi-over-alphas} and~\cref{tbl:amazon_ci} are computed over these trials.

\begin{table}[H]
\scriptsize
\setlength{\tabcolsep}{3pt}
\renewcommand{\arraystretch}{1.2}
\caption{Amazon results across $\alpha\in\{0.05,0.10,0.20,0.30\}$, averaged over 200 random splits with $n=200$ tasks. width-ratio is normalized by the Classical interval. Entries are reported as mean $\pm$ Monte Carlo SE.}
\label{tbl:amazon_ci}
\makebox[\textwidth][c]{
\begin{tabular}{ccccccccc}
\toprule
$\alpha$
&
& \textbf{Classical}
& \textbf{Pred Mean}
& \textbf{PT}
& \textbf{EBCI Classical}
& \textbf{EBCI PT}
& \textbf{RB Normal}
& \textbf{RB NPMLE} \\
\hline

\multirow{3}{*}{0.05}
& mis-cov (\%)
& $4.8{\scriptstyle\,\pm\,0.1}$
& $6.4{\scriptstyle\,\pm\,0.1}$
& $3.2{\scriptstyle\,\pm\,0.1}$
& $6.3{\scriptstyle\,\pm\,0.1}$
& $3.5{\scriptstyle\,\pm\,0.1}$
& $2.0{\scriptstyle\,\pm\,0.1}$
& $2.0{\scriptstyle\,\pm\,0.1}$ \\
& width
& $0.656{\scriptstyle\,\pm\,0.001}$
& $0.342{\scriptstyle\,\pm\,0.000}$
& $0.412{\scriptstyle\,\pm\,0.001}$
& $0.609{\scriptstyle\,\pm\,0.001}$
& $0.400{\scriptstyle\,\pm\,0.001}$
& $0.343{\scriptstyle\,\pm\,0.001}$
& $0.341{\scriptstyle\,\pm\,0.001}$ \\
& width-ratio
& $1.000{\scriptstyle\,\pm\,0.000}$
& $0.537{\scriptstyle\,\pm\,0.001}$
& $0.646{\scriptstyle\,\pm\,0.001}$
& $0.940{\scriptstyle\,\pm\,0.001}$
& $0.630{\scriptstyle\,\pm\,0.001}$
& $0.538{\scriptstyle\,\pm\,0.001}$
& $0.533{\scriptstyle\,\pm\,0.001}$ \\
\hline

\multirow{3}{*}{0.10}
& mis-cov (\%)
& $8.4{\scriptstyle\,\pm\,0.1}$
& $11.3{\scriptstyle\,\pm\,0.1}$
& $5.7{\scriptstyle\,\pm\,0.1}$
& $10.6{\scriptstyle\,\pm\,0.1}$
& $6.4{\scriptstyle\,\pm\,0.1}$
& $3.6{\scriptstyle\,\pm\,0.1}$
& $3.7{\scriptstyle\,\pm\,0.1}$ \\
& width
& $0.551{\scriptstyle\,\pm\,0.001}$
& $0.287{\scriptstyle\,\pm\,0.000}$
& $0.346{\scriptstyle\,\pm\,0.001}$
& $0.510{\scriptstyle\,\pm\,0.001}$
& $0.336{\scriptstyle\,\pm\,0.001}$
& $0.288{\scriptstyle\,\pm\,0.001}$
& $0.286{\scriptstyle\,\pm\,0.001}$ \\
& width-ratio
& $1.000{\scriptstyle\,\pm\,0.000}$
& $0.537{\scriptstyle\,\pm\,0.001}$
& $0.646{\scriptstyle\,\pm\,0.001}$
& $0.938{\scriptstyle\,\pm\,0.001}$
& $0.630{\scriptstyle\,\pm\,0.001}$
& $0.538{\scriptstyle\,\pm\,0.001}$
& $0.532{\scriptstyle\,\pm\,0.001}$ \\
\hline

\multirow{3}{*}{0.20}
& mis-cov (\%)
& $16.7{\scriptstyle\,\pm\,0.2}$
& $21.0{\scriptstyle\,\pm\,0.1}$
& $12.1{\scriptstyle\,\pm\,0.2}$
& $19.1{\scriptstyle\,\pm\,0.2}$
& $13.6{\scriptstyle\,\pm\,0.2}$
& $8.1{\scriptstyle\,\pm\,0.1}$
& $8.3{\scriptstyle\,\pm\,0.1}$ \\
& width
& $0.429{\scriptstyle\,\pm\,0.001}$
& $0.224{\scriptstyle\,\pm\,0.000}$
& $0.269{\scriptstyle\,\pm\,0.001}$
& $0.398{\scriptstyle\,\pm\,0.001}$
& $0.262{\scriptstyle\,\pm\,0.001}$
& $0.224{\scriptstyle\,\pm\,0.001}$
& $0.222{\scriptstyle\,\pm\,0.001}$ \\
& width-ratio
& $1.000{\scriptstyle\,\pm\,0.000}$
& $0.537{\scriptstyle\,\pm\,0.001}$
& $0.646{\scriptstyle\,\pm\,0.001}$
& $0.940{\scriptstyle\,\pm\,0.001}$
& $0.630{\scriptstyle\,\pm\,0.001}$
& $0.538{\scriptstyle\,\pm\,0.001}$
& $0.531{\scriptstyle\,\pm\,0.001}$ \\
\hline

\multirow{3}{*}{0.30}
& mis-cov (\%)
& $25.9{\scriptstyle\,\pm\,0.2}$
& $30.5{\scriptstyle\,\pm\,0.2}$
& $19.8{\scriptstyle\,\pm\,0.2}$
& $28.6{\scriptstyle\,\pm\,0.2}$
& $22.1{\scriptstyle\,\pm\,0.2}$
& $14.2{\scriptstyle\,\pm\,0.2}$
& $14.7{\scriptstyle\,\pm\,0.2}$ \\
& width
& $0.347{\scriptstyle\,\pm\,0.001}$
& $0.181{\scriptstyle\,\pm\,0.000}$
& $0.218{\scriptstyle\,\pm\,0.001}$
& $0.323{\scriptstyle\,\pm\,0.001}$
& $0.212{\scriptstyle\,\pm\,0.001}$
& $0.181{\scriptstyle\,\pm\,0.001}$
& $0.179{\scriptstyle\,\pm\,0.001}$ \\
& width-ratio
& $1.000{\scriptstyle\,\pm\,0.000}$
& $0.537{\scriptstyle\,\pm\,0.001}$
& $0.646{\scriptstyle\,\pm\,0.001}$
& $0.941{\scriptstyle\,\pm\,0.001}$
& $0.630{\scriptstyle\,\pm\,0.001}$
& $0.538{\scriptstyle\,\pm\,0.001}$
& $0.530{\scriptstyle\,\pm\,0.001}$ \\
\bottomrule
\end{tabular}
}
\end{table}

\paragraph{Predictor.} The prediction model $h$ is a fine-tuned BERT~\citep{devlin2019bert} neural network. We start from the publicly available \texttt{bert-base-multilingual} checkpoint\footnote{\url{https://huggingface.co/nlptown/bert-base-multilingual-uncased-sentiment}}\!\!, which is pre-trained on general multilingual product reviews (not exclusive to Amazon) for the same 1--5 star prediction task, and we further fine-tune it for two full epochs on the reviews \emph{outside} the top-200 products, using the Hugging Face \texttt{transformers} library. Fine-tuning improves prediction accuracy on our test set (top-200 products) from $66.6\%$ (off-the-shelf checkpoint) to $76.6\%$ (fine-tuned after 2 epochs). We thus use this fine-tuned version for our model $h$. Fine-tuning uses 341,112 reviews from 67,419 products outside the held-out top 200 ones; none of the held-out review rows is used for fine-tuning. Crucially, the per-product samples used to compute $\biased_i$ and $\hat b_i$ in \eqref{eq:ppi_b_bhat_pt} are entirely disjoint from the corpus on which $h$ was fine-tuned.

\subsection{Synthetic study}
\label{app:experiment-synth}

\paragraph{Data generation process.} For each task $i\in\{1,\ldots,n\}$ with $n=200$, we draw a true primary parameter $\theta_{i}$ and a prediction bias $b_i$ from the distributions specified for each setting below. Synthetic summary-level observations $(\bar Y_i,\bar Z_i^h,\tilde Z_i^h)$ are then generated from an empirical covariance structure extracted from one task chosen at random from the Amazon dataset; from these we form the PT summary pair $(\hat\theta_i^{\mathrm{PT}},\,\hat b_{i,\hat\lambda_i})$
using the power-tuning constant as in the main text for each $i$. Each setting is repeated for $K=500$ Monte Carlo replications, and we report average coverage, average length, and average length-ratio.

\subsubsection{Sensitivity to the bias distribution}

In this setting, we fix the true primary parameter distribution at $\theta_{i}\sim\mathrm{N}(4,0.5^2)$ and vary the distribution of $b_i$ (normal or two-point) as described below. The pseudo-ground-truth parameters $\theta_{i}$ in the Amazon data have the sample standard deviation of approximately \(0.48\), motivating \(\sigma_0=0.5\) as a comparable reference scale.

\paragraph{Choice of $A$ for the normal bias distribution.} For
$b_i\sim\mathrm{N}(-0.1,A)$, we sweep $A^{1/2}\in\{0.01,0.03,0.05,0.1\}$. The grid is anchored on the parametric bias distribution fitted from the real Amazon split in Section~\ref{app:experiment-amazon}, whose marginal-MLE estimate has standard deviation $\approx 0.045$. The smallest value ($A^{1/2}=0.01$) corresponds to a regime in which the biases are tightly concentrated around their nonzero common mean of $-0.1$; therefore, pooling across tasks provides strong information about each task-specific bias, enabling stronger shrinkage and shorter rebiased intervals. The largest ($A^{1/2}=0.1$) is roughly $2\times$ wider than the empirical estimate and allows substantially greater variation in the biases, weakening the benefit of shrinkage and making
the rebiased procedure behave more like the fully debiased procedure. We also choose two intermediate values.

\paragraph{Per-$A$ numerical results.}
Table~\ref{tbl:synthetic_amazon} reports average coverage, average width, and average width-ratio for each estimator under the simulation with normal bias distribution, across the four values of $A^{1/2}$.

\begin{table}[H]
\footnotesize
\setlength{\tabcolsep}{4pt}
\renewcommand{\arraystretch}{1.2}
\caption{Simulation results under normal bias distribution with covariance structure adopted from the Amazon data ($n=200$, $K=500$ replications). Entries are reported as mean $\pm$ Monte Carlo SE.}
\label{tbl:synthetic_amazon}
\makebox[\textwidth][c]{
\begin{tabular}{cccccc}
\toprule
\textbf{$A^{1/2}$} &
& \textbf{Oracle}
& \textbf{Classical}
& \textbf{Pred Mean}
& \textbf{PT} \\
\hline

\multirow{3}{*}{0.01}
& mis-cov (\%)
& $5.0{\scriptstyle\,\pm\,0.1}$
& $5.0{\scriptstyle\,\pm\,0.1}$
& $26.7{\scriptstyle\,\pm\,0.1}$
& $4.9{\scriptstyle\,\pm\,0.1}$ \\
& width
& $0.321{\scriptstyle\,\pm\,0.001}$
& $0.660{\scriptstyle\,\pm\,0.001}$
& $0.341{\scriptstyle\,\pm\,0.001}$
& $0.419{\scriptstyle\,\pm\,0.001}$ \\
& width-ratio
& $0.496{\scriptstyle\,\pm\,0.001}$
& $1.000{\scriptstyle\,\pm\,0.000}$
& $0.531{\scriptstyle\,\pm\,0.001}$
& $0.653{\scriptstyle\,\pm\,0.001}$ \\
\hline

\multirow{3}{*}{0.03}
& mis-cov (\%)
& $5.0{\scriptstyle\,\pm\,0.1}$
& $5.0{\scriptstyle\,\pm\,0.1}$
& $27.2{\scriptstyle\,\pm\,0.1}$
& $4.9{\scriptstyle\,\pm\,0.1}$ \\
& width
& $0.331{\scriptstyle\,\pm\,0.001}$
& $0.660{\scriptstyle\,\pm\,0.001}$
& $0.341{\scriptstyle\,\pm\,0.001}$
& $0.419{\scriptstyle\,\pm\,0.001}$ \\
& width-ratio
& $0.513{\scriptstyle\,\pm\,0.001}$
& $1.000{\scriptstyle\,\pm\,0.000}$
& $0.531{\scriptstyle\,\pm\,0.001}$
& $0.653{\scriptstyle\,\pm\,0.001}$ \\
\hline

\multirow{3}{*}{0.05}
& mis-cov (\%)
& $5.0{\scriptstyle\,\pm\,0.1}$
& $5.0{\scriptstyle\,\pm\,0.1}$
& $28.6{\scriptstyle\,\pm\,0.1}$
& $4.9{\scriptstyle\,\pm\,0.1}$ \\
& width
& $0.344{\scriptstyle\,\pm\,0.001}$
& $0.660{\scriptstyle\,\pm\,0.001}$
& $0.341{\scriptstyle\,\pm\,0.001}$
& $0.419{\scriptstyle\,\pm\,0.001}$ \\
& width-ratio
& $0.537{\scriptstyle\,\pm\,0.001}$
& $1.000{\scriptstyle\,\pm\,0.000}$
& $0.531{\scriptstyle\,\pm\,0.001}$
& $0.653{\scriptstyle\,\pm\,0.001}$ \\
\hline

\multirow{3}{*}{0.10}
& mis-cov (\%)
& $4.9{\scriptstyle\,\pm\,0.1}$
& $5.0{\scriptstyle\,\pm\,0.1}$
& $34.6{\scriptstyle\,\pm\,0.1}$
& $4.9{\scriptstyle\,\pm\,0.1}$ \\
& width
& $0.372{\scriptstyle\,\pm\,0.001}$
& $0.660{\scriptstyle\,\pm\,0.001}$
& $0.341{\scriptstyle\,\pm\,0.001}$
& $0.419{\scriptstyle\,\pm\,0.001}$ \\
& width-ratio
& $0.584{\scriptstyle\,\pm\,0.001}$
& $1.000{\scriptstyle\,\pm\,0.000}$
& $0.531{\scriptstyle\,\pm\,0.001}$
& $0.653{\scriptstyle\,\pm\,0.001}$ \\
\bottomrule
\end{tabular}
}
\end{table}

\vspace{-1.4em}

\begin{table}[H]
\ContinuedFloat
\footnotesize
\setlength{\tabcolsep}{4pt}
\renewcommand{\arraystretch}{1.2}
\makebox[\textwidth][c]{
\begin{tabular}{cccccc}
\toprule
\textbf{$A^{1/2}$} &
& \textbf{EBCI Classical}
& \textbf{EBCI PT}
& \textbf{RB Normal}
& \textbf{RB NPMLE} \\
\hline

\multirow{3}{*}{0.01}
& mis-cov (\%)
& $5.1{\scriptstyle\,\pm\,0.1}$
& $5.0{\scriptstyle\,\pm\,0.1}$
& $5.0{\scriptstyle\,\pm\,0.1}$
& $4.9{\scriptstyle\,\pm\,0.1}$ \\
& width
& $0.614{\scriptstyle\,\pm\,0.001}$
& $0.407{\scriptstyle\,\pm\,0.001}$
& $0.323{\scriptstyle\,\pm\,0.001}$
& $0.325{\scriptstyle\,\pm\,0.001}$ \\
& width-ratio
& $0.942{\scriptstyle\,\pm\,0.001}$
& $0.637{\scriptstyle\,\pm\,0.001}$
& $0.499{\scriptstyle\,\pm\,0.001}$
& $0.501{\scriptstyle\,\pm\,0.001}$ \\
\hline

\multirow{3}{*}{0.03}
& mis-cov (\%)
& $5.1{\scriptstyle\,\pm\,0.1}$
& $5.0{\scriptstyle\,\pm\,0.1}$
& $5.2{\scriptstyle\,\pm\,0.1}$
& $5.2{\scriptstyle\,\pm\,0.1}$ \\
& width
& $0.614{\scriptstyle\,\pm\,0.001}$
& $0.407{\scriptstyle\,\pm\,0.001}$
& $0.330{\scriptstyle\,\pm\,0.001}$
& $0.331{\scriptstyle\,\pm\,0.001}$ \\
& width-ratio
& $0.942{\scriptstyle\,\pm\,0.001}$
& $0.637{\scriptstyle\,\pm\,0.001}$
& $0.512{\scriptstyle\,\pm\,0.001}$
& $0.513{\scriptstyle\,\pm\,0.001}$ \\
\hline

\multirow{3}{*}{0.05}
& mis-cov (\%)
& $5.1{\scriptstyle\,\pm\,0.1}$
& $5.0{\scriptstyle\,\pm\,0.1}$
& $5.2{\scriptstyle\,\pm\,0.1}$
& $5.3{\scriptstyle\,\pm\,0.1}$ \\
& width
& $0.614{\scriptstyle\,\pm\,0.001}$
& $0.407{\scriptstyle\,\pm\,0.001}$
& $0.343{\scriptstyle\,\pm\,0.001}$
& $0.343{\scriptstyle\,\pm\,0.001}$ \\
& width-ratio
& $0.942{\scriptstyle\,\pm\,0.001}$
& $0.637{\scriptstyle\,\pm\,0.001}$
& $0.535{\scriptstyle\,\pm\,0.001}$
& $0.534{\scriptstyle\,\pm\,0.001}$ \\
\hline

\multirow{3}{*}{0.10}
& mis-cov (\%)
& $5.1{\scriptstyle\,\pm\,0.1}$
& $5.0{\scriptstyle\,\pm\,0.1}$
& $5.0{\scriptstyle\,\pm\,0.1}$
& $5.5{\scriptstyle\,\pm\,0.1}$ \\
& width
& $0.614{\scriptstyle\,\pm\,0.001}$
& $0.407{\scriptstyle\,\pm\,0.001}$
& $0.372{\scriptstyle\,\pm\,0.001}$
& $0.369{\scriptstyle\,\pm\,0.001}$ \\
& width-ratio
& $0.942{\scriptstyle\,\pm\,0.001}$
& $0.637{\scriptstyle\,\pm\,0.001}$
& $0.583{\scriptstyle\,\pm\,0.001}$
& $0.576{\scriptstyle\,\pm\,0.001}$ \\
\bottomrule
\end{tabular}
}
\end{table}

\paragraph{Choice of $b_0$ for the two-point bias distribution.} For the misspecification case, we draw $b_i$ from a discrete two-point distribution \(b_i\sim \frac{1}{2}\delta_0+\frac{1}{2}\delta_{b_0}\), and we sweep \(b_0\in\{0.05,0.1,0.2,0.5\}\). This grid controls the degree to which the bias distribution departs from the normal working model.  When $b_0=0.05$, the support points are close enough that the induced distribution of the observed bias estimates can be difficult to distinguish from that under a tightly concentrated normal working prior centered near $b_0/2$. As $b_0$ increases, the support-point separation and prior variance increase, making the discrete two-point structure more pronounced and the normal working model more substantially misspecified. The two intermediate values trace the transition from a nearly normal-approximable regime to a strongly non-normal regime.

\paragraph{Per-$b_0$ numerical results.}
Table~\ref{tbl:synthetic_amazon_twopoint} reports average coverage, average width, and average width-ratio for each estimator under the two-point bias distribution simulation, across the four values of $b_0$.

\begin{table}[H]
\footnotesize
\setlength{\tabcolsep}{4pt}
\renewcommand{\arraystretch}{1.2}
\caption{Simulation results under two-point bias distribution with covariance structure adopted from the Amazon data ($n=200$, $K=500$ replications). Entries are reported as mean $\pm$ Monte Carlo SE.}
\label{tbl:synthetic_amazon_twopoint}
\centering

\begin{tabular}{cccccc}
\toprule
\textbf{$b_0$}
&
& \textbf{Oracle}
& \textbf{Classical}
& \textbf{Pred Mean}
& \textbf{PT} \\
\midrule

\multirow{3}{*}{0.05}
& mis-cov (\%)
& $4.9{\scriptstyle\,\pm\,0.1}$
& $4.9{\scriptstyle\,\pm\,0.1}$
& $7.8{\scriptstyle\,\pm\,0.1}$
& $4.9{\scriptstyle\,\pm\,0.1}$ \\
& width
& $0.328{\scriptstyle\,\pm\,0.001}$
& $0.660{\scriptstyle\,\pm\,0.001}$
& $0.341{\scriptstyle\,\pm\,0.001}$
& $0.419{\scriptstyle\,\pm\,0.001}$ \\
& width-ratio
& $0.508{\scriptstyle\,\pm\,0.001}$
& $1.000{\scriptstyle\,\pm\,0.000}$
& $0.531{\scriptstyle\,\pm\,0.001}$
& $0.653{\scriptstyle\,\pm\,0.001}$ \\
\midrule

\multirow{3}{*}{0.10}
& mis-cov (\%)
& $4.9{\scriptstyle\,\pm\,0.1}$
& $4.9{\scriptstyle\,\pm\,0.1}$
& $15.9{\scriptstyle\,\pm\,0.1}$
& $4.9{\scriptstyle\,\pm\,0.1}$ \\
& width
& $0.342{\scriptstyle\,\pm\,0.001}$
& $0.660{\scriptstyle\,\pm\,0.001}$
& $0.341{\scriptstyle\,\pm\,0.001}$
& $0.419{\scriptstyle\,\pm\,0.001}$ \\
& width-ratio
& $0.533{\scriptstyle\,\pm\,0.001}$
& $1.000{\scriptstyle\,\pm\,0.000}$
& $0.531{\scriptstyle\,\pm\,0.001}$
& $0.653{\scriptstyle\,\pm\,0.001}$ \\
\midrule

\multirow{3}{*}{0.20}
& mis-cov (\%)
& $4.9{\scriptstyle\,\pm\,0.1}$
& $4.9{\scriptstyle\,\pm\,0.1}$
& $35.3{\scriptstyle\,\pm\,0.2}$
& $4.9{\scriptstyle\,\pm\,0.1}$ \\
& width
& $0.360{\scriptstyle\,\pm\,0.001}$
& $0.660{\scriptstyle\,\pm\,0.001}$
& $0.341{\scriptstyle\,\pm\,0.001}$
& $0.419{\scriptstyle\,\pm\,0.001}$ \\
& width-ratio
& $0.557{\scriptstyle\,\pm\,0.001}$
& $1.000{\scriptstyle\,\pm\,0.000}$
& $0.531{\scriptstyle\,\pm\,0.001}$
& $0.653{\scriptstyle\,\pm\,0.001}$ \\
\midrule

\multirow{3}{*}{0.50}
& mis-cov (\%)
& $4.9{\scriptstyle\,\pm\,0.1}$
& $4.9{\scriptstyle\,\pm\,0.1}$
& $52.1{\scriptstyle\,\pm\,0.2}$
& $4.9{\scriptstyle\,\pm\,0.1}$ \\
& width
& $0.335{\scriptstyle\,\pm\,0.001}$
& $0.660{\scriptstyle\,\pm\,0.001}$
& $0.341{\scriptstyle\,\pm\,0.001}$
& $0.419{\scriptstyle\,\pm\,0.001}$ \\
& width-ratio
& $0.512{\scriptstyle\,\pm\,0.001}$
& $1.000{\scriptstyle\,\pm\,0.000}$
& $0.531{\scriptstyle\,\pm\,0.001}$
& $0.653{\scriptstyle\,\pm\,0.001}$ \\
\bottomrule
\end{tabular}

\end{table}

\vspace{-1.4em}

\begin{table}[H]
\ContinuedFloat
\footnotesize
\setlength{\tabcolsep}{4pt}
\renewcommand{\arraystretch}{1.2}
\centering

\begin{tabular}{cccccc}
\toprule
\textbf{$b_0$}
&
& \textbf{EBCI Classical}
& \textbf{EBCI PT}
& \textbf{RB Normal}
& \textbf{RB NPMLE} \\
\midrule

\multirow{3}{*}{0.05}
& mis-cov (\%)
& $5.1{\scriptstyle\,\pm\,0.1}$
& $4.9{\scriptstyle\,\pm\,0.1}$
& $5.1{\scriptstyle\,\pm\,0.1}$
& $5.0{\scriptstyle\,\pm\,0.1}$ \\
& width
& $0.614{\scriptstyle\,\pm\,0.001}$
& $0.407{\scriptstyle\,\pm\,0.001}$
& $0.327{\scriptstyle\,\pm\,0.001}$
& $0.329{\scriptstyle\,\pm\,0.001}$ \\
& width-ratio
& $0.942{\scriptstyle\,\pm\,0.001}$
& $0.637{\scriptstyle\,\pm\,0.001}$
& $0.507{\scriptstyle\,\pm\,0.001}$
& $0.510{\scriptstyle\,\pm\,0.001}$ \\
\midrule

\multirow{3}{*}{0.10}
& mis-cov (\%)
& $5.1{\scriptstyle\,\pm\,0.1}$
& $4.9{\scriptstyle\,\pm\,0.1}$
& $5.0{\scriptstyle\,\pm\,0.1}$
& $5.1{\scriptstyle\,\pm\,0.1}$ \\
& width
& $0.614{\scriptstyle\,\pm\,0.001}$
& $0.407{\scriptstyle\,\pm\,0.001}$
& $0.343{\scriptstyle\,\pm\,0.001}$
& $0.343{\scriptstyle\,\pm\,0.001}$ \\
& width-ratio
& $0.942{\scriptstyle\,\pm\,0.001}$
& $0.637{\scriptstyle\,\pm\,0.001}$
& $0.536{\scriptstyle\,\pm\,0.001}$
& $0.535{\scriptstyle\,\pm\,0.001}$ \\
\midrule

\multirow{3}{*}{0.20}
& mis-cov (\%)
& $5.1{\scriptstyle\,\pm\,0.1}$
& $4.9{\scriptstyle\,\pm\,0.1}$
& $4.9{\scriptstyle\,\pm\,0.1}$
& $4.9{\scriptstyle\,\pm\,0.1}$ \\
& width
& $0.614{\scriptstyle\,\pm\,0.001}$
& $0.407{\scriptstyle\,\pm\,0.001}$
& $0.372{\scriptstyle\,\pm\,0.001}$
& $0.362{\scriptstyle\,\pm\,0.001}$ \\
& width-ratio
& $0.942{\scriptstyle\,\pm\,0.001}$
& $0.637{\scriptstyle\,\pm\,0.001}$
& $0.583{\scriptstyle\,\pm\,0.001}$
& $0.562{\scriptstyle\,\pm\,0.001}$ \\
\midrule

\multirow{3}{*}{0.50}
& mis-cov (\%)
& $5.1{\scriptstyle\,\pm\,0.1}$
& $4.9{\scriptstyle\,\pm\,0.1}$
& $4.8{\scriptstyle\,\pm\,0.1}$
& $4.7{\scriptstyle\,\pm\,0.1}$ \\
& width
& $0.614{\scriptstyle\,\pm\,0.001}$
& $0.407{\scriptstyle\,\pm\,0.001}$
& $0.405{\scriptstyle\,\pm\,0.001}$
& $0.340{\scriptstyle\,\pm\,0.001}$ \\
& width-ratio
& $0.942{\scriptstyle\,\pm\,0.001}$
& $0.637{\scriptstyle\,\pm\,0.001}$
& $0.633{\scriptstyle\,\pm\,0.001}$
& $0.521{\scriptstyle\,\pm\,0.001}$ \\
\bottomrule
\end{tabular}

\end{table}

\subsubsection{Sensitivity to primary parameter dispersion}
In this setting, we fix the bias distribution at \(b_i\sim \mathrm{N}(-0.1,0.01^2)\) and vary the dispersion of the primary parameter distribution by taking \(\theta_{i}\sim \mathrm{N}(4,\sigma_0^2)\) with \(\sigma_0\in\{0.01,0.1,0.5,1\}\). The range of \(\sigma_0\) spans primary parameter distributions from highly concentrated to substantially dispersed across tasks. When \(\sigma_0\) is small, EBCI can benefit substantially from borrowing information across the \(\theta_i\)'s, while this advantage is expected to diminish as \(\sigma_0\) increases. Table~\ref{tbl:synthetic_amazon_vary_theta} reports average coverage, average width, and average width-ratio for each estimator under different choices of $\sigma_0$.

\begin{figure}[H]
    \centering
    \includegraphics[width=0.6\linewidth]{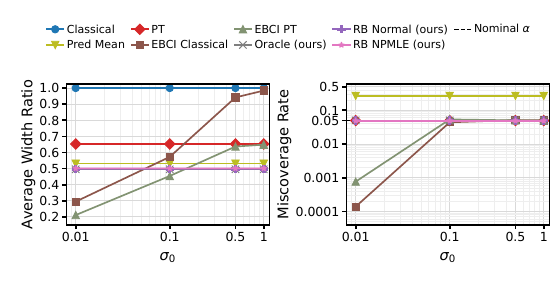}
    \caption{\textbf{Interval width ratio and miscoverage rate in the simulation test for primary parameter dispersion.} Average width ratio (left) and miscoverage rate (right) for $n=200$ simulated tasks in the synthetic Amazon dataset with different levels of primary parameter dispersion. In Table~\ref{tbl:synthetic_amazon_vary_theta} we include the standard errors for the metrics.}
    \label{fig:synthetic_theta}
\end{figure}

\begin{table}[H]
\footnotesize
\setlength{\tabcolsep}{4pt}
\renewcommand{\arraystretch}{1.2}
\caption{Simulation results under $\theta_i\sim \mathrm{N}(4,\sigma_0^2)$ with covariance structure adopted from the Amazon data ($n=200$, $K=500$ replications). Entries are reported as mean $\pm$ Monte Carlo SE.}
\label{tbl:synthetic_amazon_vary_theta}
\makebox[\textwidth][c]{
\begin{tabular}{cccccc}
\toprule
\textbf{$\sigma_0$} &
& \textbf{Oracle}
& \textbf{Classical}
& \textbf{Pred Mean}
& \textbf{PT} \\
\hline

\multirow{3}{*}{0.01}
& mis-cov (\%)
& $5.0{\scriptstyle\,\pm\,0.1}$
& $5.0{\scriptstyle\,\pm\,0.1}$
& $26.7{\scriptstyle\,\pm\,0.1}$
& $4.9{\scriptstyle\,\pm\,0.1}$ \\
& width
& $0.321{\scriptstyle\,\pm\,0.001}$
& $0.660{\scriptstyle\,\pm\,0.001}$
& $0.341{\scriptstyle\,\pm\,0.001}$
& $0.419{\scriptstyle\,\pm\,0.001}$ \\
& width-ratio
& $0.496{\scriptstyle\,\pm\,0.001}$
& $1.000{\scriptstyle\,\pm\,0.000}$
& $0.531{\scriptstyle\,\pm\,0.001}$
& $0.653{\scriptstyle\,\pm\,0.001}$ \\
\hline

\multirow{3}{*}{0.10}
& mis-cov (\%)
& $5.0{\scriptstyle\,\pm\,0.1}$
& $5.0{\scriptstyle\,\pm\,0.1}$
& $26.7{\scriptstyle\,\pm\,0.1}$
& $4.9{\scriptstyle\,\pm\,0.1}$ \\
& width
& $0.321{\scriptstyle\,\pm\,0.001}$
& $0.660{\scriptstyle\,\pm\,0.001}$
& $0.341{\scriptstyle\,\pm\,0.001}$
& $0.419{\scriptstyle\,\pm\,0.001}$ \\
& width-ratio
& $0.496{\scriptstyle\,\pm\,0.001}$
& $1.000{\scriptstyle\,\pm\,0.000}$
& $0.531{\scriptstyle\,\pm\,0.001}$
& $0.653{\scriptstyle\,\pm\,0.001}$ \\
\hline

\multirow{3}{*}{0.50}
& mis-cov (\%)
& $5.0{\scriptstyle\,\pm\,0.1}$
& $5.0{\scriptstyle\,\pm\,0.1}$
& $26.7{\scriptstyle\,\pm\,0.1}$
& $4.9{\scriptstyle\,\pm\,0.1}$ \\
& width
& $0.321{\scriptstyle\,\pm\,0.001}$
& $0.660{\scriptstyle\,\pm\,0.001}$
& $0.341{\scriptstyle\,\pm\,0.001}$
& $0.419{\scriptstyle\,\pm\,0.001}$ \\
& width-ratio
& $0.496{\scriptstyle\,\pm\,0.001}$
& $1.000{\scriptstyle\,\pm\,0.000}$
& $0.531{\scriptstyle\,\pm\,0.001}$
& $0.653{\scriptstyle\,\pm\,0.001}$ \\
\hline

\multirow{3}{*}{1.00}
& mis-cov (\%)
& $5.0{\scriptstyle\,\pm\,0.1}$
& $5.0{\scriptstyle\,\pm\,0.1}$
& $26.7{\scriptstyle\,\pm\,0.1}$
& $4.9{\scriptstyle\,\pm\,0.1}$ \\
& width
& $0.321{\scriptstyle\,\pm\,0.001}$
& $0.660{\scriptstyle\,\pm\,0.001}$
& $0.341{\scriptstyle\,\pm\,0.001}$
& $0.419{\scriptstyle\,\pm\,0.001}$ \\
& width-ratio
& $0.496{\scriptstyle\,\pm\,0.001}$
& $1.000{\scriptstyle\,\pm\,0.000}$
& $0.531{\scriptstyle\,\pm\,0.001}$
& $0.653{\scriptstyle\,\pm\,0.001}$ \\
\bottomrule
\end{tabular}
}
\end{table}

\vspace{-1.4em}

\begin{table}[H]
\ContinuedFloat
\footnotesize
\setlength{\tabcolsep}{4pt}
\renewcommand{\arraystretch}{1.2}
\makebox[\textwidth][c]{
\begin{tabular}{cccccc}
\toprule
\textbf{$\sigma_0$} &
& \textbf{EBCI Classical}
& \textbf{EBCI PT}
& \textbf{RB Normal}
& \textbf{RB NPMLE} \\
\hline

\multirow{3}{*}{0.01}
& mis-cov (\%)
& $0.0{\scriptstyle\,\pm\,0.0}$
& $0.1{\scriptstyle\,\pm\,0.1}$
& $5.0{\scriptstyle\,\pm\,0.1}$
& $4.9{\scriptstyle\,\pm\,0.1}$ \\
& width
& $0.175{\scriptstyle\,\pm\,0.004}$
& $0.125{\scriptstyle\,\pm\,0.003}$
& $0.323{\scriptstyle\,\pm\,0.001}$
& $0.325{\scriptstyle\,\pm\,0.001}$ \\
& width-ratio
& $0.295{\scriptstyle\,\pm\,0.006}$
& $0.212{\scriptstyle\,\pm\,0.004}$
& $0.499{\scriptstyle\,\pm\,0.001}$
& $0.501{\scriptstyle\,\pm\,0.001}$ \\
\hline

\multirow{3}{*}{0.10}
& mis-cov (\%)
& $4.4{\scriptstyle\,\pm\,0.1}$
& $5.4{\scriptstyle\,\pm\,0.1}$
& $5.0{\scriptstyle\,\pm\,0.1}$
& $4.9{\scriptstyle\,\pm\,0.1}$ \\
& width
& $0.352{\scriptstyle\,\pm\,0.002}$
& $0.278{\scriptstyle\,\pm\,0.001}$
& $0.323{\scriptstyle\,\pm\,0.001}$
& $0.325{\scriptstyle\,\pm\,0.001}$ \\
& width-ratio
& $0.573{\scriptstyle\,\pm\,0.003}$
& $0.454{\scriptstyle\,\pm\,0.001}$
& $0.499{\scriptstyle\,\pm\,0.001}$
& $0.501{\scriptstyle\,\pm\,0.001}$ \\
\hline

\multirow{3}{*}{0.50}
& mis-cov (\%)
& $5.1{\scriptstyle\,\pm\,0.1}$
& $5.0{\scriptstyle\,\pm\,0.1}$
& $5.0{\scriptstyle\,\pm\,0.1}$
& $4.9{\scriptstyle\,\pm\,0.1}$ \\
& width
& $0.614{\scriptstyle\,\pm\,0.001}$
& $0.407{\scriptstyle\,\pm\,0.001}$
& $0.323{\scriptstyle\,\pm\,0.001}$
& $0.325{\scriptstyle\,\pm\,0.001}$ \\
& width-ratio
& $0.942{\scriptstyle\,\pm\,0.001}$
& $0.637{\scriptstyle\,\pm\,0.001}$
& $0.499{\scriptstyle\,\pm\,0.001}$
& $0.501{\scriptstyle\,\pm\,0.001}$ \\
\hline

\multirow{3}{*}{1.00}
& mis-cov (\%)
& $5.0{\scriptstyle\,\pm\,0.1}$
& $5.0{\scriptstyle\,\pm\,0.1}$
& $5.0{\scriptstyle\,\pm\,0.1}$
& $4.9{\scriptstyle\,\pm\,0.1}$ \\
& width
& $0.647{\scriptstyle\,\pm\,0.001}$
& $0.416{\scriptstyle\,\pm\,0.001}$
& $0.323{\scriptstyle\,\pm\,0.001}$
& $0.325{\scriptstyle\,\pm\,0.001}$ \\
& width-ratio
& $0.984{\scriptstyle\,\pm\,0.001}$
& $0.649{\scriptstyle\,\pm\,0.001}$
& $0.499{\scriptstyle\,\pm\,0.001}$
& $0.501{\scriptstyle\,\pm\,0.001}$ \\
\bottomrule
\end{tabular}
}
\end{table}

\subsection{Family-based GWAS}
\label{app:experiment-gwas}

\paragraph{Datasets.} The height family-based GWAS summary statistics from \citet{Guan2025-jq} are publicly available at \url{https://thessgac.com}. The estimates were obtained by running family-based SNP-wise regressions on $44{,}570$ ``white British'' individuals in the UK Biobank \citep{Bycroft2018-zp}, controlling for 40 genetic principal components and other covariates. Sibling (close to our target direct effect) estimate summary statistics from \citet{Howe2022-fm} are accessible from OpenGWAS \citep{opengwas} through the \texttt{ieugwasr} R package \citep{ieugwasr} with OpenGWAS ID ieu-b-4813. 1000 Genomes Phase 3 EUR reference panel \citep{1000_Genomes_Project_Consortium2015-vn} can be obtained from \url{http://fileserve.mrcieu.ac.uk/ld/1kg.v3.tgz}.

\paragraph{Details on overlap and LD-matching analysis.}
We compare each method's height GWAS discoveries with putatively significant SNPs from the sibling-based direct-effect analysis of \citet{Howe2022-fm}, defined using the threshold $\text{p-value} < 1\times 10^{-6}$. We first restrict our discoveries and Howe et al.'s putatively significant SNPs to variants present in the 1000 Genomes European reference panel \cite{1000_Genomes_Project_Consortium2015-vn}, which was used to estimate linkage disequilibrium (LD). We then count a discovery as matching Howe et al. if it is either identical to one of their putatively significant SNPs or is in LD with one, defined as being within 250 kb and having $r^2 \ge 0.8$. This LD-based comparison provides a stringent operational definition for treating nearby correlated variants as evidence for the same underlying association signal. We used \texttt{plink2} \citep{PLINK} to compute LD and retain variant pairs within 250 kb with $r^2 \ge 0.8$.

\paragraph{Overlap analysis on LD clumped discoveries.}
One possible concern is that the larger number of overlaps from our rebiasing procedure, compared with the population-effect baseline, could simply reflect the discovery of more correlated SNPs. To investigate this further, we apply LD clumping \citep{PurcellPLINK} to each method's BH discoveries, selecting representative SNPs in order of increasing method-specific unadjusted $p$-value and excluding nearby correlated SNPs (within 250~kb and $r^2>0.1$). We repeat the overlap analysis, with results summarized in Table~\ref{tab:discovery-overlap-LD-clumped}. The observed patterns are similar to those shown in Table~\ref{tab:discovery-overlap}, with our RB-NPMLE-based procedure showing more overlaps than the biased population-effect baseline but with fewer total LD-clumped discoveries.

\begin{figure}
  \begin{minipage}{0.58\textwidth}
    \includegraphics[width=\linewidth]{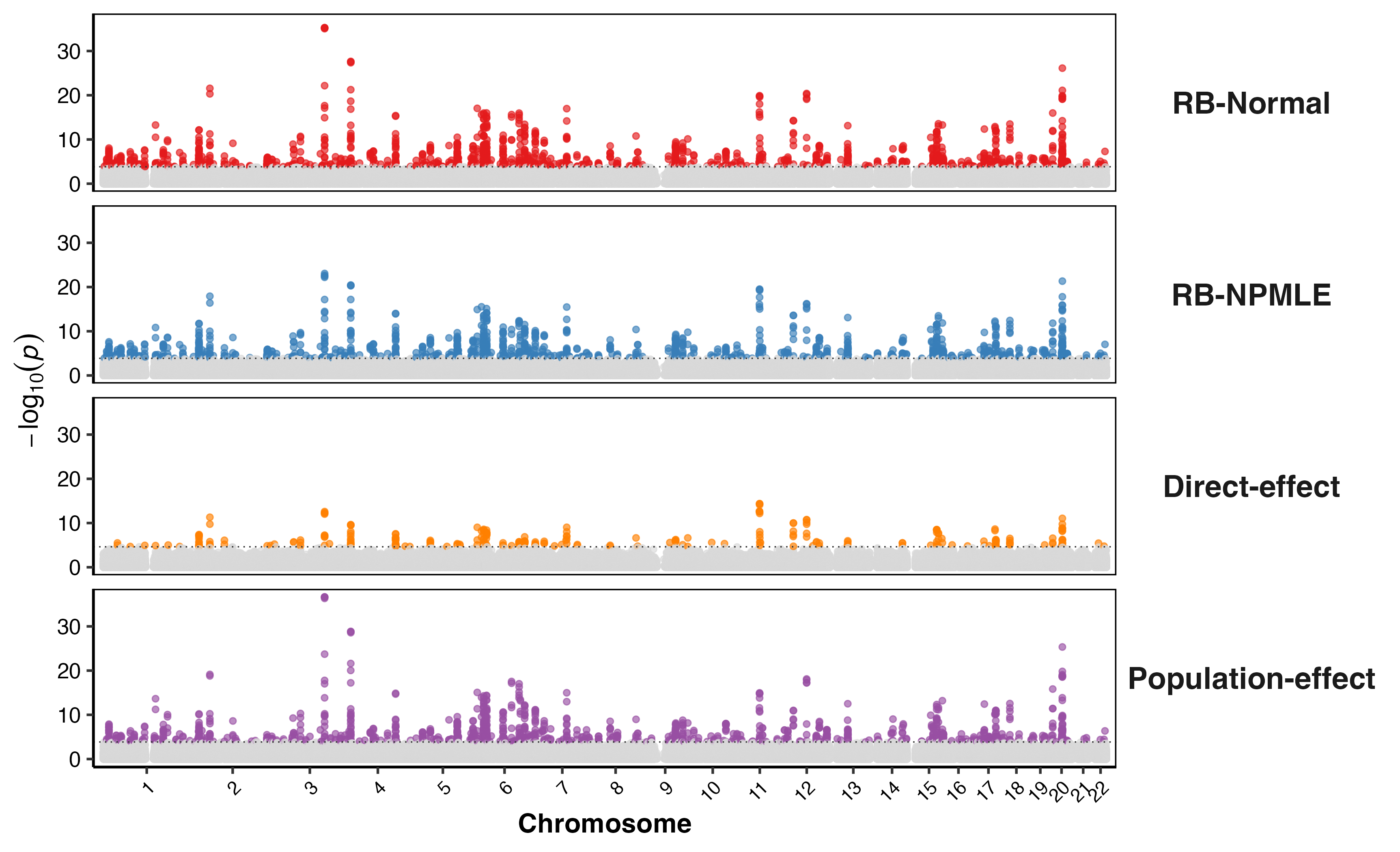}
  \end{minipage}%
  \hfill 
  \begin{minipage}{0.4\textwidth}
    \vspace{-2em}
    \caption{\textbf{SNP-level p-values in family-based GWAS.} Manhattan plots of SNP-level p-values, shown on the $-\log_{10}$ scale. Horizontal dashed lines indicate the corresponding B-H thresholds. SNPs passing the threshold are highlighted as significant. For visualization, values exceeding $-\log_{10}(10^{-50})=50$ are truncated at 50.}
    \label{fig:gwas_manhattan}
  \end{minipage}
\end{figure}

\begin{table}
  \centering
  \caption{Summary of LD-clumped SNP discoveries from the height family-based GWAS summary statistics. We report the numbers of LD-clumped discoveries from BH at FDR $0.05$ applied to our \textit{RB-NPMLE} and \textit{RB-Normal} p-values, the (unbiased) direct-effect
p-values, and the (biased) population-effect p-values, together with their
overlaps with \citet{Howe2022-fm} signals.
  }
  \label{tab:discovery-overlap-LD-clumped}
  \begin{tabular}{lrrrrr}
  \toprule
   & \multicolumn{1}{c}{\textbf{RB-NPMLE}}
   & \multicolumn{1}{c}{\textbf{RB-Normal}}
   & \multicolumn{1}{c}{\textbf{Direct-effect}}
   & \multicolumn{1}{c}{\textbf{Population-effect}}
    \\
   & \multicolumn{1}{c}{\textbf{(ours)}}
   & \multicolumn{1}{c}{\textbf{(ours)}}
   & \multicolumn{1}{c}{\textbf{(unbiased)}}
   & \multicolumn{1}{c}{\textbf{(biased)}}
    \\
  \midrule
  \textbf{\# discoveries} & 514 & 546 & 105 & 533 \\
  \textbf{\# overlaps}    & 179 & 180 &  71 & 159 \\
  \bottomrule
\end{tabular}
\end{table}

\paragraph{Comments on small bias assumptions for height.} 
We expect the true bias $b_i$ for height to be small for the following reasons: (1) the top 40 genetic principal components are included as controls in the regressions, so much of the confounding bias due to population stratification should be accounted for \citep{Price2006-ng}; and (2) previous analyses \citep{Yengo2018-jb, Kong2018-ld, Young2022-tz} have shown that parental indirect genetic effects are weak for height. Although confounding due to assortative mating may still be present, it is expected to be small.

\subsection{Simulation study under realistic linkage disequilibrium} \label{app:experiment-ld-simulation}

We used sim1000G \citep{Dimitromanolakis2019-zt} and phased 1000 Genomes Phase 3 EUR haplotypes to generate parental haplotypes, from which we simulated Mendelian transmission to 5,000 two-sibling families (10,000 analyzed offspring). We analyzed 6,000 SNPs across 173 approximately independent LDetect regions \citep{Berisa2016-dl} on chromosomes 18–22, selecting a contiguous window of 34–35 SNPs centered within each region. Parental genotype sums were imputed using sibling genotypes and IBD states, and association testing used a sibship mixed model \citep{Young2022-tz}. We simulated phenotypes based on the following null model, with zero direct genetic effect ($\theta_i=0$):
    
$$
Y = \sum_{i=1}^{6000} b_i g_{\mathrm{par}, i}+\varepsilon, \ b_i \sim \mathrm{N}(0, \sigma^2_b), \ \varepsilon \sim \mathrm{N}(0, \sigma^2_\varepsilon).
$$
To represent a modest but non-negligible parental indirect genetic effect, we rescaled the parental effects $b_j$ (the bias term), so that the parental component explained 0.5\% and 1\% of phenotypic variance. These choices are motivated by small parental polygenic score (PRS) contributions to height variation~\citep{Ghatan2025-ww, Tan2025-ht}\footnote{Using the reported parental PRS coefficients and correlations from the height analyses, we calculate combined parental component variances of approximately 0.02--0.20\% for \citet{Ghatan2025-ww} and 0.26--0.32\% for \citet{Tan2025-ht}. \citet{Ghatan2025-ww} studied childhood height, so we use their estimates as a rough guide for our simulations.}.

We generated 400 phenotype replicates on the same genotype panel and applied BH at $q=0.05$. Under this global null, every rejection is a false discovery. Therefore, if $R$ denotes the number of rejections, $\mathrm{FDR}=\mathbb P(R\ge 1)$ (the family-wise error rate), which we estimated as the proportion of simulation replicates with at least one rejection. Empirical FDRs for our methods in both parental variance regimes are shown in Table~\ref{tab:gwas-fdr-ld-sim}.

\begin{table}
  \centering
  \caption{Empirical FDR comparison across methods in simulation study with realistic LD. Entries report empirical FDR with exact 95\% binomial confidence intervals, based on 400 Monte Carlo replicates.}
  \label{tab:gwas-fdr-ld-sim}
  \begin{tabular}{l c c}
    \toprule
    \textbf{Method} & \textbf{0.5\% parental variance} & \textbf{1\% parental variance} \\
    \midrule
    RB-NPMLE
      & $0.0300{\scriptstyle\,[0.0156,\,0.0518]}$
      & $0.0425{\scriptstyle\,[0.0249,\,0.0672]}$ \\
    RB-Normal
      & $0.0600{\scriptstyle\,[0.0388,\,0.0880]}$
      & $0.0800{\scriptstyle\,[0.0554,\,0.1111]}$ \\
    Direct-effect
      & $0.0300{\scriptstyle\,[0.0156,\,0.0518]}$
      & $0.0300{\scriptstyle\,[0.0156,\,0.0518]}$ \\
    Population-effect
      & $0.0550{\scriptstyle\,[0.0348,\,0.0821]}$
      & $0.1100{\scriptstyle\,[0.0811,\,0.1448]}$ \\
    \bottomrule
  \end{tabular}
\end{table}

These results suggest that RB-NPMLE remains well calibrated under realistic local LD and parental-effect sizes relevant to traits such as height. Although the bias terms were drawn from a single normal distribution, LD causes the single-SNP estimates to reflect marginal effects that aggregate signals from nearby SNPs. Variation in LD patterns across SNPs can induce heterogeneity in marginal effect variances, so their pooled distribution may depart from a single normal, which may explain the mild inflation observed for RB-Normal. However, for traits with stronger parental effects, LD may compromise calibration even for RB-NPMLE, which we acknowledge as a potential limitation of our methods when applied to family-based GWAS data. 

\subsection{Data and model licenses}
\label{app:asset-licenses}

We summarize the licenses and use conditions stated by the providers
of the datasets and models used in our LMArena, Amazon and GWAS experiments.

\paragraph{LMArena data.}
The \texttt{arena-human-preference-140k} dataset card states that
user prompts are licensed under Creative Commons Attribution 4.0
(CC BY 4.0), while model-generated responses are subject to the
respective model providers' terms.\footnote{
\url{https://huggingface.co/datasets/lmarena-ai/arena-human-preference-140k}}
Our filtering and preprocessing are described in~\cref{app:experiment-lmarena}.

\paragraph{Skywork reward model.}
We use the checkpoint
\texttt{Skywork/Skywork-Reward-V2-Llama-3.1-8B},
whose model card specifies the Llama 3.1 Community License.\footnote{
\url{https://huggingface.co/Skywork/Skywork-Reward-V2-Llama-3.1-8B}}

\paragraph{Amazon Fine Food Reviews data.}
We use the Amazon Fine Food Reviews corpus associated with
\citet{mcauley2013amateurs}, downloaded from SNAP.\footnote{
\url{https://snap.stanford.edu/data/web-FineFoods.html}}
The Kaggle distribution lists CC0 1.0 Universal; the original SNAP
download page does not separately state an explicit dataset
license.\footnote{
\url{https://www.kaggle.com/datasets/snap/amazon-fine-food-reviews}}
Our deduplication and filtering procedures are described in~\cref{app:experiment-amazon}.

\paragraph{BERT model.}
We initialize our predictor from the multilingual BERT sentiment model released by NLP Town\footnote{
\url{https://huggingface.co/nlptown/bert-base-multilingual-uncased-sentiment}}, whose model card declares the MIT License, and further fine-tune it as described in~\cref{app:experiment-amazon}.

\paragraph{GWAS data.} The GWAS inputs used in our experiments are described in
Appendix~\ref{app:experiment-gwas}. SSGAC requires each user to accept its terms and prohibits redistribution.\footnote{\url{https://thessgac.com/register/}} The 1000 Genomes reference data are publicly available and should be cited.\footnote{\url{https://www.internationalgenome.org/IGSR_disclaimer/}} OpenGWAS datasets remain subject to the original providers' terms.\footnote{\url{https://opengwas.io/policies/eula}} We release a SNP-selection file for reproducibility, but no GWAS summary statistics or reference genotypes.

\section{Diagnostics for the theoretical assumptions}
\label{app:robustness_diagnostic}
As a complement to the robustness results in Appendix~\ref{app:robustness}, we assess the main approximation assumptions separately for each empirical application, including the plausibility of the normal approximation and plug-in second moments, as well as the extent of cross-task dependence. We also provide qualitative diagnostics for the exchangeability assumptions.

We split our discussion into two parts. Appendix~\ref{app:robustness_diagnostic1} considers assessments and diagnostics that can be carried out using the data available in practice. Appendix~\ref{app:robustness_diagnostic2} provides additional diagnostics for the partially Bayes assumption in two PPI applications, using also the labels that we artificially masked in our evaluation.

\subsection{Assessment of assumptions in practice}
\label{app:robustness_diagnostic1}

The marginal fits in Figure~\ref{fig:bias_fit} provide a basic diagnostic for the fitted bias model: a substantial discrepancy between the observed and implied distributions of $\smash{\hat{b}_i}$ would indicate marginal model misspecification. The close agreement in Figure~\ref{fig:bias_fit} provides no evidence of such a failure for our three applications.

Since the latent biases \(b_i\) are generally unobserved, we examine whether the estimated biases $\smash{\hat{b}_i}$ exhibit systematic variation with task-specific second-moment quantities  $q_i\in\{\rho_i,\tau_i^2,\sigma_i^2\}$. Specifically, for any such $q_i$, since \smash{$\hat{b}_i$} is unbiased for $b_i$, we have that $\smash{\mathbb E[\hat b_i\mid q_i]=\mathbb E[b_i\mid q_i]}$. Thus, systematic variation in the conditional mean of $\smash{\hat{b}_i}$ with $q_i$ can indicate departures from exchangeability. Figure~\ref{fig:diagnostic_nuisance} shows the relationship between $\smash{\hat{b}_i}$ and the task-specific covariance quantities $q_i \in \{\rho_i,\tau_i^2,\sigma_i^2\}$ across the three applications. In practice, the $q_i$ used in Figure~\ref{fig:diagnostic_nuisance} are plug-in estimates computed from the same data as $\hat b_i$ and are treated as fixed. Since these estimates may share estimation noise with $\hat b_i$, we only view Figure~\ref{fig:diagnostic_nuisance} as a qualitative diagnostic.

As a possible extension, one could also examine second-moment departures from exchangeability using $\VarInline{b_i\mid q_i}=\VarInline{\hat b_i\mid q_i}-\EEInline{\tau_i^2\mid q_i}$, and assess whether the conditional variance of $b_i$ varies systematically with $q_i$. We leave this and potential extensions to higher-order moments for future work.

\begin{figure}[t]
    \centering
    \includegraphics[width=\linewidth]{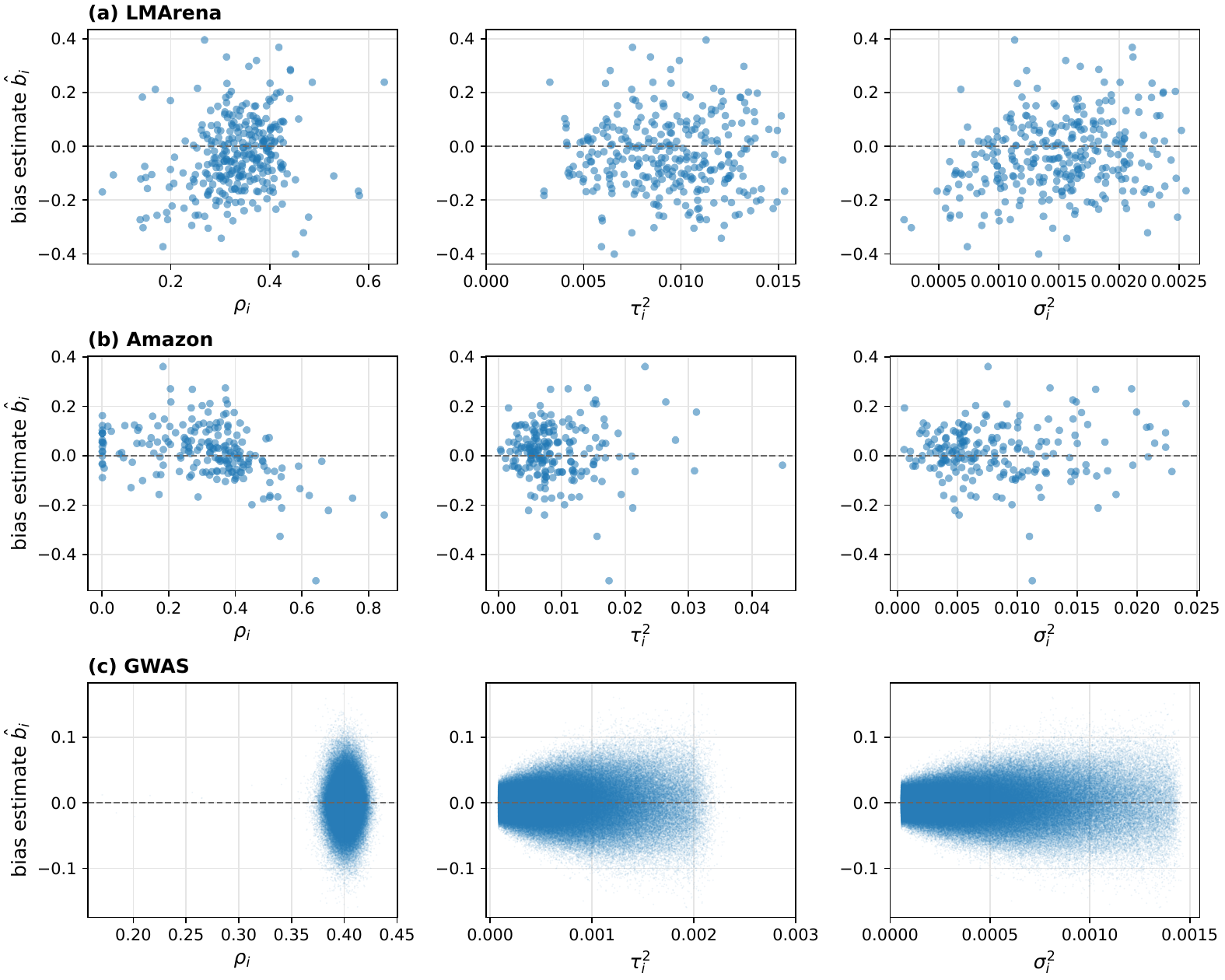}
    \caption{\textbf{Exchangeability diagnostic: Scatter plot of bias estimates vs. nuisance parameters in our three applications.}}
    \label{fig:diagnostic_nuisance}
\end{figure}

\paragraph{LMArena.}  For a representative random labeled/unlabeled (20/80) split, the 298 model-pair tasks contain 20--55 labeled observations and 80--221 unlabeled observations. Given these small sample sizes and binary outcomes $Y_i$, the normal approximation may be poor, and the plug-in second moments may also be imprecisely estimated. Each comparison belongs to a single LLM-pair task, so substantial cross-task dependence due to overlapping observations is less of a concern, although some dependence may still arise when different model-pair tasks share a common model. Fig.~\ref{fig:diagnostic_nuisance}a reveals some systematic variation in the estimated biases with the nuisance quantities, suggesting possible departures from exchangeability.

\paragraph{Amazon reviews.} For a representative random labeled/unlabeled (20/80) split of the Amazon data, each task contains 30--182 labeled observations and 121--730 unlabeled observations. The normal approximation may be more plausible here than for LM Arena, given the larger labeled sample sizes and the use of five-point ratings for $Y_i$ rather than binary outcomes. Since the product-level tasks are constructed from separate review sets, substantial cross-task dependence is less of a concern. In Fig. ~\ref{fig:diagnostic_nuisance}b, estimated biases $\smash{\hat{b}_i}$ exhibit only weak associations with the nuisance quantities.  Overall, these plots provide no clear evidence of systematic departures from exchangeability with respect to the task-specific covariance quantities.

\paragraph{Family-based GWAS.} GWAS summary statistics are obtained from regressions with a sample size of 44,570, so normal approximation and plug-in second moments are expected to be reasonably accurate. Cross-task independence is not satisfied in this example due to LD. Fig.~\ref{fig:diagnostic_nuisance}c shows no clear systematic relationship between the estimated biases and the nuisance quantities. 

\subsection{Diagnostics for the partially Bayes assumption based on pseudo-ground-truth quantities}
\label{app:robustness_diagnostic2}

Here we seek to check whether biases $b_i$ are exchangeable conditional on the primary parameters $\theta_i$. This is required for the conclusion under our partially Bayes framing (e.g., Theorem~\ref{thm:rate_coverage}), but not for our coverage results in the hierarchical setting and compound settings (Theorems~\ref{thm:rate_coverage_joint} and~\ref{thm:rate_average_coverage}).

For the two PPI applications (LM Arena and Amazon Reviews), we additionally have access to full-corpus observations, allowing us to construct pseudo-ground-truth for the primary parameters and biases. For each task $i$, we randomly partition the full corpus using a 50/50 split, with the two resulting subsets denoted by $\mathcal{I}_i^{\theta}$ and $\mathcal{I}_i^{b}$. We define the pseudo-primary parameter and pseudo-prediction bias as
$$
\tilde{\theta}_i := \frac{1}{|\mathcal{I}_i^{\theta}|}\sum_{j\in\mathcal{I}_i^{\theta}} Y_{ij},\qquad \tilde{b}_i :=\frac{1}{|\mathcal{I}_i^{b}|}\sum_{j\in\mathcal{I}_i^{b}}(h(X_{ij})-Y_{ij}).
$$
This 50/50 split-sample construction prevents the same outcome noise from entering $\tilde{\theta}_i$ and $\tilde{b}_i$ with opposite signs, which could otherwise induce a mechanical association between them. Using disjoint samples avoids the dependence that would arise from reusing the same outcomes. We then plot $\tilde b_i$ against $\tilde{\theta}_i$ across tasks to assess whether prediction bias varies systematically with the primary parameter. If $\theta_i$ and $b_i$ are independent, no systematic pattern should appear between their split-sample proxies.

\begin{figure}[t]
    \centering
    \includegraphics[width=0.8\linewidth]{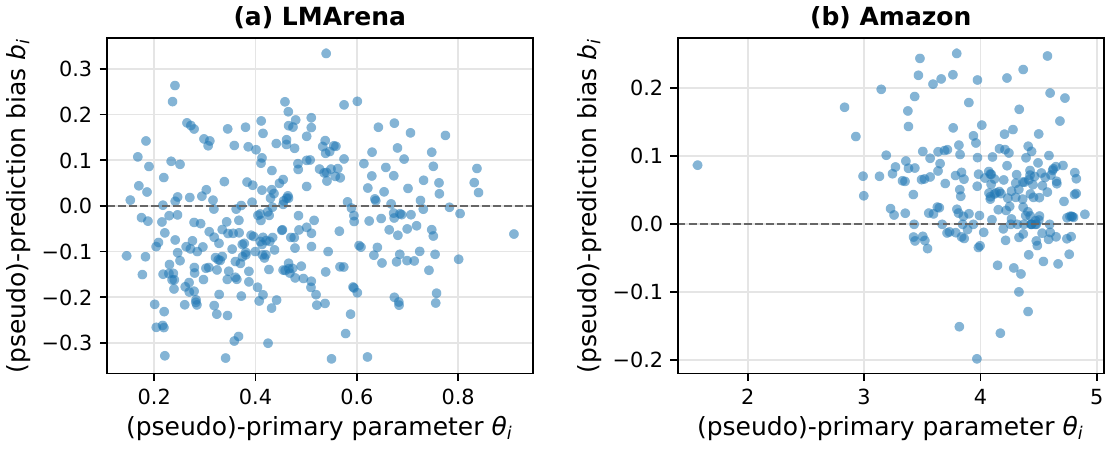}
    \caption{\textbf{Diagnostic for the partially Bayes assumption: Scatter plots of (pseudo)-prediction biases vs.\ (pseudo)-ground-truth primary parameters in our two PPI applications.}}
    \label{fig:diagnostic_pseudo}
\end{figure}

\paragraph{LMArena.} Figure~\ref{fig:diagnostic_pseudo}a shows substantial variation in the pseudo-ground-truth biases, but no clear systematic relationship with the primary parameter $\theta_i$.

\paragraph{Amazon reviews.} Figure~\ref{fig:diagnostic_pseudo}b shows a mild association between the pseudo-biases $b_i$ and the pseudo-parameters $\theta_i$,  although the observations are substantially dispersed. This diagnostic suggests that the independence assumption underlying the partially Bayes formulation, $\theta_i \perp b_i$ with $b_i \sim G$ and $\theta_i$ fixed, may hold only approximately in this application (see Section~\ref{sec:setting}). However, this pattern does not directly contradict our more general theoretical results, which allow dependence between the primary parameters and biases in the hierarchical setting and impose no independence assumption between the fixed $(\theta_i,b_i)$ sequences in the compound setting (see Section~\ref{sec:theory}).


\end{document}